\documentclass[11pt]{article}
\usepackage{graphicx} 
\usepackage{amsmath, amssymb}
\usepackage{changepage}
\usepackage{amsthm}
\usepackage[hidelinks]{hyperref}
\usepackage{subcaption}
\usepackage{tikz}
\usepackage{parskip}
\usepackage[
  backend=biber,
  style=authoryear,
  maxcitenames=2,
  mincitenames=1,
  maxbibnames=99,
  giveninits=true,
  uniquename=false,
  doi=false,
  url=false,
  isbn=false,
  eprint=false
]{biblatex}
\usepackage{tikz}
\usetikzlibrary{calc,arrows.meta}
\usepackage{algorithm}
\usepackage{algorithmic}
\usepackage{pgfplots}
\usepackage{setspace}
\usepackage{float}
\usepackage{booktabs}
\usepackage{bbm}
\usepackage{makecell}
\usepackage{multirow}
\usepackage{enumitem}
\usepackage{titletoc}

\DeclareNameAlias{author}{family-given}
\DeclareNameAlias{editor}{family-given}
\DeclareNameAlias{translator}{family-given}

\renewbibmacro{in:}{}

\newtheorem{theorem}{Theorem}[section]
\newtheorem{proposition}[theorem]{Proposition}
\preto\proposition{\vspace{1em}} 
\apptocmd\endproposition{\vspace{1em}}{}{} 

\newtheorem{assumption}[theorem]{Assumption}
\preto\assumption{\vspace{1em}} 
\apptocmd\endassumption{\vspace{1em}}{}{} 

\newtheorem{remark}[theorem]{Remark}
\preto\remark{\vspace{1em}}
\apptocmd\endremark{\vspace{1em}}{}{}

\newtheorem{lemma}[theorem]{Lemma}
\preto\lemma{\vspace{1em}}
\apptocmd\endlemma{\vspace{1em}}{}{}

\newtheorem{example}[theorem]{Example}
\preto\example{\vspace{1em}}
\apptocmd\endexample{\vspace{1em}}{}{}

\newtheorem{corollary}[theorem]{Corollary}
\preto\corollary{\vspace{1em}}
\apptocmd\endcorollary{\vspace{1em}}{}{}

\newtheorem{definition}[theorem]{Definition}
\preto\definition{\vspace{1em}}
\apptocmd\enddefinition{\vspace{1em}}{}{}

\preto\theorem{\vspace{1em}}
\apptocmd\endtheorem{\vspace{1em}}{}{}

\newcommand{\equalcontrib}{\textsuperscript{*}}

\newcommand{\indep}{\perp \!\!\! \perp}

\usepackage[margin=1in]{geometry}

\title{Estimating the Wasserstein barycenter of one-dimensional distributions under sparse sampling}

\date{\today}

\author{%
\begin{tabular}{cc}
\begin{minipage}[t]{0.43\textwidth}
\centering
James Peng\equalcontrib\\[0.25em]
\small Department of Biostatistics\\
\small University of Washington\\
\small Seattle, WA, USA\\
\small jpspeng@uw.edu
\end{minipage}
&
\begin{minipage}[t]{0.43\textwidth}
\centering
Florian Stijven\equalcontrib\\[0.25em]
\small I-BioStat\\
\small KU Leuven\\
\small Leuven, Belgium\\
\small florian.stijven@kuleuven.be
\end{minipage}
\\[6em]
\begin{minipage}[t]{0.43\textwidth}
\centering
Linbo Wang\\[0.25em]
\small Department of Statistical Sciences\\
\small University of Toronto\\
\small Toronto, Canada\\
\small linbo.wang@utoronto.ca
\end{minipage}
&
\begin{minipage}[t]{0.43\textwidth}
\centering
Peter B.~Gilbert\\[0.25em]
\small Vaccine and Infectious Disease Division\\
\small Fred Hutchinson Cancer Center\\
\small Seattle, WA, USA \\
\small pgilbert@fredhutch.org
\end{minipage}
\end{tabular}
}

\begin{document}
\maketitle

\begin{center}
\textsuperscript{*}These authors contributed equally to this work.
\end{center}

\vspace{2em}

\begin{abstract}
We study distributional data under sparse sampling where each unit is represented by a probability distribution on the real line that is observed only through a small i.i.d.~sample. A natural notion of central tendency for one-dimensional distributional data is the Wasserstein barycenter, a distribution whose quantile function is the pointwise average of the unit-level quantile functions. We focus on pointwise estimation of the quantile function of the Wasserstein barycenter: at a given quantile level, the target is the population mean of the corresponding unit-level quantiles. A naïve plug-in estimator is the empirical Wasserstein barycenter, which treats the observed unit-level empirical distributions as if they were the true latent unit-level distributions. Under sparse sampling, however, this estimator can be severely biased. We propose an approach that avoids directly estimating either the unit-level distributions or the full population law of distributions. We start with the more ambitious goal of characterizing the distribution of the latent unit-level quantiles at a given quantile level. We show that this distribution can be written in terms of the marginal distributions of the unit-level CDF values, which can be estimated using binomial mixture methods. 
This motivates our estimator, the marginal-constructed barycenter (MCB) estimator, which is obtained by taking the mean of the estimated distribution of latent unit-level quantiles. We establish conditions under which the MCB estimator is pointwise consistent and asymptotically normal, and we show through simulations that it can substantially outperform the empirical Wasserstein barycenter under sparse sampling. We illustrate the method in an analysis of HIV-1 sequence data from the HVTN 502/503 vaccine efficacy trials, using the barycenter to summarize and compare within-participant distributions of viral sequence features when only a small number of sequences are available per participant.
\end{abstract}

\doublespacing

\newpage

\section{Introduction}

We study the statistical analysis of distributional data, where the units of analysis are distributions on the real line. Our primary motivation comes from HIV-1 studies, where the biological object of interest is not a single homogeneous sample but a heterogeneous population of viral variants within an individual with HIV-1. These studies often summarize each sample by a single consensus sequence or a small number of dominant variants \parencite{edlefsen2015comprehensive, decamp2017sieve, juraska2024prevention}. With modern sequencing technologies, however, it is possible to obtain multiple viral sequences from the same individual, giving a more detailed view of the within-host viral population \parencite{gregori_viral_2016}. HIV-1 exhibits extensive genetic diversity within an individual, forming a quasispecies of continually evolving viral variants. When a one-dimensional numerical feature, such as a genetic distance or physicochemical property, is computed for each sequence, the resulting collection of values can be viewed as a sample from an individual-specific distribution. This distribution captures the heterogeneity of the within-host viral population.

A central difficulty in this setting is that, despite these technological advances, the number of sequences observed for each individual can be small. In some HIV-1 sequencing studies, there may be fewer than 10 sequences per individual \parencite{mullins_long-read_2025}. Thus, the scientific target is distributional, but each individual distribution might be observed only through a \textit{sparse sample}. This differs from better-studied examples of distributional data, such as wearable device studies, where accelerometers measure physical activity intensity at high frequency over multiple days. In that setting, the dense stream of measurements can often be regarded as repeated draws from an underlying subject-specific distribution, where the distribution characterizes the relative amount of time the individual spends at each intensity level \parencite{augustin_modelling_2017, ghosal_distributional_2023}. The number of accelerometer readings per individual typically is large, so each empirical distribution may be a reliable representation of the true distribution. In contrast, in the HIV-1 setting, the empirical distribution for individuals with low numbers of sequences may be a poor representation of their true individual distribution. More broadly, the challenge of small within-unit sample sizes is not unique
to sequencing studies: it can arise from resource
constraints or from applications in which each unit naturally contains only
a potentially small number of subunits. Additional data examples are discussed in
Appendix~\ref{appendix:additional-data-applications}. In this paper, we
focus on this sparse sampling setting, where each unit-level distribution is
observed through only a small number of samples.

A key task in the analysis of distributional data is to summarize a collection of distributions by a single representative distribution. In the HIV-1 application, such a summary could describe the average pattern of within-host viral heterogeneity across a population or within a clinically relevant subgroup. This notion of an average depends on the geometry of the space in which the objects are compared. In this paper, we focus on the Wasserstein barycenter, defined as the Fréchet mean of a collection of distributions under the Wasserstein metric \parencite{agueh2011barycenters}. Compared with pointwise Euclidean averaging of the distributions, the Wasserstein barycenter better preserves salient features of the component distributions \parencite{bigot_upper_2018, panaretos_invitation_2020}. For example, the barycenter of bimodal distributions of the same shape is itself bimodal, whereas the pointwise mean of the densities does not preserve that shape (Figure \ref{fig:barycenter_vs_mean}). Moreover, for one-dimensional distributions, the barycenter has a simple and interpretable form: its quantile function is the mean of the unit-level quantile functions \parencite{bigot_geodesic_2017}. Since the quantile function fully determines the distribution, we focus on estimating the quantile function of the barycenter.

\begin{figure}[!htb]
    \centering
    \includegraphics[width=\textwidth]{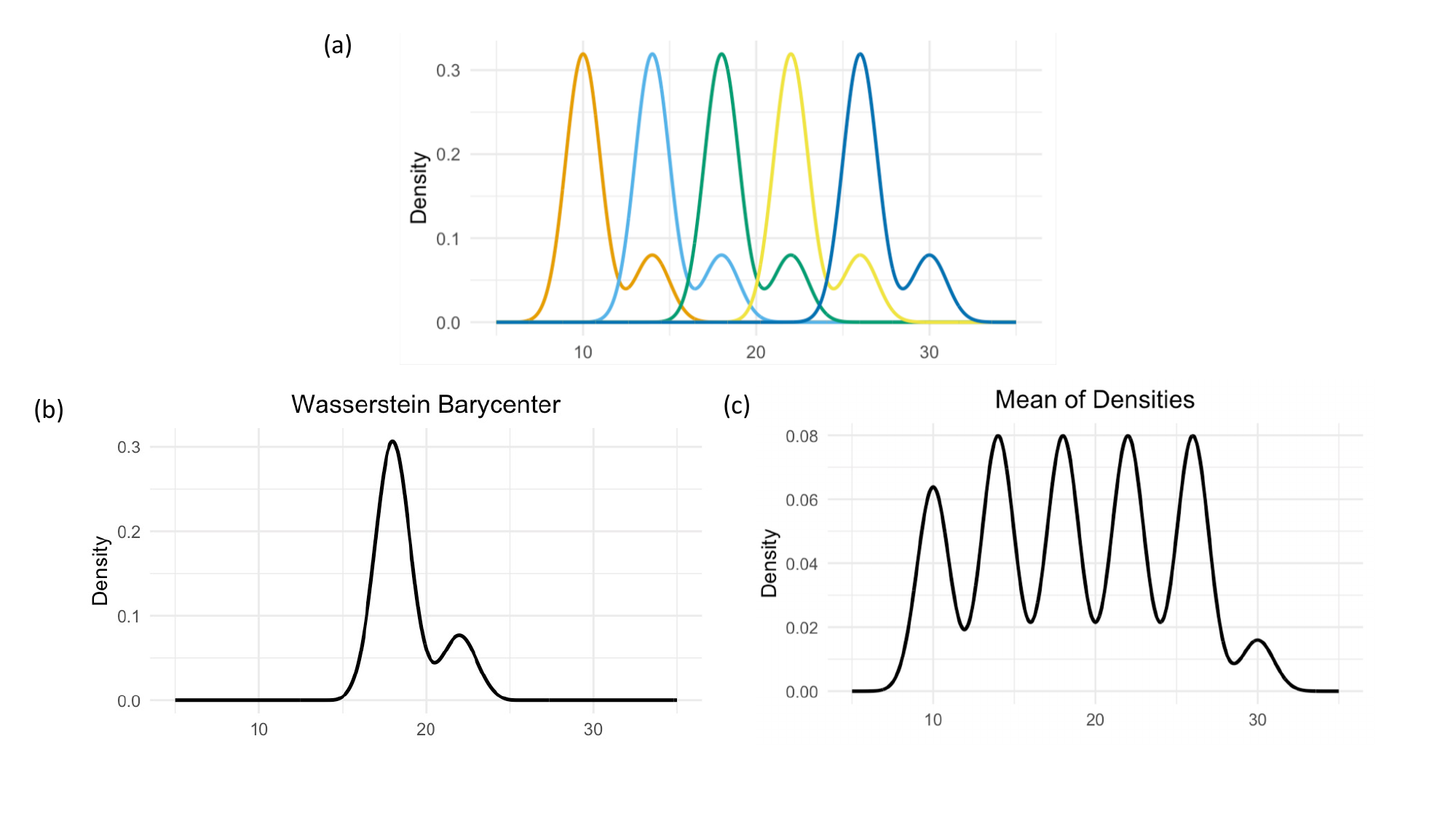}
    \caption{(a) The density of five bimodal distributions, (b) the Wasserstein barycenter of the five distributions, (c) the distribution obtained by taking the pointwise mean of the five density functions.}
    \label{fig:barycenter_vs_mean}
\end{figure}

A simple approach to estimating the barycenter is to treat the observed empirical distributions as if they were the true unit-level distributions and then compute the corresponding average quantile function. This yields the \textit{empirical Wasserstein barycenter}. Under sparse sampling, however, this estimator can be substantially biased. This is because the barycenter is calculated by averaging unit-level quantile functions, and empirical quantiles, unlike the empirical CDF, are generally biased when observed through only a small sample \parencite{bobkov2019one}. In particular, they are well known to generally be downward biased at upper tail quantiles, and upward biased at lower tail quantiles. 
Figure \ref{fig:method_preview} illustrates the problem in a simulation where 1,000  distributions are drawn from the 5 distributions in Figure \ref{fig:barycenter_vs_mean}a and only 7 i.i.d.~samples are observed from each. The empirical barycenter deviates substantially from the truth at most quantiles, and particularly in the tail quantile levels. Nevertheless, there is hope in that even though the unit-level quantile functions are poorly measured, the observed data still contain information across many independent units. This cross-unit information could be leveraged to recover the barycenter.

Motivated by this problem, we study barycenter estimation under sparse sampling. We make five main contributions. 
First, we show that when the number of observations per unit is bounded, the barycenter is generally not identifiable from the observed data distribution without additional structure. 
Second, we derive a representation result, which we call the
\textit{marginal construction}. The key idea is to begin with the seemingly more ambitious task of characterizing the distribution of the unobserved unit-level quantiles, from which the barycenter quantile is recovered as the mean. We show that the distribution of unit-level quantiles can be expressed in terms of the marginal distributions of unit-level CDF values. The latter distributions arise as mixing distributions in binomial mixture models. Because the binomial mixture
problem is well studied, this representation suggests how structural assumptions on these mixing distributions can restore identifiability of the barycenter. 
Third, motivated by this representation, we propose the \textit{marginal-constructed barycenter} (MCB) estimator, which estimates the barycenter's quantile function by first estimating the relevant binomial mixing distributions and then plugging those estimates into the representation formula. 
Fourth, under suitable structural conditions, we show that the MCB estimator's quantile function is pointwise consistent and asymptotically normal. We then complement these theoretical results with simulations showing that the MCB estimator can outperform the empirical barycenter even when these conditions are not satisfied; Figure \ref{fig:method_preview} previews its performance in the simulation described above. 
Fifth, we show that our framework extends to estimation of several other related distributional estimands, such as weighted barycenters and Fréchet variance.

\begin{figure}[!htb]
    \centering
    \includegraphics[width=\textwidth]{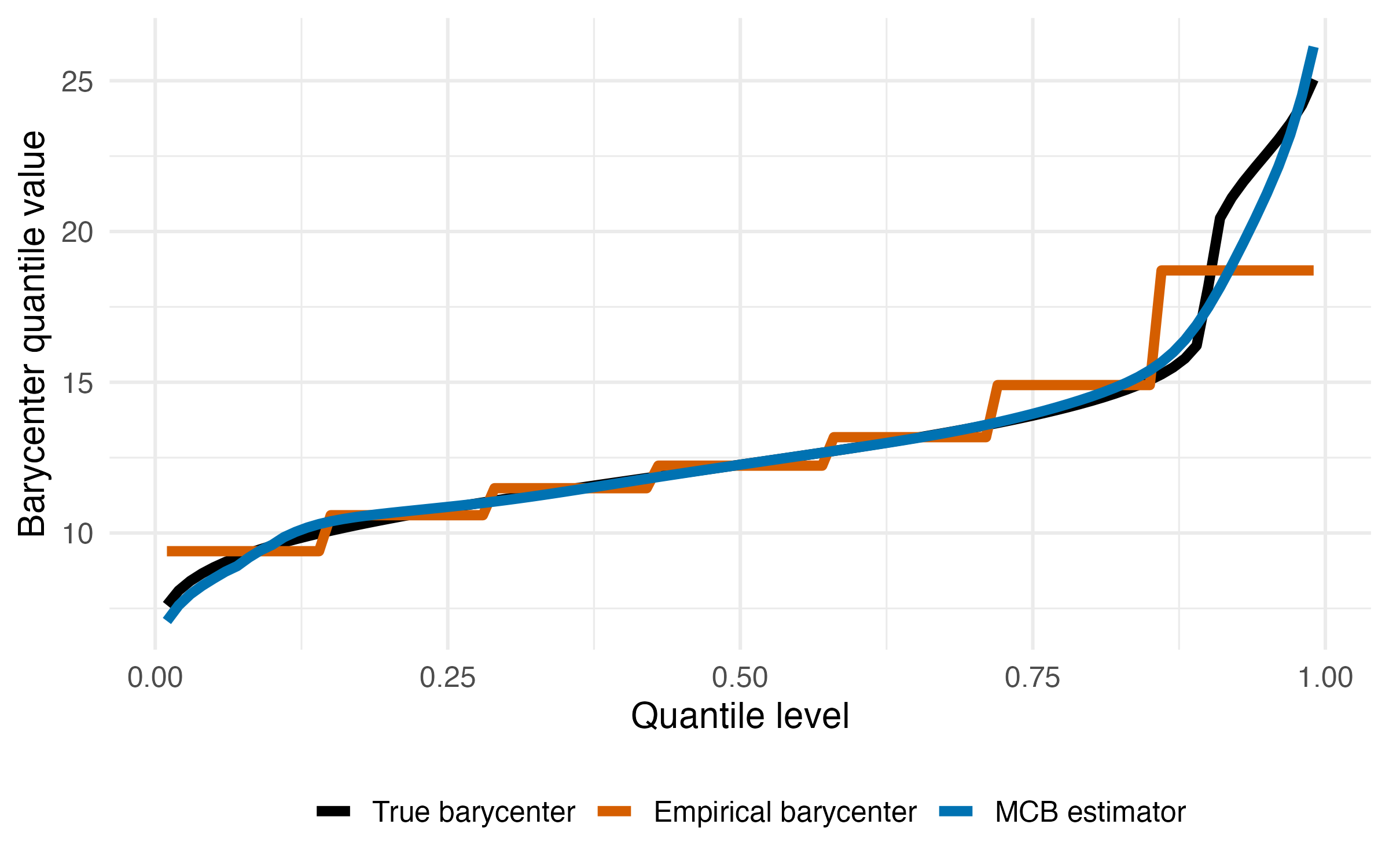}
    \caption{Quantile functions of the empirical Wasserstein barycenter (red) and the proposed MCB estimator (blue) for a single simulated data set where 1,000 distributions are drawn from the five distributions in Figure \ref{fig:barycenter_vs_mean}a and observed through only 7 i.i.d.~samples. The true barycenter is shown in black.}
    \label{fig:method_preview}
\end{figure}

To our knowledge, no existing method explicitly corrects for sparse sampling bias in estimating the Wasserstein barycenter. Previous work typically avoids this difficulty by assuming that each latent distribution can be accurately recovered from its unit-level observations, either by the empirical measure or by a smoothed version of it \parencite{bigot_upper_2018, petersen2019frechet, panaretos_invitation_2020, lin_causal_2023}. These assumptions are natural when the number of observations per distribution is large, but they do not address the sparse regime considered here, where each distribution is observed through only a small, bounded number of samples. The closest related work is \textcite{zhou_wasserstein_2024}, who explicitly study sparse sampling in Wasserstein regression. Their results characterize conditions under which the bias from using empirical measures is asymptotically negligible. However, they do not propose a correction to this bias, and consistency requires the per-distribution sample sizes to grow sufficiently quickly relative to the number of distributions. While such increasing-sample-size regimes provide important theoretical guarantees, they offer limited guidance when analyzing a fixed sparse dataset, where the number of observations per distribution is small and cannot be increased asymptotically.

The remainder of the article is organized as follows. Section \ref{sec:preliminaries} formalizes the data structure and defines the estimand of interest. Section \ref{sec:challenges} discusses estimation challenges under sparse sampling and the failure of the empirical barycenter. Section \ref{sec:identification_estimation} presents our main representation result and the resulting MCB estimation procedure. Section \ref{sec:binomial_mixtures} reviews estimation methods for the binomial mixing problem on which our estimator relies, and Section \ref{sec:asymptotics} develops the MCB estimator's asymptotic properties. Section \ref{sec:extensions} discusses extensions, including weighting and Fr\'echet variance estimation. Section \ref{sec:simulation_study} reports simulation results, Section~\ref{sec:data_application} illustrates the proposed method using data from
the HVTN 502 and HVTN 503 HIV-1 vaccine efficacy trials, and Section \ref{sec:discussion} concludes.

\section{Preliminaries}\label{sec:preliminaries}

\subsection{Data structure}\label{sec:data-structure}

We consider a two-stage sampling framework. Each observational unit $i=1,\dots,n$ is associated with an unobserved probability distribution $\nu_i$ on $\mathcal I \subseteq \mathbb R$, where $\nu_i \in \mathcal{P}_2(\mathcal{I})$ and $\mathcal{P}_2(\mathcal{I})$ denotes the set of probability measures on $\mathcal I$ with finite second moment. We assume that $\nu_1,\dots,\nu_n$ are i.i.d. draws from a population distribution $\Pi$ on $\mathcal{P}_2(\mathcal{I})$. Thus, $\Pi$ is a distribution on distributions that captures distributional heterogeneity across units.
Rather than observing $\nu_i$ directly, we observe $m_i$ i.i.d.~draws from $\nu_i$, denoted by $(X_{i,j})_{j=1}^{m_i}$, where $m_i$ is drawn from a distribution $\eta$ on the positive integers and is independent of $\nu_i$. The sampling scheme is summarized in Figure \ref{fig:two_stage_sampling} and can be written as
\begin{equation}\label{eq:two-stage}
\begin{aligned}
\nu_i &\overset{\mathrm{iid}}{\sim} \Pi, \quad i = 1, \dots, n, \\
m_i &\overset{\mathrm{iid}}{\sim} \eta, \quad m_i \indep \nu_i \\
X_{i,1}, \dots, X_{i,m_i} \mid \nu_i, m_i &\overset{\mathrm{iid}}{\sim} \nu_i.
\end{aligned}
\end{equation}

We denote the empirical measure associated with the $i$th unit as
$\hat{\nu}_i := \frac{1}{m_i}\sum_{j=1}^{m_i}\delta_{X_{i,j}}$,
where $\delta_a$ denotes a Dirac mass at the point $a$. Under this two-stage sampling framework, the full data consist of
$\{(\nu_i, m_i, (X_{i,j})_{j=1}^{m_i})\}_{i=1}^n$.
Since the latent distributions $\nu_i$ are unobserved, the observed data consist of
$\{O_i\}_{i=1}^n$, where $O_i := (m_i, (X_{i,j})_{j=1}^{m_i})$. Throughout, $\mathbb P$ and $\mathbb E$ denote probability and expectation
under the full two-stage sampling law. When probability or
expectation is taken only over one component of this law, we indicate this
with a subscript. For example, $\mathbb P_{\nu\sim\Pi}$ and
$\mathbb E_{\nu\sim\Pi}$ denote probability and expectation over a latent
distribution drawn from $\Pi$, while $\mathbb P_{m\sim\eta}$ and
$\mathbb E_{m\sim\eta}$ denote probability and expectation over a sample
size drawn from $\eta$.

\begin{figure}[!htb]
    \centering

\begin{tikzpicture}[edge from parent/.style={draw,-{Latex}}]
  \tikzstyle{level 1}=[sibling distance=2.4cm, level distance = 1cm] 
  \tikzstyle{level 2}=[sibling distance=2cm, level distance = 1cm] 

  \node (full) {$\Pi \times \eta$}
    child {node {$\nu_1, m_1$}
      child {node (x1) [draw] {$X_{1,1},\dots,X_{1,m_1}; m_1$}
        child {node (a1) [yshift=-0.15cm,draw] {$\hat{\nu}_1; m_1$} edge from parent[dashed]}
      }
    }
    child {node {$\ldots$} edge from parent[draw=none]
      child {node [draw] {$\ldots$} edge from parent[draw=none]}
    }
    child {node {$\nu_n, m_n$}
      child {node (xn) [draw] {$X_{n,1},\dots,X_{n,m_n}; m_n$}
        child {node (c1) [yshift=-0.15cm,draw] {$\hat{\nu}_n; m_n$} edge from parent[dashed]}
      }
    };

\end{tikzpicture}
    \caption{Schematic of the two-stage sampling scheme. Full arrows denote sampling, and dashed arrows denote the construction of the empirical measures. Observable elements are surrounded by a rectangle.}    \label{fig:two_stage_sampling}
\end{figure}
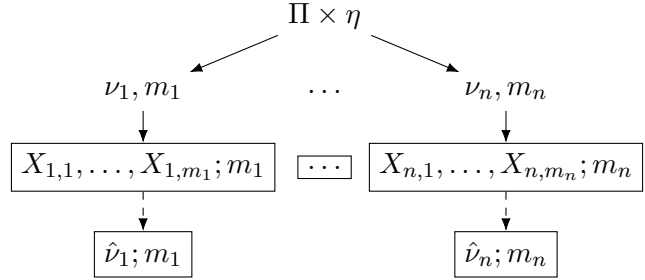
For any $\nu \in \mathcal{P}_2(\mathcal{I})$, let $F_\nu$ denote its cumulative distribution function, and define its quantile function by
$$
F_\nu^{-}(t) := \inf \left\{ x \in \mathbb{R} : F_\nu(x) \ge t \right\}, \qquad t \in (0,1).
$$

We equip $\mathcal{P}_2(\mathcal{I})$ with the 2-Wasserstein metric, defined for $\nu_1, \nu_2 \in \mathcal{P}_2(\mathcal{I})$ by
$$
d_{W_2}^2(\nu_1, \nu_2) := \int_0^1 \left(F_{\nu_1}^{-}(t) - F_{\nu_2}^{-}(t)\right)^2 \, dt.
$$
Under this metric, $\mathcal{P}_2(\mathcal{I})$ forms a metric space \parencite{panaretos_invitation_2020}.

\begin{remark}
The assumption $m_i \indep \nu_i$ is often referred to as non-informative cluster size, meaning that the number of observations for unit $i$ is independent of its latent distribution. Informative cluster size, where $m_i$ depends on $\nu_i$, is a standard concern in the clustered data literature and can invalidate inference; see, e.g., \textcite{williamson_marginal_2003}. Although $\nu_i$ is not observed directly, the non-informative cluster size assumption has testable implications \parencite{nevalainen2017tests}. 
\end{remark}

\subsection{Estimand}\label{sec:estimand}

Our primary target is the population Wasserstein barycenter of $\Pi$, defined as:
$$
\bar{\nu}_\Pi
\in
\operatorname*{argmin}_{\mu\in\mathcal P_2(\mathcal I)}
\mathbb{E}_{\nu\sim\Pi}\bigl[d_{W_2}^2(\mu,\nu)\bigr].
$$
When the dependence on $\Pi$ is clear from context, we write $\bar{\nu}$ in place of $\bar{\nu}_\Pi$. For one-dimensional distributions, the Wasserstein barycenter is unique, and its quantile function equals the mean of the quantile functions of measures drawn from $\Pi$ (Proposition 4.1 in \textcite{bigot_geodesic_2017}):
$$
t \mapsto F_{\bar{\nu}}^{-}(t)
=
\mathbb{E}_{\nu\sim\Pi}\bigl[F_\nu^{-}(t)\bigr],
\qquad t\in(0,1).
$$
Thus, inference on the barycenter can be reframed as inference on the pointwise average quantile function of the latent distributions.

For a fixed quantile level $\alpha\in(0,1)$, we define the scalar target
$$
q_\alpha^\star
:=
F_{\bar{\nu}}^{-}(\alpha)
=
\mathbb{E}_{\nu\sim\Pi}\bigl[F_\nu^{-}(\alpha)\bigr].
$$
This is the $\alpha$-quantile of the barycenter. Our main theoretical development focuses on pointwise inference for $q_\alpha^\star$ at fixed $\alpha$, since this is the quantity directly characterized by our representation result. The full barycenter is determined by considering the collection $\{q_t^\star : t\in(0,1)\}$, which is exactly the quantile function of $\bar{\nu}$. Although our main focus is the barycenter and its pointwise quantiles, the same ideas extend to other related estimands; these are discussed in Section~\ref{sec:extensions}.

\section{Challenges for estimating the barycenter}\label{sec:challenges}

In this section, we highlight two obstacles to estimating the barycenter under sparse sampling. First, when the per-unit sample sizes are bounded, the barycenter's quantile function is not pointwise identifiable from the observed data distribution without additional restrictions on $\Pi$. Second, the empirical barycenter, the natural plug-in estimator formed by averaging unit-level empirical quantiles, is generally inconsistent.

\subsection{Non-identifiability}

Different laws $\Pi$ on latent distributions may generate the same observed data distribution while having different barycenters. This is formalized in the following proposition.

\begin{proposition}\label{prop:barycenter_nonidentifiability}
Assume that the unit-specific sample sizes are bounded, in the sense that 
$\mathbb{P}_{m \sim \eta}(m \le C)=1$ for some $C < \infty$. 
Then, for every $\alpha \in (0,1)$, there exist two distributions on distributions, denoted by $\Pi_1$ and $\Pi_2$, which, together with the same $\eta$, induce the same distribution for the observed data 
$O_i=(m_i,(X_{i,j})_{j=1}^{m_i})$, but satisfy
$
q_{\alpha,\Pi_1}^\star \neq q_{\alpha,\Pi_2}^\star.
$
%Therefore, the barycenter's quantile function is not pointwise identifiable.
\end{proposition}

% \begin{proposition}
% Assume that the unit-specific sample sizes are bounded, in the sense that $P_{\eta}(m_i \le C)=1$ for some $C < \infty$. Then there exist two distributions on distributions, denoted by $\Pi_1$ and $\Pi_2$, such that
% $
% \bar{\nu}_{\Pi_1} \neq \bar{\nu}_{\Pi_2},
% $
% while $\Pi_1$ and $\Pi_2$ induce the same distribution for the observed data $O_i=(m_i,(X_{i,j})_{j=1}^{m_i})$.
% \end{proposition}

This non-identifiability is a fundamental obstacle for inference: without
additional assumptions on $\Pi$, the observed-data distribution does not in
general determine $q_{\alpha,\Pi}^\star$. Therefore, consistent estimation of $q_{\alpha,\Pi}^\star$ is generally impossible without changing the asymptotic regime or restricting the
model class. The lack of identifiability can be addressed in two ways. One approach, as discussed in the introduction, is to allow the sampling distribution $\eta$ to vary with $n$ so that unit-specific sample sizes increase and the barycenter becomes identifiable in the limit \parencite{zhou_wasserstein_2024}. 
A second approach is to assume a parametric model for the law $\Pi$ that is identifiable even when unit-specific sample sizes are bounded. Asymptotic results then follow from standard maximum likelihood theory; see Example \ref{example:gaussian_random_intercept}. 

\begin{example}\label{example:gaussian_random_intercept}
A simple fully parametric example is a Gaussian random-intercept model:
$$
\nu_i = \mathcal{N}(\mu_i,\sigma^2),
\qquad
\mu_i \overset{\mathrm{iid}}{\sim} \mathcal{N}(\mu_0,\tau^2).
$$
Then the population barycenter is
$
\bar{\nu} = \mathcal{N}(\mu_0,\sigma^2)
$, which is identifiable if some units have at least two observations. Estimation and inference for $\bar{\nu}$ can then be performed by fitting the random-intercept model using maximum likelihood or restricted maximum likelihood \parencite{molenberghs2000linear}.
\end{example}

% Our proposed approach is semiparametric: rather than fully specifying $\Pi$, we impose structure only on the aspects of $\Pi$ needed to identify the barycenter, while leaving the remaining features of $\Pi$ unrestricted. 

\subsection{The empirical Wasserstein barycenter}\label{sec:empirical_barycenter}

The non-identifiability result shows that fully nonparametric recovery of $q_\alpha^{*}$ is generally impossible under bounded unit-level sample sizes. Nevertheless, a natural plug-in estimator remains available. First, if the latent distributions $\nu_i$ were directly observed, a natural estimator of $q_\alpha^\star$ would be the oracle plug-in estimator
$$
\hat q_{\alpha}^{\,\mathrm{or}}
:=
\frac{1}{n}\sum_{i=1}^n F_{\nu_i}^{-}(\alpha).
$$
However, since the $\nu_i$ are unobserved, this estimator is infeasible. A substitute is to replace each latent distribution $\nu_i$ by its empirical measure $\hat{\nu}_i$, yielding
$$
\hat q_{\alpha}^{\,\mathrm{emp}}
:=
\frac{1}{n}\sum_{i=1}^n F_{\hat{\nu}_i}^{-}(\alpha).
$$
As $\alpha$ varies over $(0,1)$, this defines the quantile function of the \textit{empirical Wasserstein barycenter} \parencite{bigot_upper_2018}.

The next proposition shows that the expectation of $\hat q_{\alpha}^{\,\mathrm{emp}}$ is governed by order statistics from the population barycenter, averaged over the sample size distribution $\eta$. This makes clear why sparse sampling induces bias: for fixed $m$, the empirical $\alpha$-quantile behaves like the $\lceil m\alpha\rceil$th order statistic of a sample of size $m$, whose expectation generally differs from the population $\alpha$-quantile. This bias is especially pronounced for extreme quantiles, as was seen in Figure \ref{fig:method_preview}. 

\begin{proposition}\label{prop:barycenter_mean}
For each $m\ge 1$, let
$X_{(1)}^{*(m)},\dots,X_{(m)}^{*(m)}$ denote the order statistics of an
i.i.d.~sample of size $m$ drawn from the barycenter $\bar{\nu}$. Then, under
the two-stage sampling scheme in \eqref{eq:two-stage},
\begin{align}
\mathbb{E}\!\left[\hat q_{\alpha}^{\,\mathrm{emp}}\right]
&=
\mathbb{E}_{m\sim\eta}\!\left[
\mathbb{E}_{\bar{\nu}}\!\left[
X_{(\lceil m\alpha\rceil)}^{*(m)}
\right]
\right] \notag \\
&=
\sum_{m'=1}^{\infty}\eta(m')\,
\mathbb{E}_{\bar{\nu}}\!\left[
X_{(\lceil m'\alpha\rceil)}^{*(m')}
\right],
\end{align}
where $\eta(m')=\mathbb{P}_{m\sim\eta}(m=m')$ and, with slight abuse of
notation, $\mathbb{E}_{\bar{\nu}}$ denotes expectation with respect to the
i.i.d.~sample from $\bar{\nu}$ whose order statistic appears inside the
expectation.
\end{proposition}
Proposition \ref{prop:barycenter_mean} shows that $\hat q_{\alpha}^{\,\mathrm{emp}}$ is unbiased for $q_\alpha^\star$ if and only if 
\begin{equation}
\sum_{m'=1}^{\infty}\eta(m')\,\mathbb{E}_{\bar{\nu}}\!\left[X_{(\lceil m'\alpha\rceil)}^{*(m')}\right]
=
q_\alpha^\star,
\label{eq:consistent_barycenter}
\end{equation}
which only holds in exceptional cases. Since the left-hand side does not depend on $n$, the empirical barycenter is generally inconsistent for $q_\alpha^\star$ unless \eqref{eq:consistent_barycenter} holds.
Example \ref{ex:biased_emp} illustrates this point in a simple setting.

\begin{example}\label{ex:biased_emp}
Suppose the latent law $\Pi$ is degenerate at $\mathcal{N}(0,1)$ (i.e., $\nu_i = \mathcal{N}(0,1)$ for all $i$) and hence $\bar{\nu} = \mathcal{N}(0,1)$. Also suppose that all units have the same sample size $m$. Let $X_{(1)}^{*(m)},\dots,X_{(m)}^{*(m)}$ denote the order statistics of an i.i.d.~sample of size $m$ drawn from $\mathcal{N}(0,1)$. Then, as $n\to\infty$,
$$
\hat q_{\alpha}^{\,\mathrm{emp}}
\overset{P}{\longrightarrow}
\mathbb{E}\!\left[X_{(\lceil m\alpha\rceil)}^{*(m)}\right].
$$
For fixed $m$, the right-hand side is piecewise constant in $\alpha$, whereas $q_\alpha^\star = F_{\bar{\nu}}^{-1}(\alpha)$ is strictly increasing and continuous in $\alpha$, so the equality $\mathbb{E}\!\left[X_{(\lceil m\alpha\rceil)}^{*(m)}\right]=q_\alpha^\star$ can only hold at isolated values of $\alpha$. In particular, the empirical barycenter quantile estimator is inconsistent for almost all $\alpha$. By symmetry, when $m$ is odd, $\alpha=0.5$ is one value where it is consistent.
\end{example}

\section{MCB construction and estimator} \label{sec:identification_estimation}

\subsection{Alternative representation via marginal mixing distributions}

The empirical barycenter is a natural estimator because it replaces each latent (unobserved) distribution by its unit-level empirical distribution. However, as shown above, this plug-in estimator is generally inconsistent under sparse sampling. 
The key difficulty is that subject-specific quantiles
$F_\nu^{-}(\alpha)$ can be poorly estimated when only a small number of observations are available for each distribution. 

Our strategy is based on a reformulation of the problem. Although
our target,
$q_\alpha^\star=\mathbb E_{\nu\sim\Pi}\{F_\nu^{-}(\alpha)\}$,
is the expectation of the latent $\alpha$-quantiles, we instead begin with
the seemingly more ambitious goal of estimating the \textit{distribution} of the
latent $\alpha$-quantiles across units. The rationale is that this
distribution is directly linked to the distribution of the unit-level CDF
values. Indeed, for any $x\in\mathcal I$,
\begin{equation}\label{eq:cdf_quantile_equiv}
\{F_\nu^{-}(\alpha)\le x\}
\iff
\{F_\nu(x)\ge \alpha\}.
\end{equation}
Thus, although the expectation of a random quantile is not a simple transformation
of the expectation of the corresponding distribution functions, the distribution of the
latent quantiles can be translated into the distribution of the latent CDF
values. This reformulation is useful because, as we argue below,
estimating the distribution of the latent CDF values can be reduced to the
well-studied binomial mixture problem.

To make this precise, for each fixed $x\in\mathcal I$, let
\begin{equation}
G(x,t) := \mathbb{P}_{\nu\sim\Pi}\left(F_\nu(x) \le t\right), \qquad t\in[0,1],
\label{eq:g_definition}
\end{equation}
so that $G(x,\cdot)$ is the CDF of the random variable $F_\nu(x)$ when $\nu\sim\Pi$. This is illustrated in Figure \ref{fig:identification}. Equivalently, for each fixed $x$, the law $G(x,\cdot)$ serves as the mixing distribution in a binomial mixture model. Writing $U_i(x):=F_{\nu_i}(x)$ for the unobserved CDF value of unit $i$ at $x$ and
\begin{equation}
Y_i(x) := \sum_{j=1}^{m_i} \mathbf 1\{X_{i,j}\le x\}
\label{eq:binomial_count}
\end{equation}
for the number of ``successes'' in the $i$th unit's sample at $x$, we have the conditional model
\begin{equation}
Y_i(x) \mid m_i, U_i(x) \sim \operatorname{Bin}\left(m_i, U_i(x)\right).
\label{eq:mixture_model}
\end{equation}
By construction, the latent values $U_1(x),\dots,U_n(x)$ are i.i.d.~with CDF $G(x,\cdot)$, so marginally $Y_i(x)$ is a binomial mixture with mixing distribution $G(x,\cdot)$.

Theorem \ref{theorem:identification} below is a representation result: it represents the target $q_\alpha^\star$ in terms of the collection of marginal mixing distributions $\{G(x,\cdot): x \in \mathcal I\}$ induced by $\Pi$.
Hence, $q_\alpha^\star$ is identified with the observed data if the mixing distribution $G(x,\cdot)$ is identifiable for each $x \in \mathcal I$. We return to those conditions in Section~\ref{sec:binomial_mixtures}.

\begin{theorem}[Representation via marginal mixing distributions]\label{theorem:identification}
For $x\in\mathcal I$, define $G(x,\alpha-) := \lim_{t\uparrow\alpha} G(x,t),$
and let
$$
H_\alpha(x) := 1-G(x,\alpha-).
$$
Then the CDF of the distribution of $\alpha$-quantiles $F_\nu^{-}(\alpha)$, where $\nu\sim\Pi$, equals $x\mapsto H_\alpha(x)$ for $x\in\mathcal I$.
Consequently, 
\begin{equation*}
    q_\alpha^\star = \int_{\mathcal I} x\,dH_\alpha(x) = -\int_{\mathcal I} x\,dG(x,\alpha-).
\end{equation*}
\end{theorem}

The key observation behind Theorem~\ref{theorem:identification} is the
equivalence in \eqref{eq:cdf_quantile_equiv}. For each $x\in\mathcal I$,
this gives
$$
\mathbb P_{\nu\sim\Pi}\{F_\nu^{-}(\alpha)\le x\}
=
\mathbb P_{\nu\sim\Pi}\{F_\nu(x)\ge\alpha\}
=
1-G(x,\alpha-)
=
H_\alpha(x).
$$
Therefore, $x\mapsto H_\alpha(x)$ is the CDF of the latent
$\alpha$-quantile $F_\nu^{-}(\alpha)$ under $\nu\sim\Pi$. Thus, once the
marginal mixing distributions $\{G(x,\cdot):x\in\mathcal I\}$ are
identified from the observed data, the distribution of the latent
$\alpha$-quantiles is identified as well. Its mean,
$q_\alpha^\star$, is then identified by integrating with respect to
$H_\alpha$. We refer to this representation as the \emph{marginal
construction} of the barycenter.

\begin{figure}[!htb]
    \centering
    \includegraphics[width=\textwidth]{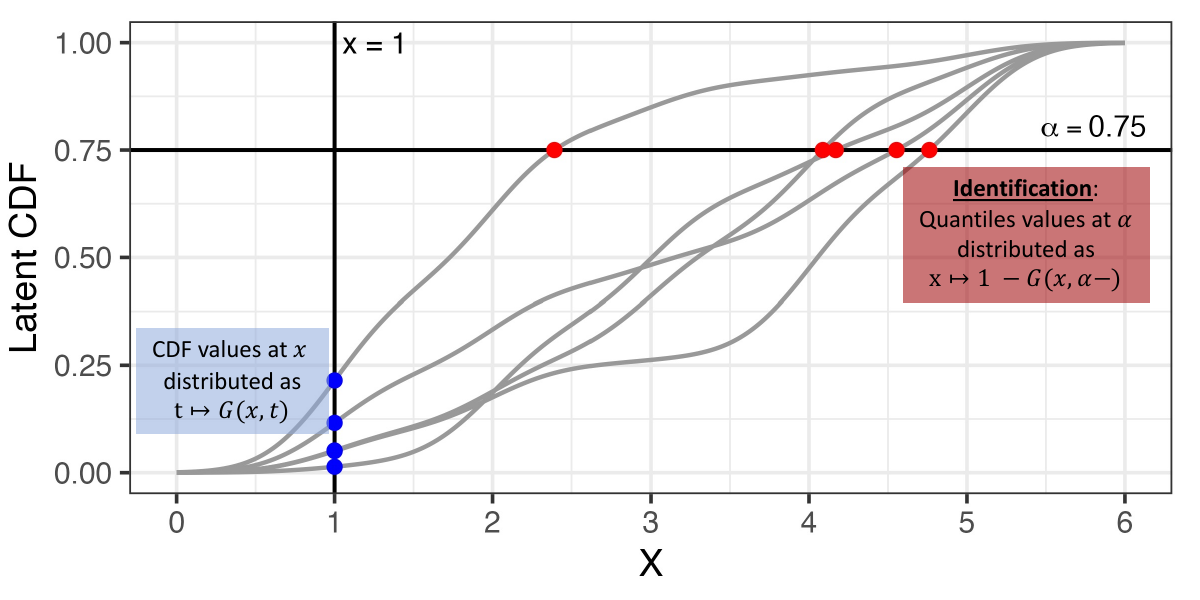}
    \caption{The figure shows the CDFs of five latent distributions. The vertical positions of the blue points are the values $F_\nu(1)$ for the five latent distributions; across $\nu\sim\Pi$, these values have CDF $t \mapsto G(1,t)$ by \eqref{eq:g_definition}. The horizontal positions of the red points are the $0.75$-quantiles of the latent distributions; by Theorem~\ref{theorem:identification}, these values have CDF $H_{0.75}(x) = 1-G(x,0.75-)$.}    \label{fig:identification}
\end{figure}

\subsection{Estimation procedure}

Theorem~\ref{theorem:identification} suggests a plug-in strategy for estimating
$q_\alpha^\star$. For a fixed $\alpha \in (0,1)$, suppose that we have an estimator
$\hat G_n(x,\alpha-)$ of $G(x,\alpha-)$ for each $x \in \mathcal I$, and define $\hat H_{n, \alpha}(x) := 1-\hat G_n(x,\alpha-)$.
Then the \emph{marginal-constructed barycenter} (MCB) estimator is given by
\begin{align}
\hat q_{\alpha}^{\,\mathrm{MCB}}
&:= \int_{\mathcal I} x \, d\hat H_{n, \alpha}(x).
\end{align}

While the true function $x \mapsto G(x,\alpha-)$ may vary over the entire interval $\mathcal I$, the data used to estimate $G(x,\cdot)$ change only at values of $x$ where observations occur.  That is, for any two points $x$ and $x'$ lying between the same pair of consecutively observed values, the indicators $\mathbf{1}\{X_{i,j}\le x\}$ and $\mathbf{1}\{X_{i,j}\le x'\}$ will coincide for each $i, j$, so the data used to estimate the corresponding mixing distributions are identical. Therefore, it suffices to estimate $G(x,\cdot)$ only at these locations. 

We implement the MCB estimator for $q_\alpha^\star$ in the following steps, illustrated in Figure \ref{fig:estimation}:
\begin{enumerate}
    \item Let
$$
\mathcal{X} := \{X_{i,j} : 1 \le i \le n,\; 1 \le j \le m_i\}
$$
denote the set of observed points across all units and let $x_1 < \cdots < x_K$ be the sorted values of $\mathcal{X}$, with $K := |\mathcal{X}|$.
    \item For each point $x_k$ with $k = 1,\dots,K-1$, estimate the binomial mixing distribution $t \mapsto G(x_k,t)$ using the Binomial data $\{(Y_i(x_k),m_i)\}_{i=1}^n$. Denote the resulting estimator by $t \mapsto \hat G_n(x_k,t)$.
    % With $\alpha$ fixed, define $\hat H_{n, \alpha}(x)$ as the stepwise function
    % $$
    % \hat H_{n, \alpha}(x)
    % :=
    % \begin{cases}
    % 0, & x < x_1, \\[0.5em]
    % 1 - \hat G_n(x_k,\alpha), & x \in [x_k,x_{k+1}), \quad k=1,\dots,K-1, \\[0.5em]
    % 1, & x \ge x_K.
    % \end{cases}
    % $$
    
    \item    With $\alpha$ fixed, define $x \mapsto \hat{H}_{n, \alpha}(x)$ as the stepwise function
    $$
    \hat H_{n, \alpha}(x)
    :=
    \begin{cases}
    0, & x < x_1, \\[0.5em]
    1 - \hat G_n(x_k,\alpha-), & x \in [x_k,x_{k+1}), \quad k=1,\dots,K-1, \\[0.5em]
    1, & x \ge x_K.
    \end{cases}
    $$ 
    In practice, we evaluate $\hat G_n(x,t)$ at $t=\alpha$ rather than $\alpha-$; this is equivalent under continuity of $t \mapsto G(x,t)$ at $t=\alpha$.
    
    We then compute the MCB estimate for $q_\alpha^\star$ as follows:
    \begin{align*}
        \hat q_{\alpha}^{\,\mathrm{MCB}}
        &:= \int_{\mathcal I} x \, d   \hat H_{n, \alpha}(x)  \\
        &= \sum_{k=1}^{K} x_k \left\{ \hat G_n(x_{k-1},\alpha) - \hat G_n(x_k,\alpha) \right\},
    \end{align*}
    where for notational convenience $\hat G_n(x_0,\alpha) := 1$ and $\hat G_n(x_K,\alpha) := 0$.

\end{enumerate}

\begin{figure}[!htb]
    \centering
    \includegraphics[width=0.95\textwidth]{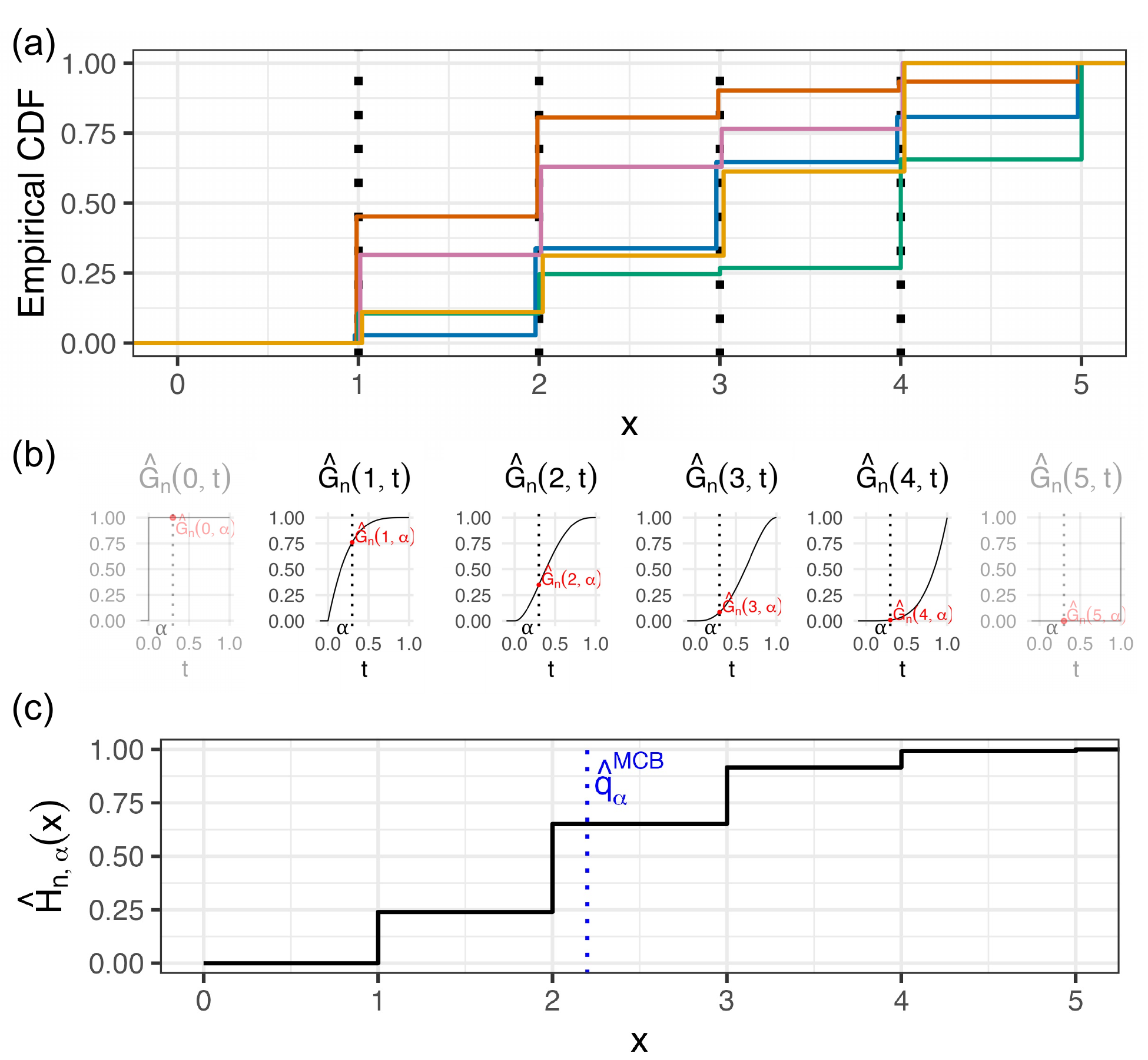}
\caption{
(a) Empirical CDFs in a toy example with five unit-level distributions. The dotted vertical lines indicate the grid points $x=1,2,3,4$ at which the binomial mixing distributions are estimated. 
(b) Estimated mixing distributions $t\mapsto \hat G_n(x,t)$ at each grid point. For points below and above the grid, the boundary distributions are fixed by construction to be degenerate at 0 and 1, respectively, and are shown with reduced opacity.
(c) Construction of $\hat H_{n,\alpha}(x)=1-\hat G_n(x,\alpha)$ from the estimated mixing distributions in panel (b). The blue dotted line marks the MCB estimate $\hat q_{\alpha}^{\,\mathrm{MCB}}$, defined as the mean of the distribution with CDF $\hat H_{n,\alpha}$.
}    \label{fig:estimation}
\end{figure}

The MCB estimator is modular: its construction is separated from the choice of binomial mixture estimator for the mixing distributions. Specifically, by reformulating barycenter estimation in terms of the marginal mixing distributions $G(x,\cdot)$, the MCB framework allows different estimators of $G(x,\cdot)$ to produce different barycenter estimators.
In Section~\ref{sec:binomial_mixtures}, we discuss potential binomial mixture estimators. 
A simple choice is to treat the empirical CDF values $\{F_{\hat\nu_i}(x)\}_{i=1}^n$ as if they were direct draws from the mixing distribution $G(x,\cdot)$, which yields the empirical histogram estimator of the mixing distribution \parencite{ye_binomial_2021}:
$$
\hat G_n^{\mathrm{emp}}(x,t)
:=
\frac{1}{n} \sum_{i=1}^n \mathbf{1}\{ F_{\hat\nu_i}(x) \le t \},
\qquad t \in [0,1].
$$
Notably, if we follow the estimation procedure using $\hat G_n^{\mathrm{emp}}(x,\cdot)$, we recover the empirical Wasserstein barycenter.

\begin{proposition}\label{prop:empirical_barycenter}
The MCB estimator that uses the empirical mixing distributions $\hat G_n^{\mathrm{emp}}(x,\cdot)$ coincides with the empirical Wasserstein barycenter. That is,
$$
\hat q_{\alpha}^{\,\mathrm{MCB}}
=
\frac{1}{n} \sum_{i=1}^n F_{\hat\nu_i}^{-}(\alpha)
=
\hat q_{\alpha}^{\,\mathrm{emp}}.
$$
\end{proposition}

Proposition~\ref{prop:empirical_barycenter} shows that the empirical barycenter arises as a special case of the MCB estimator.
This characterization shows the advantage of our formulation over the empirical barycenter: by replacing the empirical histogram estimator of $G(x,\cdot)$ with alternative estimators within the MCB framework, one may obtain barycenter estimators with improved statistical properties. In particular, if one is willing to impose parametric structure on these mixing distributions, pointwise consistent and asymptotically normal estimators of the barycenter's quantile function can be obtained. 
More generally, even without parametric structure, the empirical histogram estimator is suboptimal for estimating binomial mixing distributions \parencite{vinayak_maximum_2019, tian_learning_2017}, suggesting that other estimators of $G(x,\cdot)$ within the MCB framework may have better performance.

\begin{remark}\label{remark:isotonic-regression}
For a generic binomial mixture estimator, $x \mapsto \hat H_{n,\alpha}(x)$ may fail to be monotone increasing, even though
$H_\alpha(x)$ is a distribution function on $\mathcal{I}$. This lack of guaranteed
monotonicity arises because the marginal mixing distributions
$G(x,\cdot)$ are estimated separately across values of $x$. In implementation, one may enforce monotonicity by replacing
$\hat H_{n,\alpha}$ with its projection onto the cone of nondecreasing
functions in a post-processing step using isotonic regression. General results on isotonic
correction of non-monotone estimators imply that, under suitable regularity
conditions, such projections are asymptotically equivalent to the
unprojected estimator \parencite{westling2020correcting}. Related monotonicity corrections have
also been used in Wasserstein estimation problems; see, e.g.,
\textcite{zhou_wasserstein_2024}. 
\end{remark}

\begin{remark}
In practice, it may be computationally infeasible to estimate mixing distributions at all unique points of $\mathcal X$. In that case, we can pre-specify a grid of points $x_0 < x_1 < \cdots < x_L$ in $\mathcal I$ and approximate the integral in $\hat q_{\alpha}^{\,\mathrm{MCB}}$ by a midpoint Riemann--Stieltjes sum. For each grid point $x_\ell$, $\ell=0,\dots,L$, we estimate the binomial mixing distribution $G(x_\ell,\cdot)$ and denote the estimate by $\hat G_n(x_\ell,\cdot)$. We then approximate $\hat q_{\alpha}^{\,\mathrm{MCB}}$ as
$$
\hat q_{\alpha}^{\,\mathrm{MCB,approx}}
:=
\sum_{\ell=1}^L
\bar x_\ell
\left\{
\hat G_n(x_{\ell-1},\alpha)-\hat G_n(x_\ell,\alpha)
\right\},
\qquad
\bar x_\ell := \frac{x_{\ell-1}+x_\ell}{2}.
$$
% As the grid is refined, this approximation approaches the estimator based on all unique points in $\mathcal X$.
\end{remark}

\section{The fixed-$x$ binomial mixture problem}\label{sec:binomial_mixtures}

In this section, we review the estimation of the mixing distribution $G(x,\cdot)$ for a fixed $x\in\mathcal I$.
This fixed-$x$ binomial mixture problem is the core building block of the MCB framework: once an estimator of $G(x,\cdot)$ is available for each $x$, it can be plugged into the construction in Section~\ref{sec:identification_estimation} to form an estimator of $q_\alpha^\star$. We first clarify when $G(x,\cdot)$ is identifiable under sparse sampling, and then outline potential estimators under structural restrictions.

Recall that, for fixed $x\in\mathcal I$, we defined $
U_i(x):=F_{\nu_i}(x)$ and
$Y_i(x):=\sum_{j=1}^{m_i}\mathbf 1\{X_{i,j}\le x\}.
$ Then, conditional on $m_i$ and $U_i(x)$, $Y_i(x)\mid m_i,U_i(x)\sim \operatorname{Bin}\bigl(m_i,U_i(x)\bigr)$,
and $U_i(x)\sim G(x,\cdot)$. Therefore, inference on $G(x,\cdot)$ reduces to a binomial mixture problem with conditional likelihood
\begin{equation}
L_n\bigl(G(x,\cdot)\bigr)
=
\prod_{i=1}^n
\int_0^1
\binom{m_i}{Y_i(x)}
t^{Y_i(x)}(1-t)^{m_i-Y_i(x)}
\,dG(x,t).
\label{eq:binmix_likelihood}
\end{equation}

\subsection{Identifiability}

A key difficulty in estimating the mixing distribution $G(x,\cdot)$ is that it is not identifiable from the observed data without additional restrictions \parencite{wood_binomial_1999}. 
This is because the observed data only provide information on the first $m_i$ moments of $G(x,\cdot)$ for each unit $i$, and if the $m_i$ are bounded, then only finitely many moments of $G(x,\cdot)$ are identified. This is formalized in the following proposition.
\begin{proposition}\label{prop:mixing_nonid}
Fix $x\in\mathcal I$ and suppose that $\mathbb{P}_{m \sim \eta}(m \le C)=1$ for some finite constant $C$. Then the observed distribution of $(Y_i(x),m_i)$ depends on $G(x,\cdot)$ only through its first $C$ moments. Consequently, $G(x,\cdot)$ is not identifiable over the class of all probability distributions on $[0,1]$.
\end{proposition}

Proposition~\ref{prop:mixing_nonid} clarifies where structural assumptions can enter to ensure barycenter identifiability. We showed in Theorem \ref{theorem:identification} that the barycenter depends on $\Pi$ only through the collection of mixing distributions $\{G(x,\cdot):x\in\mathcal I\}$. 
Proposition~\ref{prop:mixing_nonid} now shows that these mixing distributions are themselves not nonparametrically identifiable from the observed data. Therefore, if structural restrictions on $G(x,\cdot)$ can make the pointwise mixing distributions identifiable from the binomial data, then the barycenter is identifiable through Theorem~\ref{theorem:identification}.

\begin{remark}
    Structural assumptions are not the only way to address the non-identifiability of the mixing distributions. For example, if $\mathbb{P}_{m \sim \eta}(m \ge M) > 0$ the first $M$ moments of $G(x,\cdot)$ are identifiable. Hence, given the first $M$ moments, one can obtain bounds for $G(x,\alpha)$ at each $x$ and consequently bounds for $q_\alpha^\star$. We do not pursue this approach in the current paper, but it may be an interesting direction for future work.
\end{remark}

\subsection{Estimation under structural restrictions}\label{sec:structural-restrictions}

To address the lack of identifiability of the collection of mixing distributions, we restrict attention to estimation under parametric models for each $G(x, \cdot)$.
Specifically, we consider finite-dimensional families for the mixing distribution at every $x\in\mathcal I$. 
Let
$$
\mathcal G_x
=
\bigl\{
G(x,\cdot;\beta):\beta\in\mathcal B_x\subseteq\mathbb R^d
\bigr\}
$$
denote a parametric class of distribution functions on $[0,1]$. 
We assume that this model is correctly specified and identifiable in the sense that $G(x,\cdot)=G(x,\cdot;\beta_0(x))$ for a unique $\beta_0(x)\in\mathcal B_x$, for each $x\in\mathcal I$.
An estimator of $\beta_0(x)$ may then be obtained by parametric maximum likelihood estimation, and the resulting estimator $\hat \beta_n(x)$ induces an estimator $\hat G_n(x,\cdot) := G(x,\cdot;\hat\beta_n(x))$ of the mixing distribution. 
% If the map from $\beta$ to the first $C$ moments of $G(x,\cdot;\beta)$ is one-to-one on $\mathcal B_x$, then $\beta_0(x)$ is identifiable.

Below, we outline two possible choices for the parametric family $\mathcal G_x$.
% These structural assumptions trade modeling bias for identifiability and more stable estimation when unit sample sizes are limited.

\paragraph{Beta specification.}
Suppose that, for every $x \in \mathcal I$, $G(x,\cdot)$ is the CDF of a $\operatorname{Beta}(a_x,b_x)$ distribution for some $a_x>0$ and $b_x>0$. Then the induced model for the observed data at $x$ is a Beta-binomial model with probability mass function:
$$
\mathbb{P} \bigl(Y_i(x)=y\mid m_i,a_x,b_x\bigr)
=
\binom{m_i}{y}
\frac{B(y+a_x,m_i-y+b_x)}{B(a_x,b_x)},
\qquad y=0,\dots,m_i,
$$
where $B(\cdot,\cdot)$ denotes the Beta function. For every $x \in \mathcal I$, the parameters $(a_x,b_x)$ may be estimated by maximizing the corresponding likelihood over $(0,\infty)^2$, yielding estimators $\hat a_x$ and $\hat b_x$. The estimated mixing distribution $\hat G_n(x,\cdot)$ is then defined as the CDF of the $\operatorname{Beta}(\hat a_x,\hat b_x)$ distribution.

\paragraph{Spline-based classes.}
A more flexible alternative to the Beta specification models the mixing density $g(x,t):=\partial_t G(x,t)$ via a finite-dimensional basis. Let $Q(t)\in\mathbb{R}^p$ be a spline basis on $[0,1]$. The density at each $x \in \mathcal{I}$ can then be modeled as follows \parencite{efron_empirical_2016}:
\begin{equation}
g\bigl(x,t;\beta(x)\bigr)
=
\exp\{Q(t)^\top\beta(x)-\phi(\beta(x))\},
\qquad
\phi(\beta)=\log\int_0^1 \exp\{Q(u)^\top\beta\}\,du.
\label{eq:spline_mixing_density}
\end{equation}
For every $x \in \mathcal I$, the parameter $\beta(x)$ may be estimated by maximizing the corresponding likelihood over $\mathbb R^p$, yielding an estimator $\hat\beta_n(x)$ and inducing an estimated mixing distribution $\hat G_n(x,\cdot)$ with density $g(\cdot;\hat\beta_n(x))$. 

Although we impose a parametric model for each marginal mixing distribution
$G(x,\cdot)$, this does not require specifying a fully parametric model for the
law $\Pi$ on unit-level distributions. The restrictions are imposed only on the marginal laws of the random CDF values $F_\nu(x)$, separately
for each $x \in \mathcal I$. The remaining aspects of the population law $\Pi$,
including the dependence structure among $\{F_\nu(x):x\in\mathcal I\}$, are left
unrestricted as long as they induce these marginal laws. In this sense,
the resulting MCB estimator is semiparametric: it uses parametric structure on
the features of $\Pi$ needed to identify the barycenter, while avoiding a full
parametric specification of $\Pi$.

\begin{remark}[Connection to the Dirichlet process]
The Beta specification includes the one-dimensional marginal distributions induced by a Dirichlet process (DP) as a special case \parencite{ferguson1973bayesian}. If $\nu\sim\mathrm{DP}(\kappa,P_0)$, where $P_0$ is the DP base measure and $\kappa>0$ is the concentration parameter, then for each $x \in \mathcal{I}$,
\[
F_\nu(x)\sim\mathrm{Beta}\bigl(\kappa\,P_0(( -\infty,x]),\;\kappa\,(1-P_0(( -\infty,x]))\bigr).
\]
\end{remark}

\subsection{Other methods}

Several other estimators for binomial mixing distributions are available, including the nonparametric maximum likelihood estimator (NPMLE) \parencite{kiefer_consistency_1956, wood_binomial_1999}, moment-based estimators  \parencite{tian_learning_2017}, and kernel density-smoothed estimators \parencite{lee_learning_2025}. 
The NPMLE, which maximizes \eqref{eq:binmix_likelihood} over all probability measures on $[0,1]$, is particularly attractive because it has no tuning parameters. In the bounded $m_i$ setting, however, the NPMLE may not be unique and its asymptotic analysis is difficult. We reserve the theoretical study of these estimators for future work, although we examine the performance of the NPMLE in our simulation study. 
 
\section{Asymptotic properties} \label{sec:asymptotics}

In this section, we remain agnostic about the particular estimator for the mixing distributions and instead impose high-level conditions on $\hat G_n(x,\cdot)$ that are sufficient for consistency and asymptotic normality of the MCB estimator $\hat q_{\alpha}^{\,\mathrm{MCB}}$ for $q_\alpha^\star$.

\subsection{Consistency}

The MCB estimator is built from pointwise estimates of the marginal mixing distributions $G(x,\cdot)$, but the target $q_\alpha^\star$ depends on these estimates through an integral over $x$. Specifically, $\hat q_{\alpha}^{\,\mathrm{MCB}} = \int_{\mathcal I} x\,d\hat H_{n,\alpha}(x)$, where $\hat H_{n,\alpha}(x)=1-\hat G_n(x,\alpha-)$. Consequently, pointwise consistency of $\hat G_n(x,\alpha-)$ for each fixed $x$ does not by itself guarantee consistency of $\hat q_{\alpha}^{\,\mathrm{MCB}}$; additional conditions are needed to control how the estimation error behaves after integration. Theorem~\ref{theorem:consistency-fixed-m-asymptotics} provides sufficient conditions under which the MCB integral is consistent.
\begin{theorem}[Consistency of the MCB estimator at fixed $\alpha$]
\label{theorem:consistency-fixed-m-asymptotics}
Assume that:
\begin{enumerate}[label=(C\arabic*), leftmargin=4em]
    \item \label{cond:pointwise_consistency}  For every continuity point $x$ of $x \mapsto H_\alpha(x)$,
    $$
        \hat H_{n,\alpha}(x) = H_\alpha(x) + r_n(x),
        \qquad r_n(x) = o_P(1).
    $$
    \item \label{cond:distribution_function} With probability tending to one, $x \mapsto \hat H_{n,\alpha}(x)$ is a cumulative distribution function on $\mathcal I$.

    \item \label{cond:tail_mass} Additionally, at least one of the following holds:
    \begin{enumerate}
        \item the support is bounded, $\mathcal I = [a,b] \subset \mathbb R$ with $|a|,|b|<\infty$; or
        \item \label{cond:uniform_integrability} the sequence $(\hat H_{n,\alpha})_{n\ge1}$ is asymptotically uniformly integrable with respect to $f:x\mapsto |x|$, per Definition~\ref{def:asymptotic-uniform-integrability} in Appendix~\ref{appendix:weak-convergence-in-prob}.
    \end{enumerate}
\end{enumerate}

Then the MCB estimator is consistent:
$
    \hat q_{\alpha}^{\,\mathrm{MCB}}
    =
    q_\alpha^\star + o_P(1).
$
\end{theorem}

The conditions in Theorem~\ref{theorem:consistency-fixed-m-asymptotics} are mild. Condition \ref{cond:pointwise_consistency} is simply pointwise consistency of the estimated CDF at each $x$, condition \ref{cond:distribution_function} can be enforced directly by projecting onto the cone of distribution functions (i.e., isotonic regression), and condition \ref{cond:tail_mass} is automatic when $\mathcal I$ is bounded. When $\mathcal I$ is unbounded, condition \ref{cond:uniform_integrability} prevents mass in the tails from dominating the integral functional but may be more difficult to verify.

While our main interest is in pointwise inference for $q_\alpha^\star$, the same argument also yields convergence in $\mathcal{P}_2(\mathcal I)$ with respect to the 2-Wasserstein metric under additional regularity conditions.

\begin{corollary}\label{corollary:consistency-in-L2}
Assume that, with probability tending to one, $t \mapsto \hat q_t^{\,\mathrm{MCB}}$ is a quantile function, and let $\hat\nu_n$ denote the corresponding probability measure. If the assumptions of Theorem~\ref{theorem:consistency-fixed-m-asymptotics} hold for every continuity point of the true barycenter quantile function $t \mapsto F_{\bar\nu_\Pi}^{-}(t)$, then
$$
d_{W_2}\!\left(\hat\nu_n,\bar\nu_\Pi\right)\overset{P}{\to}0,
$$
if $\mathcal{I}$ is bounded or if $\hat \nu_n$ is asymptotically uniformly integrable with respect to $f(x)=x^2$.
\end{corollary}

% The assumptions in Theorem~\ref{theorem:consistency-fixed-m-asymptotics} and Corollary~\ref{corollary:consistency-in-L2} are also compatible with an increasing-$m$ asymptotic argument, where the per-unit sample sizes grow with $n$. In that regime, pointwise consistency may follow under minimal (or no) structural restrictions on the mixing distributions.

\subsection{Asymptotic linearity}

We next study the asymptotic distribution of $\hat q_{\alpha}^{\,\mathrm{MCB}}$. Theorem~\ref{theorem:asymptotic-distribution} states the conditions under which pointwise asymptotic linearity of $\hat H_{n,\alpha}(x)$ for each fixed $x\in\mathcal I$ implies asymptotic linearity of the integrated estimator $\hat q_{\alpha}^{\,\mathrm{MCB}}$.

\begin{theorem}[Asymptotic linearity of the MCB estimator]
\label{theorem:asymptotic-distribution}
Assume that:
\begin{enumerate}[label=(AL\arabic*), ref=AL\arabic*, leftmargin=4em]
    \item \label{cond:AL-linear-expansion} The estimator $\hat H_{n,\alpha}$ admits the linear expansion
    $$
        \sqrt n\bigl(\hat H_{n,\alpha}(x)-H_\alpha(x)\bigr)
        =
        \frac{1}{\sqrt n}\sum_{i=1}^n \psi_\alpha(O_i;x) + r_{\alpha, n}(x),
    $$
    where, for each fixed $x$, the map $O\mapsto\psi_\alpha(O;x)$ has mean zero $\mathbb{E}[\psi_\alpha(O;x)]=0$ and $\mathbb{E}[\int_{\mathcal I} |\psi_\alpha(O;x)|\,dx] < \infty$, and the remainder satisfies $\int_{\mathcal I} |r_{\alpha, n}(x)|\,dx = o_P(1)$.

    \item \label{cond:AL-cdf-integrability} With probability tending to one, $x\mapsto\hat H_{n,\alpha}(x)$ is a cumulative distribution function on $\mathcal I$ satisfying
    $$
        \int_{\mathcal I} |x|\,d\hat H_{n,\alpha}(x) < \infty.
    $$

    \item \label{cond:AL-integrated-if} Define the integrated influence function
    $$
        \tilde\psi_\alpha(O)
        :=
        -\int_{\mathcal I} \psi_\alpha(O;x)\,dx,
    $$
    and assume that $\operatorname{Var}\bigl(\tilde\psi_\alpha(O)\bigr) =: \sigma_\alpha^2<\infty$.
\end{enumerate}

Then
$$
\sqrt n\bigl(\hat q_{\alpha}^{\,\mathrm{MCB}}-q_\alpha^\star\bigr)
    =
    \frac{1}{\sqrt n}\sum_{i=1}^n \tilde\psi_\alpha(O_i) + o_P(1)
    \overset{d}{\to}
    \mathcal N(0,\sigma_\alpha^2).
$$
\end{theorem}

The hardest-to-verify assumption is the remainder condition $\int_{\mathcal I} |r_{\alpha, n}(x)|\,dx = o_P(1)$,
which ensures that the pointwise expansion remains valid after integration over $x$. Pointwise negligibility of $r_n(x)$ for each fixed $x$ is not enough on its own, since small errors can accumulate when integrated over the whole domain. This condition is automatically satisfied when all distributions are supported on a common finite set of points, since in that case the integral reduces to a finite sum of pointwise remainder terms. Moreover, if the interval $\mathcal{I}$ is bounded, this assumption can be replaced with the uniform asymptotic linearity condition that $\sup_{x \in \mathcal{I}} |r_n(x)| = o_P(1)$.

The asymptotic variance $\sigma_\alpha^2$ can be estimated consistently by $\hat\sigma_{n,\alpha}^2 := \frac{1}{n}\sum_{i=1}^n \hat{\tilde\psi}_\alpha( O_i)^2$ under standard conditions on the estimator $\hat{\tilde\psi}_\alpha$ of $\tilde\psi_\alpha$. This yields asymptotically valid $(1-\gamma)$ Wald confidence intervals for $q_\alpha^\star$ of the form
$$
    \left[
        \hat q_{\alpha}^{\,\mathrm{MCB}}
        \;\pm\;
        z_{1-\gamma/2}\frac{\hat\sigma_{n,\alpha}}{\sqrt n}
    \right],
$$
where $z_{1-\gamma/2}$ is the $(1-\gamma/2)$-quantile of the standard normal distribution.

In practice, direct estimation of $\tilde\psi_\alpha$ requires first estimating the pointwise influence function $\psi_\alpha(\,\cdot\,;x)$ and then integrating:
$$
\hat{\tilde\psi}_\alpha = -\int_{\mathcal I} \hat\psi_\alpha(\,\cdot\,;x)\,dx.
$$
Even when $\hat\psi_\alpha(\,\cdot\,;x)$ has a closed form, this integral is typically evaluated as a large discrete sum, since $\hat\psi_\alpha(\,\cdot\,;x)$ changes at the observed points $x \in \mathcal X$. The bootstrap avoids this additional calculation and may therefore be more practical, especially when the fixed-$x$ mixture estimator is computationally efficient.

Under stronger uniformity conditions, the pointwise asymptotic linearity result extends to weak convergence of the estimated barycenter quantile process to a Gaussian process.

\begin{corollary}[Weak convergence to a Gaussian process]\label{corollary:uniform-convergence}
Suppose that the conditions of Theorem~\ref{theorem:asymptotic-distribution} hold for all $\alpha$ in a compact subset $\mathcal{A} = [\alpha_{\min},\alpha_{\max}]\subset(0,1)$ and that the following conditions hold:
\begin{itemize}
    \item[(i)] the probability (tending to one) statements from Theorem~\ref{theorem:asymptotic-distribution} hold uniformly over $\alpha \in \mathcal{A}$,
    \item[(ii)] $\sup_{\alpha \in \mathcal{A}}\int_{\mathcal I}|r_{\alpha,n}(x)|\,dx = o_P(1)$, and
    \item[(iii)] the class $\{\tilde\psi_\alpha:\alpha\in \mathcal{A}\}$ is $P$-Donsker.
\end{itemize}
Then the stochastic process
$$
\sqrt{n}\bigl(\hat q_\alpha^{\,\mathrm{MCB}} - q_\alpha^\star\bigr)_{\alpha\in \mathcal{A}}
$$
converges weakly in $\ell^\infty(\mathcal{A})$ to a centered Gaussian process with covariance function
$$
\Gamma(\alpha,\alpha') := \mathbb E\left[\tilde\psi_\alpha(O)\,\tilde\psi_{\alpha'}(O)\right].
$$
\end{corollary}

Corollary~\ref{corollary:uniform-convergence} provides the basis for
simultaneous confidence bands over
$\mathcal A$. In particular, if $Z$ denotes
the centered Gaussian process with covariance function $\Gamma$ from
Corollary~\ref{corollary:uniform-convergence}, and
$\sigma_\alpha^2=\Gamma(\alpha,\alpha)$, then the limiting random variable
governing the maximal standardized error is $T=
\sup_{\alpha\in\mathcal A}
\left|
\frac{Z(\alpha)}{\sigma_\alpha}
\right|$. Here we assume $\inf_{\alpha\in\mathcal A} \sigma_\alpha > 0$ to avoid degeneracy.
Let $c_{1-\gamma}$ denote the $(1-\gamma)$ quantile of the distribution of
$T$. Then, if $\sup_{\alpha\in\mathcal A} |\sigma_\alpha - \hat\sigma_{n,\alpha}| = o_P(1)$, an asymptotic simultaneous $(1-\gamma)$ confidence band is
$$
\left[
\hat q_{\alpha}^{\,\mathrm{MCB}}
\pm
c_{1-\gamma}\frac{\hat\sigma_{n,\alpha}}{\sqrt n}
\right],
\qquad \alpha\in\mathcal A.
$$
In practice, $c_{1-\gamma}$ can be approximated using either a multiplier
bootstrap based on the estimated integrated influence functions
$\hat{\tilde\psi}_\alpha(O_i)$, or a nonparametric bootstrap that resamples
the observational units $O_i$ with replacement and recomputes
$\hat q_{\alpha}^{\,\mathrm{MCB}}$ over a fine grid of quantile levels.

\begin{example}[Beta specification]
For the Beta specification, in which each mixing distribution $G(x,\cdot)$ is modeled by a Beta distribution, we translate the high-level conditions of Theorem~\ref{theorem:asymptotic-distribution} into lower-level sufficient conditions in Appendix~\ref{appendix:parametric-binomial-mixtures}. Under these conditions, the MCB estimator with Beta mixing distributions is asymptotically normal at each fixed $\alpha$.
\end{example}

\section{Extensions}\label{sec:extensions}

The MCB construction extends beyond the estimation of the barycenter quantile $q_\alpha^\star$. In particular, it can be adapted to weighted barycenters and other functionals of the distribution of $\alpha$-quantiles $F^-_{\nu}(\alpha)$. We discuss these extensions briefly and leave a detailed treatment of their asymptotic properties to future work.

\subsection{Weighted barycenters}

Suppose that we observe a covariate $Z_i \in \mathcal Z$ for each unit, and that $(\nu_i,Z_i)\overset{\mathrm{iid}}{\sim}\Pi_Z$, where $\Pi_Z$ is a probability measure on $\mathcal P_2(\mathcal I)\times\mathcal Z$ whose marginal law for $\nu_i$ is $\Pi$. Further let $w:\mathcal Z\to[0,\infty)$ be a measurable weight function satisfying $\mathbb E\{w(Z)\}=1$. 
We can then define the weighted barycenter quantile by
\begin{equation}\label{eq:reweighted-target}
q_\alpha^{\star,w}
:=
\mathbb E\!\left[w(Z)\,F_\nu^{-}(\alpha)\right].
\end{equation}
Equivalently, $q_\alpha^{\star,w}$ is the barycenter under the probability measure $\Pi^w$, defined as the marginal law for $\nu$ under $\Pi_Z^w$ with $d\Pi_Z^w(\nu,z):=w(z)\,d\Pi_Z(\nu,z)$. Hence, we can also write $q_\alpha^{\star,w}=\mathbb E_{\nu \sim \Pi^w}\!\left[F_\nu^{-}(\alpha)\right]$.
Next, we explain how the MCB estimator can be adapted to estimate $q_\alpha^{\star,w}$.

The same intermediate objects as in Section \ref{sec:identification_estimation} can be defined under the probability measure $\Pi^w$. In particular, for each $x\in\mathcal I$, the random variable $F_\nu(x)$ has the following distribution function under $\Pi^w$:
\begin{equation}\label{eq:reweighted-G}
G^w(x,t)
:=
\mathbb{P}_{\nu \sim \Pi^w}\!\bigl(F_\nu(x)\le t\bigr),
\qquad t\in[0,1].
\end{equation}
Similarly, letting $H_\alpha^w(x):=1-G^w(x,\alpha-)$, the same argument as in
Section~\ref{sec:identification_estimation} yields the representation
\begin{equation}\label{eq:reweighted-representation}
q_\alpha^{\star,w}
=
\int_{\mathcal I} x\,dH_\alpha^w(x)
=
-\int_{\mathcal I} x\,dG^w(x,\alpha-).
\end{equation}

The MCB construction extends directly to this setting by replacing the unweighted binomial-mixture problem at each $x$ with a weighted version.
Let $w_i:=w(Z_i)$ and recall that $Y_i(x):=\sum_{j=1}^{m_i}\mathbbm{1}\{X_{i,j}\le x\}$.
We can estimate $G^w(x,\cdot)$ using the weighted conditional likelihood
\begin{equation}\label{eq:weighted-binmix-likelihood}
L_n\bigl(G^w(x,\cdot)\bigr)
=
\prod_{i=1}^n
\left\{
\int_0^1
\binom{m_i}{Y_i(x)}
t^{Y_i(x)}(1-t)^{m_i-Y_i(x)}
\,dG^w(x,t)
\right\}^{w_i}.
\end{equation}
Since multiplying all weights by a common constant does not change the maximizer of \eqref{eq:weighted-binmix-likelihood}, one may equivalently use normalized sample weights.
Analogous consistency and asymptotic normality results can be obtained by applying Theorems~\ref{theorem:consistency-fixed-m-asymptotics} and~\ref{theorem:asymptotic-distribution} to the weighted estimators, with $G$ and $H_\alpha$ replaced by $G^w$ and $H_\alpha^w$. Verifying these conditions will generally require additional assumptions on the weights.

Below, we give two examples of weighted barycenters that may be of interest in applications. 

\begin{example}[Conditional barycenter]\label{sec:conditional_bary}
Let $Z \in \mathcal{Z} \subset \mathbb{R}$ be a continuous covariate, and let
\begin{equation}\label{eq:conditional-target}
q_\alpha^\star(z)
:=
\mathbb E\!\left[F_\nu^{-}(\alpha)\mid Z=z\right]
\end{equation}
denote the $\alpha$-quantile of the conditional Wasserstein barycenter given $Z = z$.
A natural estimator can be obtained by local weighting \parencite{hein2009robust}. Let $K$ be a kernel and let $K_h(u)=K(u/h)/h$. For the population target, define $w_{z,h}(Z):=K_h(Z-z)/\mathbb E\{K_h(Z-z)\}$. The resulting localized target is $q_{\alpha,h}^\star(z):=\mathbb E\!\left[w_{z,h}(Z)\,F_\nu^{-}(\alpha)\right]$, which is a weighted barycenter quantile of the form \eqref{eq:reweighted-target}. In estimation, the population normalizing constant is replaced by its empirical analogue, giving sample weights $\hat w_{z,h}(Z_i):=K_h(Z_i-z)/\{\frac{1}{n}\sum_{j=1}^n K_h(Z_j-z)\}$. The localized target approximates $q_\alpha^\star(z)$ in \eqref{eq:conditional-target} as $h\to0$ under standard smoothing conditions.
\end{example}

\begin{example}[IPW barycenter]
Suppose $Z=(W,A) \in \mathcal{Z}$, where $A\in\{0,1\}$ is a treatment indicator and $W \in \mathbb{R}^d$ is a vector of baseline covariates. For $a\in\{0,1\}$, let
$q_\alpha^{\star,(a)}
:=
\mathbb E\!\left[\mathbb E\!\left\{F_\nu^{-}(\alpha)\mid W,A=a\right\}\right]$, which is well defined under positivity of treatment assignment.
Under standard causal identification assumptions (consistency and exchangeability), $q_\alpha^{\star,(a)}$ identifies the
$\alpha$-quantile of the counterfactual Wasserstein barycenter under
assignment of all units to treatment level $a$, as studied by
\textcite{lin_causal_2023}. The statistical target $q_\alpha^{\star,(a)}$ can also be written as follows under positivity:
\begin{equation}\label{eq:ipw-identification}
q_\alpha^{\star,(a)}
=
\mathbb E\!\left[
\frac{\mathbbm{1}(A=a)}{\pi_a(W)}\,F_\nu^{-}(\alpha)
\right],
\qquad
\pi_a(W):=P(A=a\mid W).
\end{equation}
Thus, $q_\alpha^{\star,(a)}$ is a weighted barycenter quantile of the form
\eqref{eq:reweighted-target}, with weight $w_i := \mathbbm{1}(A=a)/\pi_a(W_i)$.
\end{example}

\subsection{Alternative functionals of the distribution of quantiles}

The representation in Theorem~\ref{theorem:identification} also yields estimators for other functionals of the latent $\alpha$-quantile distribution. Beyond the mean, a commonly used choice is the variance of the $\alpha$-quantiles under $\Pi$:
\begin{equation}\label{eq:latent-quantile-variance}
v_\alpha^\star := \mathrm{Var}_{\nu\sim\Pi}\left(F_\nu^{-}(\alpha)\right)
= \int_{\mathcal I} x^2\,dH_\alpha(x) - \left(\int_{\mathcal I} x\,dH_\alpha(x)\right)^2 = \int_{\mathcal I} x^2\,dH_\alpha(x) - \left(q_\alpha^\star\right)^2.
\end{equation}
A simple plug-in estimator is
\begin{equation}
\hat v_{n,\alpha} := \int_{\mathcal I} x^2\,d\hat H_{n,\alpha}(x) - \left(\hat q_\alpha^{\mathrm{MCB}}\right)^2.
\end{equation}
Integrating \eqref{eq:latent-quantile-variance} over $\alpha\in(0,1)$ recovers the Fr\'echet variance under the $2$-Wasserstein metric:
\begin{equation}\label{eq:frechet-variance-from-v-alpha}
\mathbb E_{\nu\sim\Pi}\left[d_{W_2}^2(\bar\nu_\Pi,\nu)\right] = \int_0^1 v_t^\star\,dt.
\end{equation}
Replacing $v_t^\star$ by the plug-in $\hat v_{n,t}$ gives an estimator of the Fr\'echet variance, providing a scalar summary of the heterogeneity of distributions under~$\Pi$.

\section{Simulation study} \label{sec:simulation_study}

We conduct a simulation study to evaluate the operating characteristics of the MCB estimator %, including bias, root mean squared error (RMSE), standard error, and confidence interval coverage, 
and compare its performance with that of the empirical barycenter. We also examine the behavior of the estimator under misspecification of the parametric model for the mixing distributions. 

Data are generated using a two-stage sampling procedure.
In the first stage, we sample $n = 50, 200, 500$ distributions from $\Pi$. We consider two simulation settings for $\Pi$:
\begin{enumerate}
  \item[(A)] A Dirichlet process with concentration parameter 2 and base measure given by a modified Kumaraswamy distribution scaled to $(0, 10)$, with shape parameters 0.2 and 0.6 \parencite{sagrillo2021modified}.
  \item[(B)] A continuous mixture of two normal distributions. The first component has mean uniformly distributed on $(-3, 0)$ and standard deviation uniformly distributed on $(1, 2)$. The second component has mean uniformly distributed on $(2, 10)$ and standard deviation uniformly distributed on $(1, 2)$, with the mixture weight on the second component uniformly distributed on $(0.1, 0.2)$.
\end{enumerate}
In the second stage, for each distribution sampled in the first stage, we draw $m_i$ uniformly from $\{3,\ldots,10\}$ and then sample $m_i$ observations from that distribution.

We configure the MCB estimator to estimate the mixing distributions using a Beta model, splines with $df = 7$ \parencite{efron_empirical_2016,  narasimhan2020deconvolver}, or the NPMLE \parencite{koenker2017rebayes}. The Beta configuration is correctly specified for Simulation Setting A but not for Simulation Setting B, while the spline configuration is misspecified for both settings. To ease computation, we estimate each mixing distribution on 50 equally spaced cutpoints between the minimum and maximum values of the pooled observations in the simulated dataset. We estimate standard errors using the bootstrap with 500 replicates. For the Beta- and spline-configured MCB estimators, we construct Wald confidence intervals based on the bootstrap standard errors. Because theoretical results are not available for the NPMLE-configured estimator, we use the mean of the bootstrap estimates as the point estimate for stability and report percentile-based bootstrap confidence intervals for this estimator.

Across 500 simulation runs, we display the mean estimate, RMSE, mean standard error, and 95\% confidence interval coverage for the estimated barycenter quantile function on the grid $0.01, 0.02, \ldots, 0.99$.

\subsection{Simulation Setting A}

Results for Setting A are displayed in Figure~\ref{fig:simulations-dirichlet}. Overall, the empirical barycenter performs poorly relative to the MCB estimators. It has larger RMSE across most quantile levels, and its confidence interval coverage deteriorates as $n$ increases, indicating persistent bias that is not eliminated by increasing the number of units.

The RMSE results show the clearest advantage for the correctly specified Beta-configured MCB estimator. At $n = 500$, the empirical barycenter had average RMSE of 0.25 across quantile levels, compared with 0.014 for the Beta-configured MCB estimator, 0.066 for the spline-configured MCB estimator, and approximately 0.25 for the NPMLE-configured MCB estimator. Table~\ref{tab:tail_center_rmse} shows that the gains from the MCB estimators are most pronounced for tail quantiles, where bias in the empirical barycenter is especially large.

The coverage results tell a similar story. At $n = 500$, the empirical barycenter's coverage ranged from 0\% to 97\% and was well below the nominal 95\% level for many quantile levels. In contrast, the Beta-configured MCB estimator achieved near-nominal confidence interval coverage across quantile levels and sample sizes, with coverage ranging from 90\% to 98\% at $n = 500$.

The spline-configured MCB estimator also substantially outperforms the empirical barycenter, despite being misspecified. It has lower RMSE and confidence interval coverage closer to the nominal level, although its coverage decreases with increasing $n$, as expected under model misspecification. Although we did not study the theoretical properties of the NPMLE-configured MCB estimator, it exhibits close-to-nominal percentile-based confidence interval coverage. Its RMSE, however, exceeds that of the empirical barycenter for central quantile levels. Among the MCB estimators, standard errors are generally largest for the NPMLE configuration and smallest for the Beta configuration.

Overall, these results suggest that, when the parametric submodel is correctly specified or reasonably flexible, the MCB estimator improves RMSE and confidence interval coverage relative to the empirical barycenter under sparse sampling.

\begin{figure}[!htb]
    \centering
    \includegraphics[width=\textwidth]{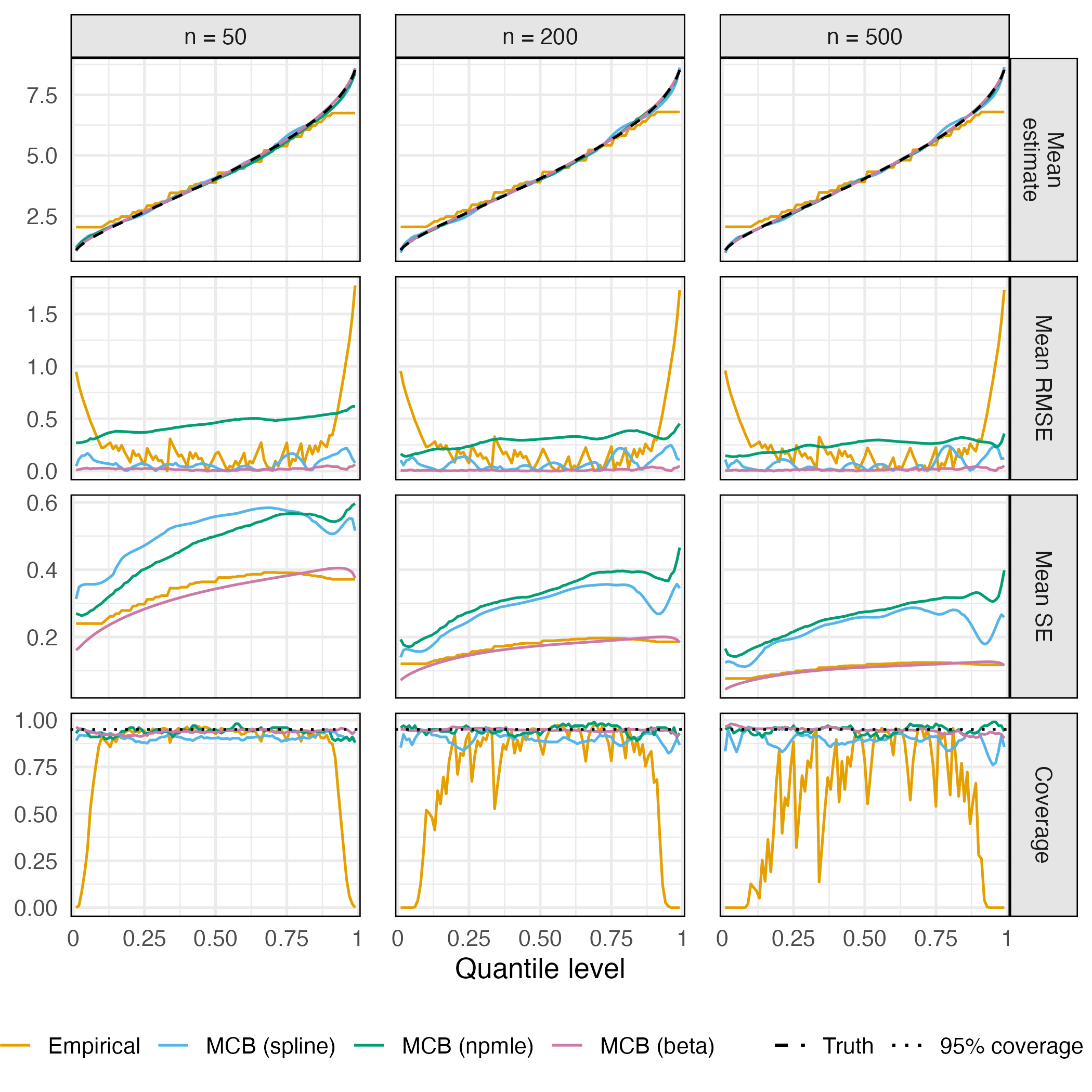}
    \caption{Results for Simulation Setting A. Rows display, from top to bottom, the mean estimate, RMSE, mean standard error, and 95\% confidence interval coverage for all quantile levels across simulation iterations. Wald confidence intervals are shown for the empirical barycenter, the Beta-configured MCB estimator, and the spline-configured MCB estimator, while percentile confidence intervals are shown for the NPMLE-configured MCB estimator.}
    \label{fig:simulations-dirichlet}
\end{figure}

\begin{table}[!htb]
\centering
\caption{Average, tail, and center RMSE by simulation setting for $n=500$. Average RMSE is averaged over all quantiles $\alpha \in \{0.01, 0.02, \ldots, 0.99\}$. Lower tail, center, and upper tail RMSE are averaged over quantiles
$\alpha\in\{0.01,\ldots,0.05\}$, $\alpha\in\{0.06,\ldots,0.94\}$,
and $\alpha\in\{0.95,\ldots,0.99\}$, respectively.}
\label{tab:tail_center_rmse}
\begin{tabular}{llcccc}
\toprule
Setting & Estimator
& \shortstack{Average\\RMSE\\$(0.01$--$0.99)$}
& \shortstack{Lower tail\\RMSE\\$(0.01$--$0.05)$}
& \shortstack{Center quantiles\\RMSE\\$(0.06$--$0.94)$}
& \shortstack{Upper tail\\RMSE\\$(0.95$--$0.99)$} \\
\midrule
\multirow{4}{*}{A} & Empirical barycenter & 0.250 & 0.748 & 0.166 & 1.254 \\
& MCB-Beta & 0.014 & 0.007 & 0.014 & 0.025 \\
& MCB-Spline & 0.066 & 0.078 & 0.059 & 0.184 \\
& MCB-NPMLE & 0.245 & 0.141 & 0.250 & 0.273 \\
\midrule
\multirow{4}{*}{B} & Empirical barycenter & 0.605 & 3.223 & 0.431 & 1.077 \\
& MCB-Beta & 0.237 & 0.211 & 0.241 & 0.205 \\
& MCB-Spline & 0.164 & 0.692 & 0.121 & 0.399 \\
& MCB-NPMLE & 0.205 & 0.264 & 0.198 & 0.264 \\
\bottomrule
\end{tabular}
\end{table}

\subsection{Simulation Setting B} 

Results for Setting B are displayed in Figure~\ref{fig:simulations-continuous}. As in Setting A, the empirical barycenter has substantially larger RMSE than the MCB estimators, with the largest errors occurring in the tails. At $n = 500$, the empirical barycenter had an average RMSE of 0.61, compared with 0.24 for the Beta-configured MCB estimator, 0.16 for the spline-configured MCB estimator, and 0.21 for the NPMLE-configured MCB estimator. As with Simulation Setting A, Table~\ref{tab:tail_center_rmse} shows that this improvement is especially pronounced for tail quantiles, where the empirical barycenter has particularly large RMSE.

Unlike Setting A, both the Beta- and spline-configured MCB estimators are misspecified in Setting B. Despite this misspecification, they still improve point estimation relative to the empirical barycenter, with the spline configuration achieving the lowest average and center-quantile RMSE. Their confidence interval coverage, however, is poor and generally worsens as $n$ increases, as expected when the fitted mixing model is misspecified. The NPMLE-configured MCB estimator achieves close-to-nominal percentile-based confidence interval coverage across nearly all quantile levels, albeit with larger standard error estimates. Overall, these results show that the MCB estimator can greatly improve point estimation even under mixing model misspecification, but that coverage can be poor because the estimator converges to a misspecified limiting target rather than the true barycenter.

\begin{figure}[!htb]
    \centering
    \includegraphics[width=\textwidth]{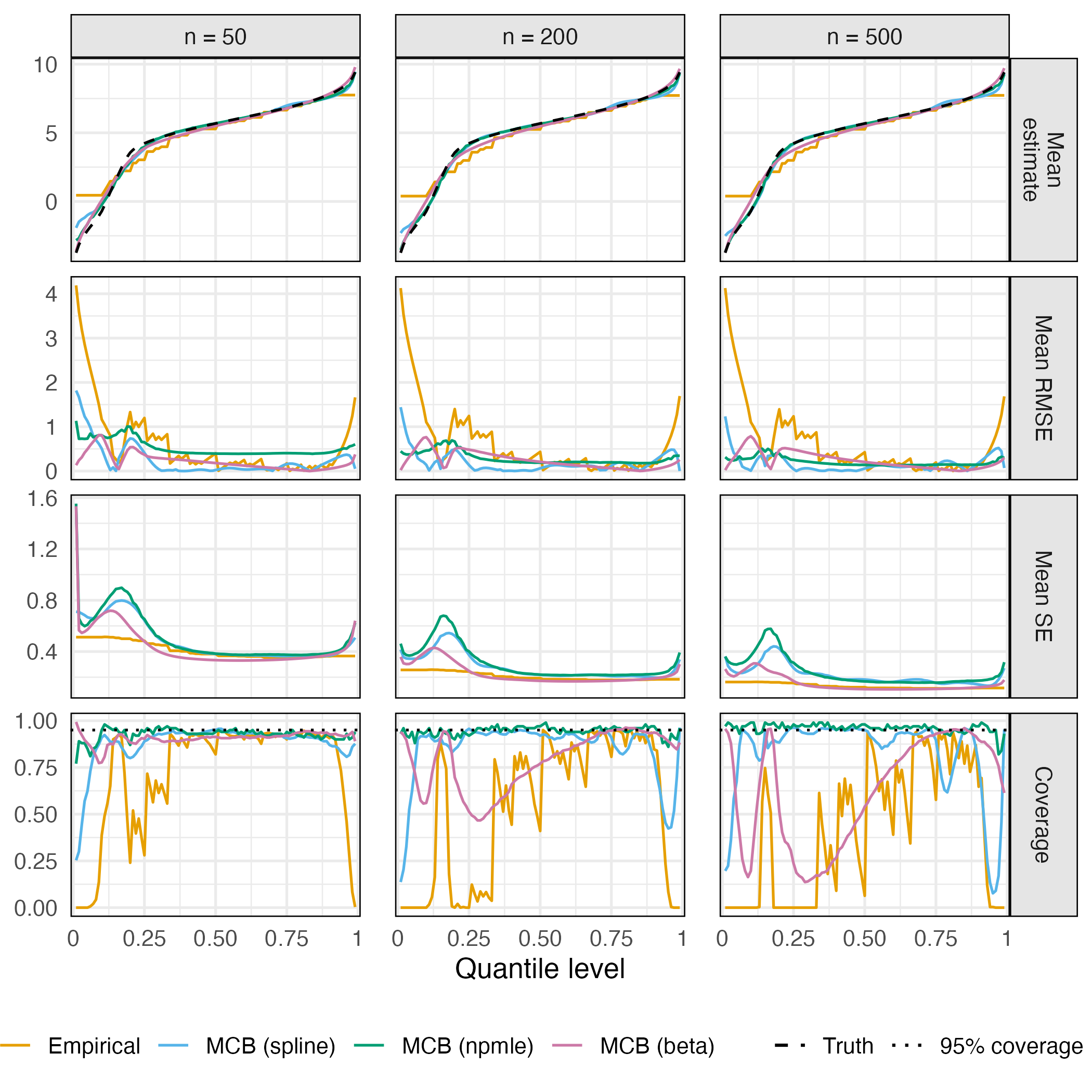}
    \caption{Results for Simulation Setting B. Rows display, from top to bottom, the mean estimate, root mean squared error (RMSE), mean standard error, and 95\% confidence interval coverage for all quantile levels across simulation iterations. Wald confidence intervals are shown for the empirical barycenter, the Beta-configured MCB estimator, and the spline-configured MCB estimator, while percentile confidence intervals are shown for the NPMLE-configured MCB estimator.}
    \label{fig:simulations-continuous}
\end{figure}

\section{Data application}\label{sec:data_application}

We illustrate the proposed method using HIV-1 sequence data from the HVTN 502 (Step) and HVTN 503 (Phambili) vaccine efficacy trials. 
Additional illustrative data applications are included in Appendix \ref{appendix:additional-data-applications}. 

Both the Step and Phambili trials were randomized placebo-controlled phase 2b trials of the Merck MRKAd5 HIV-1 subtype B Gag/Pol/Nef vaccine \parencite{gray2010overview}. The Step trial was conducted in North America, the Caribbean, South America, and Australia and enrolled 3,000 participants; vaccinations were discontinued after an interim analysis showed no efficacy against HIV-1 acquisition \parencite{buchbinder2008efficacy}. The Phambili trial evaluated the same vaccine in South Africa, a predominantly subtype C epidemic setting, and was halted early following the Step results \parencite{gray2011safety}.

Our analysis uses viral sequence data from participants who acquired HIV infection and whose viral sequences were sampled near the time of diagnosis. In the Step trial, sequence data were available from 65 participants (39 vaccine, 26 placebo), with an average of 6.58 sequences per participant (range 1--12). In the Phambili trial, sequence data were available from 43 participants (23 vaccine, 20 placebo), with an average of 6.51 sequences per participant (range 2--12). Within each trial, all viral sequences were aligned using a multiple sequence alignment, and the distances described below were computed from the aligned sequences. We summarize each viral sequence with two numeric features: (1) the physicochemically weighted mismatch distance between the sequence and the vaccine insert across positions 77–85 of the Gag protein, using the HXB2 numbering system, at which SLYNTVATL is the vaccine-strain amino acid sequence and represents a well-known epitope that elicits HIV-1-specific CD8+ T-cell responses; and (2) the same metric computed over the entire Gag protein.

For each participant, the subject-specific distribution of a given feature represents the within-host distribution of viral distances from the vaccine insert. Our objective is to compare the barycenters of these subject-specific distance distributions between the vaccine and placebo arms. Because these distances measure how similar the infecting viral populations are to the vaccine insert, differences between treatment-arm barycenters may suggest that vaccination was associated with infection by viruses with different sequence characteristics. We focus on the Gag protein because the evaluated vaccine was designed to either reduce acquisition of HIV-1 and/or reduce viral load, based on CD8+ T-cell responses that react with HIV-1 epitopes in core proteins, especially Gag. This rationale was supported by natural history studies showing that such immune responses were correlated with beneficial outcomes. Prior analyses of viral sequences from these trials also found evidence of vaccine-associated sequence divergence in Gag \parencite{rolland2011genetic, hertz2016study}. These analyses, however, reduced the observed sequences for each participant to a single representative sequence. We revisit this analysis while retaining all observed sequences, treating them as samples from the subject-specific viral population.

We analyzed the Step and Phambili trials separately. Within each trial, we applied the MCB estimator separately to the vaccine and placebo arms, yielding an estimated barycenter quantile function for each arm. We then compared the vaccine- and placebo-arm barycenters to assess whether the distribution of sequence-derived distances differed between treatment groups. We estimated these barycenter quantile functions on the grid of quantile levels $\alpha = 0.01, 0.02, \ldots, 0.99$ using the Beta-configured MCB estimator. The Beta configuration was used because it provides more stable estimates in the presence of small sample sizes. We obtain pointwise confidence intervals across the quantile grid using the nonparametric bootstrap with $B=500$ replications and Wald-type intervals based on the bootstrap standard errors.

The estimated barycenter quantile functions are displayed in Figure \ref{fig:application} together with 95\% pointwise confidence intervals. For the SLYNTVATL feature, the estimated barycenter quantile functions are nearly flat, indicating generally limited within-host heterogeneity in this sequence feature. The MCB estimator and the empirical barycenter closely match across all quantile levels. This agreement is expected when the unit-level distributions are close to point masses: in such settings, even sparsely sampled empirical quantiles can provide accurate estimates of the latent quantiles, leaving little bias for the MCB estimator to correct. There is also a larger difference between the vaccine and placebo barycenters in the Step trial than in the Phambili trial.

For the full Gag protein feature, the estimated quantile functions show greater within-host heterogeneity. In this setting, the MCB estimator and empirical barycenter agree closely in the central quantile range, but differ noticeably in the tails. This pattern is consistent with the simulation results, where sparse-sampling bias was most pronounced for tail quantiles and where the MCB estimator provided the largest correction. For the full Gag protein feature, the vaccine and placebo barycenters appear more separated in the Phambili trial than in the Step trial.

\begin{figure}[!htb]
    \centering
    \includegraphics[width=\textwidth]{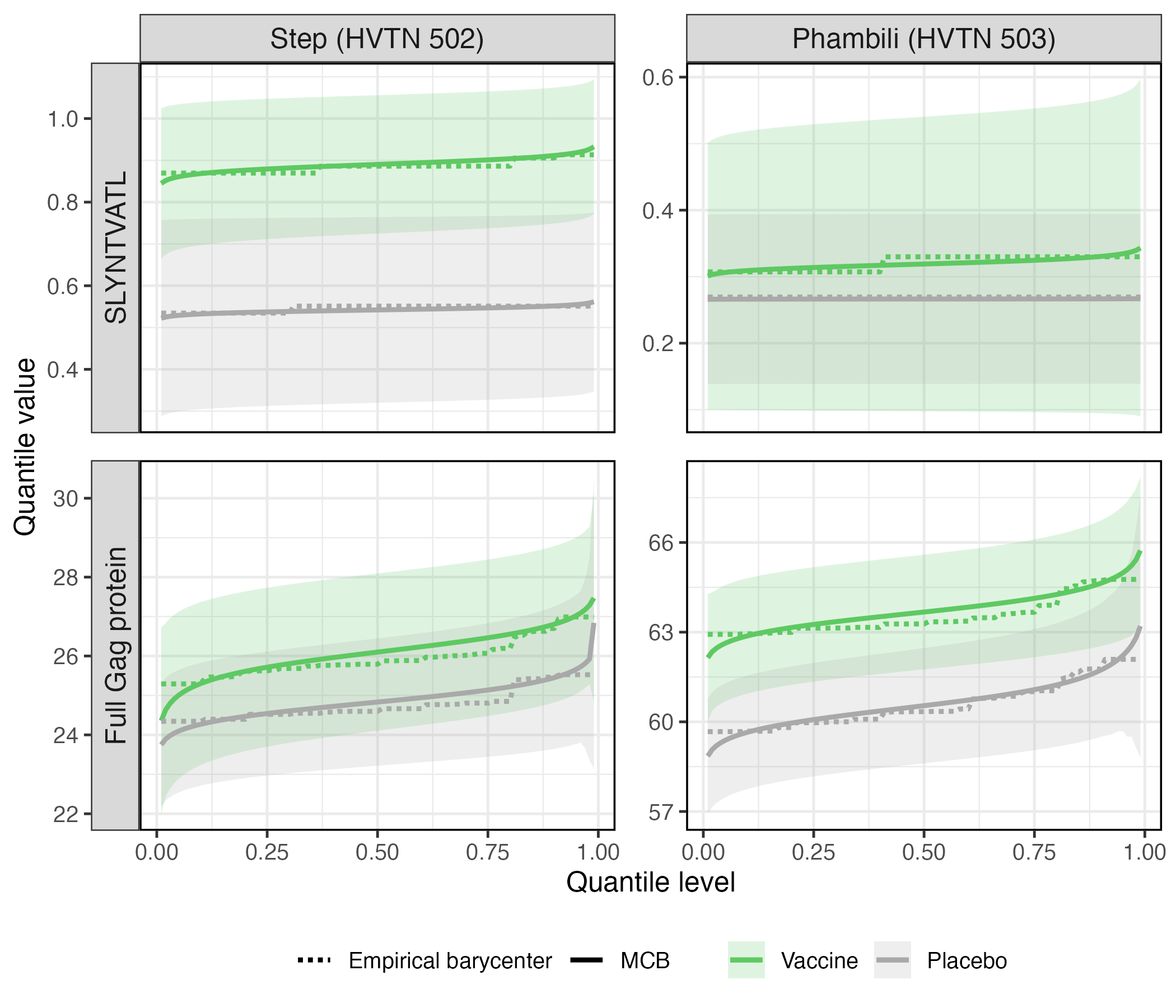}
    \caption{Estimated barycenter quantile functions for the SLYNTVATL epitope distance and full Gag protein physicochemical-weighted Hamming distance in the Step and Phambili trials, stratified by treatment arm. The MCB estimator along with its pointwise 95\% confidence intervals is displayed with a solid line, while the empirical barycenter is shown with a dotted line.}
    \label{fig:application}
\end{figure}

\section{Discussion} \label{sec:discussion}

To our knowledge, this is the first paper to propose a method that directly
corrects for sparse sampling bias in estimating the Wasserstein barycenter of
one-dimensional distributions. We approached the problem by first considering the seemingly more ambitious task of estimating the
distribution of the latent quantiles across units, from which the
barycenter quantile can be recovered as the mean. This distribution can be represented in terms of marginal laws
of unit-level CDF values, which arise as mixing distributions in
binomial mixture models. This representation leads to a natural plug-in
estimator: we first estimate the relevant binomial mixing distributions and
then plug these estimates into the representation formula. Our method
can accommodate any suitable estimator of the mixing distributions, allowing
it to draw on the extensive existing literature on binomial mixture
estimation. In simulations, the proposed approach performed well even when
the parametric submodels were misspecified.

Our results suggest that a distributional analysis can also be performed in hierarchical data settings, where the outcomes are typically treated as repeated measurements instead of distributions. Such data are traditionally analyzed using generalized estimating equations or mixed models, which focus on mean effects across groups or covariates \parencite{raudenbush2002hierarchical, gelman2007data}. 
Distributional methods could in principle be used in these settings because the repeated measurements within each unit can be viewed as i.i.d. draws from a latent unit-level distribution. However, distributional methods are rarely used in these settings because these methods typically assume that each unit-level distribution can be estimated accurately (i.e., requiring a large number of repeated measurements). 
Our framework relaxes this requirement and allows for inference on the Wasserstein barycenter of a population of unit-level distributions. Comparing barycenters across groups can reveal differences in features, such as tail behavior or spread, that may be missed by analyses focused on means. In Appendix \ref{appendix:additional-data-applications}, we illustrate the utility of our method in these settings by reanalyzing three classic data examples taken from hierarchical data analysis textbooks.

One caveat of the proposed approach is that structural assumptions are needed for estimation. These structural assumptions are of a semiparametric nature: we make parametric assumptions on the binomial mixing distributions, which does not fully specify the distribution on distributions $\Pi$.
Without these structural assumptions, the binomial mixing distributions underlying the marginal construction are not identifiable. Misspecification of these parametric models is a concern, and while the NPMLE-configured MCB estimator is an appealing alternative, its theoretical properties in the MCB procedure are not fully understood. An alternative direction for future work is to use the representation result to construct nonparametric bounds on barycenter quantiles, avoiding the need to impose structural assumptions on the mixing distributions.

Finally, additional theory is needed for estimation of the marginal mixing distributions across $x$. The current implementation estimates the mixing distribution pointwise over a grid of $x$ values and then enforces monotonicity through isotonic regression for each quantile level $\alpha  \in (0, 1)$ separately. This procedure is convenient because it can use available binomial mixture estimators, but is unlikely to be optimal because it does not impose the monotonicity constraint directly nor globally.
We are not aware of any existing theory and methods that could be directly applied to this problem.

\section{Code availability}
The implementation of the MCB estimator is provided in the \texttt{mcbarycenter} R package, available at \url{https://github.com/jpspeng/mcbarycenter}. A \href{https://jpspeng.github.io/mcbarycenter/vignettes/getting-started.html}{tutorial vignette} is also available. Please note that this repository is currently a work in progress and subject to active development.

\section{Acknowledgments}

We thank Zhenhua Lin, Thomas Richardson, and Michal Juraska for helpful discussions, and Craig Magaret for preparing the data used in the data application. The authors thank the participants and investigators of HVTN 502 and 503. Research reported in this publication was supported by the National Institute of Allergy And Infectious Diseases of the National Institutes of Health, award number R37AI054165 and U.S. Public Health Service Grant AI068635. The content is solely the responsibility of the authors and does not necessarily represent the official views of the National Institutes of Health.  

Florian Stijven was supported by the Agentschap Innoveren \& Ondernemen (VLAIO) and Johnson \& Johnson Innovative Medicine through a Baekeland Mandate [grant number HBC.2022.0145].

\newpage 

\printbibliography

\newpage

\appendix

% Appendix figure numbering: Figure A.1, A.2, ...
\setcounter{figure}{0}
\renewcommand{\thefigure}{A.\arabic{figure}}

% Optional: appendix table numbering: Table A.1, A.2, ...
\setcounter{table}{0}
\renewcommand{\thetable}{A.\arabic{table}}

\section*{Appendices}
% This command starts the "partial" TOC
\startcontents[appendices]
% This prints the partial TOC
\printcontents[appendices]{}{1}{}

\section{Proofs of main results}

\subsection{Proof of Proposition \ref{prop:barycenter_nonidentifiability}}

We first state and prove two lemmas that are needed for the proof of Proposition \ref{prop:barycenter_nonidentifiability}. Lemma \ref{lemma:truncated-moment-restriction} shows that, for any fixed $\alpha \in (0, 1)$ and $C \in \mathbb{N}$, there exist two distinct probability measures on $[0, 1]$ that have the same first $C$ moments but different CDF values at $\alpha-$. Lemma \ref{lemma:hilbert-matrix} is used in the proof of the Lemma \ref{lemma:truncated-moment-restriction}.

\begin{lemma}\label{lemma:truncated-moment-restriction}
    Let $G$ be a probability measure on $[0, 1]$ with Lebesgue density bounded away from zero. For any $\alpha \in (0, 1)$ and $C \in \mathbb{N}$, there exist two distinct probability measures $G_1$ and $G_2$ on $[0, 1]$ such that $\int_0^1 t^k\,dG_1(t) = \int_0^1 t^k\,dG_2(t)$ for all $k = 0, 1, \ldots, C$, but $\int_0^1 \mathbf{1}(t < \alpha)\,dG_1(t) \neq \int_0^1 \mathbf{1}(t < \alpha)\,dG_2(t)$.
\end{lemma}
\begin{proof}
    We will construct a signed measure $\nu^\star$ with bounded Lebesgue density such that $\int_0^1 t^k\,d\nu^\star(t) = 0$ for all $k = 0, 1, \ldots, C$, but $\int_0^1 \mathbf{1}(t < \alpha)\,d\nu^\star(t) \neq 0$. Then we can take $G_1 := G + \epsilon \nu^\star$ and $G_2 := G - \epsilon \nu^\star$ for some sufficiently small $\epsilon > 0$ (such that $G_1$ and $G_2$ are probability measures) to get the desired result. 

    Write $g$ for the Lebesgue density of $G$ and set $c := \inf_{x\in[0,1]} g(x) > 0$ (this infimum is positive by assumption).

    \textbf{Signed measure $\nu^\star$.}~~~Let $\nu$ be a signed measure on $[0, 1]$ with Lebesgue density a polynomial of degree $M \ge 2C + 2$ and coefficients $\boldsymbol{a} := (a_0, a_1, \ldots, a_{M})$. The moment conditions $\int_0^1 x^k \, d \nu(x) = 0$ for $k = 0, 1, \ldots, C$ translate to
    \begin{equation*}
        \int_0^1 x^k \sum_{j=0}^{M} a_j x^j\,dx = \sum_{j=0}^{M} a_j \int_0^1 x^{k+j}\,dx = \sum_{j=0}^{M} a_j \frac{1}{k+j+1} = 0 \quad \text{for } k = 0, 1, \ldots, C.
    \end{equation*}
    Let $H$ be the $(C + 1) \times (M + 1)$ Hilbert matrix with entries $H_{i, j} = \frac{1}{i + j - 1}$ for $i = 1, 2, \ldots, C + 1$ and $j = 1, 2, \ldots, M + 1$. The previous display can be written as $H \boldsymbol{a} = \boldsymbol{0}$.

    We further have that
    \begin{equation*}
        \int_0^\alpha d \nu(x) = \sum_{j=0}^{M} a_j \int_0^\alpha x^j\,dx = \sum_{j=0}^{M} a_j \frac{\alpha^{j+1}}{j+1}.
    \end{equation*}
    By Lemma \ref{lemma:hilbert-matrix} below, the vector $\boldsymbol{y} := (y_1,\ldots,y_{M+1})$ with $y_j:=\alpha^j/j$ does not lie in the row space of $H$ for $\alpha\in(0,1)$. Since $\ker(H)$ is the orthogonal complement of the row space of $H$, there exists a nonzero vector $\boldsymbol{a}^\star \in \ker(H)$ such that \(\boldsymbol{y}\cdot \boldsymbol{a}^\star \neq 0\). Equivalently,
    \begin{equation*}
        \sum_{j=0}^{M} a_j^\star \frac{\alpha^{j+1}}{j+1} \neq 0.
    \end{equation*}
    Taking $\nu^\star$ to have Lebesgue density $\sum_{j=0}^{M} a_j^\star x^j$ yields the desired signed measure.

    \textbf{Alternative probability measure.}~~~
    Let $b := \sup_{x\in[0,1]}|\sum_{j=0}^{M} a_j^\star x^j| < \infty$, which is finite because a polynomial is bounded on a compact set. Choose $\epsilon > 0$ such that $0 < \epsilon < \frac{c}{b}$.
    Then $G_1 := G + \epsilon \nu^\star$ and $G_2 := G - \epsilon \nu^\star$ are valid probability measures, and $\int_0^1 t^k\,dG_1(t) = \int_0^1 t^k\,dG_2(t)$ for all $k = 0, 1, \ldots, C$, but $\int_0^1 \mathbf{1}(t < \alpha)\,dG_1(t) \neq \int_0^1 \mathbf{1}(t < \alpha)\,dG_2(t)$.
\end{proof}

\begin{lemma}\label{lemma:hilbert-matrix}
    Choose $C, M \in \mathbb{N}$ such that $2C + 2 \le M + 1$ and let $y_j := \frac{\alpha^{j}}{j}$ for $j = 1, 2,  \ldots, M + 1$. If $\boldsymbol{y} := (y_1, y_2, \ldots, y_{M + 1})$ lies in the row space of the $(C + 1) \times (M + 1)$ Hilbert matrix $H$ with entries $H_{i, j} = \frac{1}{i + j - 1}$ for $i = 1, 2, \ldots, C + 1$ and $j = 1, 2, \ldots, M + 1$, then $\alpha = 1$ or $\alpha = 0$.
\end{lemma}
\begin{proof}
    Assume that $\boldsymbol{y}$ lies in the row space of $H$. Then there exists
    $\boldsymbol{b} := (b_1, b_2, \ldots, b_{C + 1})$ such that
    \begin{equation*}
        y_j = \sum_{i=1}^{C + 1} b_i H_{i, j} = \sum_{i=1}^{C + 1} b_i \frac{1}{i + j - 1} \quad \text{for } j = 1, 2, \ldots, M + 1.
    \end{equation*}
    Define
    \begin{equation*}
        D(j) := j(j+1)\cdots(j+C) = \prod_{k=0}^C (j+k).
    \end{equation*}
    For each $i = 1, 2, \ldots, C + 1$, the quotient $D(j)/(i+j-1)$ is a polynomial in $j$ of degree $C$, and hence
    \begin{equation*}
        D(j)\,y_j = \sum_{i=1}^{C+1} b_i\frac{D(j)}{i+j-1} \quad \text{for } j = 1, 2, \ldots, M + 1.
    \end{equation*}
    is a polynomial in $j$ of degree at most $C$. 
    Consequently, the $(C+1)$-th forward difference $\Delta^{C+1}[D(j)\,y_j]$ is zero for $j = 1, 2, \ldots, C + 1$. For $j > C + 1$, the forward difference may not be defined as it depends on values of $D(j)\,y_j$ for $j > 2C + 2$.

    On the other hand, by definition
    \begin{equation*}
        D(j)\,y_j = \frac{D(j)\,\alpha^{j}}{j} =: P(j)\,\alpha^{j} \quad \text{for } j = 1, 2, \ldots, M + 1,
    \end{equation*}
    where $P$ is a polynomial of degree $C$. 
    The $(C+1)$-th forward difference is as follows for $j = 1, 2, \ldots, C + 1$:
    \begin{equation*}
        \Delta^{C+1}[P(j)\,\alpha^{j}]
        = \sum_{r=0}^{C+1} (-1)^{C+1-r}\binom{C+1}{r} P(j+r)\,\alpha^{j+r}
        = \alpha^{j} \sum_{r=0}^{C+1} (-1)^{C+1-r}\binom{C+1}{r} \alpha^{r} P(j+r).
    \end{equation*}
    Assume $\alpha \neq 0$, divide both sides by $\alpha^{j}$, and denote
    \begin{equation*}
        S(j) := \sum_{r=0}^{C+1} (-1)^{C+1-r}\binom{C+1}{r} \alpha^{r} P(j+r).
    \end{equation*}
    From the previous paragraph, we have that $\Delta^{C+1}[P(j)\,\alpha^{j}] = \Delta^{C + 1}[D(j) \, y_j] = 0$ and thus $S(j) = 0$ for $j = 1, 2, \ldots, C + 1$. Since $S(j)$ is a polynomial of the same degree as $P$, which is $C$, it follows that $S$ is the zero polynomial and thus all coefficients of $S$ are zero. 
    From the definition of $P$ follows that the leading coefficient is $p_C := 1$. From the previous display follows that the leading coefficient of $S$ is
    \begin{equation*}
        p_C \sum_{r=0}^{C+1} (-1)^{C+1-r}\binom{C+1}{r} \alpha^{r} = p_C(\alpha-1)^{C+1},
    \end{equation*} 
    where the equality follows from the binomial theorem.
    This can only be zero for $\alpha = 1$. We have before excluded the case $\alpha = 0$ when we divided by $\alpha^{j}$, so the only other possibility is $\alpha = 0$.
\end{proof}

The proof of Proposition \ref{prop:barycenter_nonidentifiability} is given below.

\begin{proof}
We provide a simple example showing the failure of identification. Let $\mathcal I := \{0,1\}$. Then each latent distribution $\nu_i$ is Bernoulli, determined by the parameter $p_i := \nu_i(\{0\}) = F_{\nu_i}(0)$.
Its quantile function is
$$
F_{\nu_i}^{-}(\alpha)
=
\begin{cases}
0 & \text{if } p_i \ge \alpha,\\
1 & \text{if } p_i < \alpha.
\end{cases}
$$
Hence, for any fixed $\alpha \in (0,1)$, we have
$
q_\alpha^\star
=
\mathbb{E}_{\nu \sim \Pi} \bigl[F_{\nu}^{-}(\alpha)\bigr]
=
\mathbb{P}_{\nu \sim \Pi}(p < \alpha)$, where $p$ relates to $\nu$ as defined above.
Thus, it suffices to show that the distribution function $t \mapsto \mathbb{P}_{\nu \sim \Pi}(p \le t)$ is not identifiable.

Under this construction, the observed data reduce to the pair $(m_i,Y_i)$ where $Y_i:=\sum_{j=1}^{m_i}\mathbf 1\{X_{i,j}=0\}$. Conditional on $(m_i,p_i)$ we have $Y_i\mid m_i,p_i\sim\operatorname{Bin}(m_i,p_i)$, so the marginal law of $(Y_i,m_i)$ is a binomial mixture with mixing distribution $t\mapsto \mathbb{P}_{\nu \sim \Pi}(p \le t)$. If $\mathbb{P}_{m \sim \eta}(m \le C)=1$ for some $C < \infty$, this mixture model identifies at most the first $C$ moments of the mixing distribution \parencite{wood_binomial_1999}, but not the full mixing distribution itself. By Lemma \ref{lemma:truncated-moment-restriction}, there exist two distinct laws on $p_i$ (and hence two distinct laws $\Pi$ on distributions) that induce the same observed data distribution while yielding different values of $\mathbb{P}_{\nu \sim \Pi}(p<\alpha)$, and therefore different values of $q_\alpha^\star$.
\end{proof}

\begin{remark}
    Note that Proposition \ref{prop:barycenter_nonidentifiability} does not rule out the existence of observed-data distributions that can only be induced by a single law $\Pi$. For example, if $\mathbb{P}(X_{i,1}=0)=1$, then the observed-data law must be induced by the law $\Pi$ that places unit mass on the point-mass distribution concentrated at 0 (i.e., the degenerate distribution at 0). 
    The proposition shows that there exist observed-data distributions that can be induced by multiple distinct laws $\Pi$, and therefore the barycenter quantiles are not identifiable in general.
\end{remark}

\subsection{Proof of Proposition \ref{prop:barycenter_mean}}

\begin{proof}
For each unit $i$, define $k_i := \lceil m_i\alpha\rceil.$
Since $\hat{\nu}_i=\frac{1}{m_i}\sum_{j=1}^{m_i}\delta_{X_{i,j}},$
the empirical quantile of $\hat{\nu}_i$ at level $\alpha$ is the
$k_i$th order statistic of the observations from unit $i$. That is,
$$
F_{\hat{\nu}_i}^{-}(\alpha)
=
X_{i,(k_i)}^{(m_i)}.
$$
Therefore, by the definition of $\hat q_{\alpha}^{\,\mathrm{emp}}$,
\begin{align*}
\mathbb{E}\!\left[\hat q_{\alpha}^{\,\mathrm{emp}}\right]
&=
\mathbb{E}\!\left[
\frac{1}{n}\sum_{i=1}^n F_{\hat{\nu}_i}^{-}(\alpha)
\right] \\
&=
\frac{1}{n}\sum_{i=1}^n
\mathbb{E}\!\left[
F_{\hat{\nu}_i}^{-}(\alpha)
\right] \\
&=
\frac{1}{n}\sum_{i=1}^n
\mathbb{E}\!\left[
X_{i,(\lceil m_i\alpha\rceil)}^{(m_i)}
\right].
\end{align*}
Because the units are identically distributed under the sampling scheme
\eqref{eq:two-stage}, each term in the sum has the same expectation. Hence
\begin{align*}
\mathbb{E}\!\left[\hat q_{\alpha}^{\,\mathrm{emp}}\right]
&=
\mathbb{E}\!\left[
X_{1,(\lceil m_1\alpha\rceil)}^{(m_1)}
\right] \\
&=
\sum_{m'=1}^{\infty}
\eta(m')\,
\mathbb{E}\!\left[
X_{1,(\lceil m'\alpha\rceil)}^{(m')}
\,\middle|\,
m_1=m'
\right],
\end{align*}
where $\eta(m')=\mathbb{P}(m_1=m')$.

Now fix $m'\geq 1$ and set $k:=\lceil m'\alpha\rceil$. Conditional on
$m_1=m'$, the two-stage sampling scheme becomes
$$
\nu_1\sim\Pi,
\qquad
X_{1,1},\dots,X_{1,m'}\mid \nu_1
\overset{\mathrm{i.i.d.}}{\sim}\nu_1.
$$
Since $m_1\indep \nu_1$, conditioning on $m_1=m'$ does not change the law of
$\nu_1$. Therefore, for each fixed $m'$, we can use Lemma~A.1 in
\textcite{bigot_upper_2018}, which gives
$$
\mathbb{E}\!\left[
X_{1,(k)}^{(m')}
\,\middle|\,
m_1=m'
\right]
=
\mathbb{E}_{\bar{\nu}}\!\left[
X_{(k)}^{*(m')}
\right],
$$
where $X_{(1)}^{*(m')},\dots,X_{(m')}^{*(m')}$ are the order statistics of an
i.i.d.~sample of size $m'$ drawn from the barycenter $\bar{\nu}_{\Pi}$.

Substituting $k=\lceil m'\alpha\rceil$ into the previous identity and then
into the preceding sum yields
\begin{align*}
\mathbb{E}\!\left[\hat q_{\alpha}^{\,\mathrm{emp}}\right]
&=
\sum_{m'=1}^{\infty}
\eta(m')\,
\mathbb{E}_{\bar{\nu}}\!\left[
X_{(\lceil m'\alpha\rceil)}^{*(m')}
\right] \\
&=
\mathbb{E}_{m'\sim\eta}\!\left[
\mathbb{E}_{\bar{\nu}}\!\left[
X_{(\lceil m'\alpha\rceil)}^{*(m')}
\right]
\right].
\end{align*}
This proves the claim.
\end{proof}

\subsection{Proof of Theorem \ref{theorem:identification}}

\begin{proof}
    We will show that $G(x, \alpha-) = \mathbb{P}_{\nu \sim \Pi} \left\{ F_{\nu}^-(\alpha) > x \right\}$, from which directly follows $\mathbb{P}_{\nu \sim \Pi}\left\{ F_{\nu}^-(\alpha) \le x \right\} = 1 - G(x, \alpha-)$ as required. 
    % From this immediately follows that 
    % \begin{equation*}
    %     \int_{\mathcal{I}} x \, dF(x) = \int_{\mathcal{I}} x \, d (1 - G_x(\alpha)) = -\int_{\mathcal{I}} x \, d G_x(\alpha),
    % \end{equation*}
    % where $F(x) := \mathbb{P}\left\{ F_{\nu}^-(\alpha) \le x \right\}$, which would complete the proof.
    
    Consider the definition of $G(x, \alpha-)$:
    \begin{align*}
        G(x, \alpha-) 
        & = \lim_{t \uparrow \alpha} \mathbb{P}_{\nu \sim \Pi} \left\{ F_{\nu}(x) \le t  \right\} \\
        & = \mathbb{P}_{\nu \sim \Pi} \left\{ F_{\nu}(x) < \alpha \right\},
    \end{align*}
    where the second equality follows because the events
    $\{F_\nu(x)\le t\}$ increase to $\{F_\nu(x)<\alpha\}$ as $t\uparrow\alpha$.
    We now show that $F_{\nu}(x) < \alpha \iff F_{\nu}^-(\alpha) > x$, which will imply that $G(x, \alpha-) = \mathbb{P}_{\nu \sim \Pi} \left\{ F_{\nu}^-(\alpha) > x \right\}$ and, therefore, will complete the proof.

    \textbf{$F_{\nu}(x) < \alpha \implies F_{\nu}^-(\alpha) > x$}. \quad
    Start by assuming that $F_{\nu}(x) < \alpha$. 
    Because $x \mapsto F_{\nu}(x)$ is a cumulative distribution function, it is right-continuous, which implies that $\lim_{t \downarrow x}F_{\nu}(t) = F_{\nu}(x) < \alpha$ and, hence, there exists a $x^* > x$ such that $F_{\nu}(x^*) < \alpha$.
    Because $F_{\nu}$ is a distribution function, it is also non-decreasing, which implies that $x^* \le \inf \left\{ t: F_{\nu}(t) \ge \alpha \right\} = F_{\nu}^-(\alpha)$. Hence, $x < x^* \le F_{\nu}^-(\alpha)$ and thus, as required, $x < F_{\nu}^-(\alpha)$.

    \textbf{$F_{\nu}(x) < \alpha \impliedby F_{\nu}^-(\alpha) > x$}. \quad
    Start by assuming that $F_{\nu}^-(\alpha) > x$. Now also assume that $F_{\nu}(x) \ge \alpha$, which implies that $F_{\nu}^-(\alpha) = \inf \left\{ t: F_{\nu}(t) \ge \alpha \right\} \le x$. This is a contradiction; hence, we must have that $F_{\nu}(x) < \alpha$.
\end{proof}

\subsection{Proof of Proposition \ref{prop:empirical_barycenter}}

\begin{proof}
Define $
C_i(x) := \mathbf{1}\{F_{\hat\nu_i}(x)\ge \alpha\}.$
By the same equivalence used in the proof of Theorem~\ref{theorem:identification}, $C_i(x) = \mathbf{1}\{F_{\hat\nu_i}^{-}(\alpha)\le x\}$.
Moreover, $1-\hat G_n^{\mathrm{emp}}(x,\alpha-) = \frac{1}{n}\sum_{i=1}^n C_i(x)$ and therefore
\begin{align*}
\hat q_{\alpha}^{\,\mathrm{MCB}}
&=
\int_{\mathcal I} x \, d\left(1-\hat G_n^{\mathrm{emp}}(x,\alpha-)\right) \\
&=
\frac{1}{n}\sum_{i=1}^n \int_{\mathcal I} x \, dC_i(x) \\
&=
\frac{1}{n}\sum_{i=1}^n F_{\hat\nu_i}^{-}(\alpha) = \hat q_{\alpha}^{\,\mathrm{emp}},
\end{align*}
where the third equality follows because $C_i$ is the distribution function of a point mass at $F_{\hat\nu_i}^{-}(\alpha)$.
\end{proof}

\subsection{Proof of Proposition \ref{prop:mixing_nonid}}

\begin{proof}
For any $m\in\{1,\dots,C\}$ and $y\in\{0,\dots,m\}$,
$$
\mathbb{P} \bigl(Y_i(x)=y\mid m_i=m\bigr)
=
\int_0^1
\binom{m}{y} t^y (1-t)^{m-y}\,dG(x,t).
$$
For fixed $(m,y)$, the function $t\mapsto t^y(1-t)^{m-y}$ is a polynomial in $t$ of degree at most $m$, and hence of degree at most $C$. Therefore each conditional count probability depends only on the first $C$ moments of $G(x,\cdot)$. Since distinct probability distributions on $[0,1]$ can share the same first $C$ moments, the mixing distribution is not identifiable.
\end{proof}

\subsection{Proof of Theorem \ref{theorem:consistency-fixed-m-asymptotics}}

\begin{proof}
Let $A_n$ denote the event that $x\mapsto \hat H_{n,\alpha}(x)$ is a cumulative
distribution function on $\mathcal I$. By Condition~\ref{cond:distribution_function},
$\mathbb P(A_n)\to 1$. On $A_n$, let $Q_{n,\alpha}$ denote the probability
measure with distribution function $\hat H_{n,\alpha}$, and let $Q_\alpha$
denote the probability measure with distribution function $H_\alpha$.

By Condition~\ref{cond:pointwise_consistency} and
Lemma~\ref{lemma:from-pointwise-to-weak}, $Q_{n,\alpha}$ converges weakly in probability to $Q_\alpha$:
$$
Q_{n,\alpha} \overset{P}{\rightsquigarrow} Q_\alpha.
$$
If $\mathcal I$ is bounded, then $x\mapsto x$ is a bounded continuous function
on $\mathcal I$, and therefore
$$
\int_{\mathcal I} x\,d\hat H_{n,\alpha}(x)
=
\int_{\mathcal I} x\,dQ_{n,\alpha}(x)
\overset{P}{\to}
\int_{\mathcal I} x\,dQ_\alpha(x)
=
\int_{\mathcal I} x\,dH_\alpha(x).
$$
If $\mathcal I$ is unbounded but Condition~\ref{cond:uniform_integrability}
holds, then the same conclusion follows from
Lemma~\ref{lemma:convergence-of-moments}. Hence
$$
\hat q_{\alpha}^{\,\mathrm{MCB}}
=
\int_{\mathcal I} x\,d\hat H_{n,\alpha}(x)
\overset{P}{\to}
\int_{\mathcal I} x\,dH_\alpha(x)
=
q_\alpha^\star.
$$
This proves the claim.
\end{proof}

\subsection{Proof of Corollary \ref{corollary:consistency-in-L2}}

\begin{proof}
By the assumptions of the corollary, with probability tending to one,
$t\mapsto \hat q_t^{\,\mathrm{MCB}}$ is a quantile function. Let $\hat\nu_n$
denote the corresponding probability measure. By
Theorem~\ref{theorem:consistency-fixed-m-asymptotics}, for every continuity
point $t$ of the true barycenter quantile function
$t\mapsto F_{\bar\nu_\Pi}^{-}(t)$,
$$
F_{\hat\nu_n}^{-}(t)
=
\hat q_t^{\,\mathrm{MCB}}
=
q_t^\star + o_P(1)
=
F_{\bar\nu_\Pi}^{-}(t)+o_P(1).
$$
By Lemma~\ref{lemma:pointwise-to-weak-quantile}, this implies that $\hat{\nu}_n$ converges weakly in probability to $\bar\nu_\Pi$:
$$
\hat\nu_n \overset{P}{\rightsquigarrow} \bar\nu_\Pi.
$$
As explained in the text before Lemma \ref{lemma:from-pointwise-to-weak} and in the proof of Lemma \ref{lemma:from-pointwise-to-weak}, weak convergence in probability (as defined in Definition \ref{def:weak-convergence-in-probability}) is equivalent to the following: for every subsequence $(n_k)_{k\ge 1}$, there exists a further subsequence $(m_l)_{l\ge 1}$ such that
\begin{equation*}
    P \left( \omega: \hat \nu^\omega_{m_l} \rightsquigarrow \bar\nu_\Pi  \right) = 1.
\end{equation*}

By \textcite[Theorem 6.9]{villani2008optimal}, convergence in the
$d_{W_2}$ topology is equivalent to weak convergence together with convergence
of second moments. We next argue along subsequences.
Hence, it remains to show that, for every subsequence $(n_k)_{k\ge 1}$, there exists a further subsequence $(m_l)_{l\ge 1}$ such that
$$
\int_{\mathcal I} x^2\,d\hat\nu_{m_l}(x)
\to
\int_{\mathcal I} x^2\,d\bar\nu_\Pi(x).
$$

If $\mathcal I$ is bounded, then $x\mapsto x^2$ is a bounded continuous
function on $\mathcal I$. Therefore, weak convergence in probability implies
$$
\int_{\mathcal I} x^2\,d\hat\nu_n(x)
= 
\int_{\mathcal I} x^2\,d\bar\nu_\Pi(x) + o_P(1),
$$
which implies the desired conclusion along subsequences.
If $\mathcal I$ is unbounded, the assumed asymptotic uniform integrability
with respect to $f:x\mapsto x^2$, together with
Lemma~\ref{lemma:convergence-of-moments}, gives the same conclusion.

The results above imply that, for every subsequence $(n_k)_{k\ge 1}$, there exists a further subsequence $(m_l)_{l\ge 1}$ such that
\begin{equation*}
    P \left( \omega: \hat\nu_{m_l}^\omega \rightsquigarrow \bar\nu_\Pi  \text{ and } \int_{\mathcal I} x^2\,d\hat\nu_{m_l}^\omega(x) \to \int_{\mathcal I} x^2\,d\bar\nu_\Pi(x)  \right) = 1.
\end{equation*}
\textcite[Theorem 6.9]{villani2008optimal} now implies that, for every subsequence $(n_k)_{k\ge 1}$, there exists a further subsequence $(m_l)_{l\ge 1}$ such that
$$
d_{W_2}\!\left(\hat\nu_{m_l},\bar\nu_\Pi\right) \to 0,
$$
which immediately implies that $d_{W_2}\!\left(\hat\nu_n,\bar\nu_\Pi\right)\overset{P}{\to}0$.
\end{proof}

\subsection{Proof of Theorem \ref{theorem:asymptotic-distribution}}

\begin{proof}
    Without loss of generality, assume that $\mathcal{I} = \mathbb{R}$.
    For any distribution function $F$ for which $\int_{\mathcal{I}} |x| \, dF < \infty$, we have that 
    \begin{equation*}
        \int_{\mathcal{I}} x \, dF = \int_0^{+\infty} 1 - F(t) \, dt - \int_{0}^{+\infty} F(-t) \, dt.
    \end{equation*}
    From the above display and Condition~\ref{cond:AL-cdf-integrability} it follows that, with probability tending to one,
    \begin{equation*}
        \int_{\mathcal{I}} x \, d \hat H_{n,\alpha} - \int_{\mathcal{I}} x \, dH_{\alpha}
        =
        - \int_0^{+\infty} \{\hat H_{n,\alpha}(x) - H_\alpha(x)\} \, dx
        - \int_0^{+\infty} \{\hat H_{n,\alpha}(-x) - H_{\alpha}(-x)\} \, dx.
    \end{equation*}
    Multiplying both sides by $\sqrt{n}$ and using Condition~\ref{cond:AL-linear-expansion}, we obtain
    \begin{align*}
        \sqrt{n} \left( \int x \, d \hat H_{n,\alpha} - \int x \, dH_{\alpha}  \right)
        & =
        - \int_0^{+\infty}
        \left\{
        \frac{1}{\sqrt{n}}\sum_{i = 1}^n \psi_{\alpha}(O_i; x)
        + r_{\alpha,n}(x)
        \right\} \, dx  \\
        &\quad
        - \int_0^{+\infty}
        \left\{
        \frac{1}{\sqrt{n}}\sum_{i = 1}^n \psi_{\alpha}(O_i; -x)
        + r_{\alpha,n}(-x)
        \right\} \, dx \\
        & =
        - \frac{1}{\sqrt{n}} \int_{-\infty}^{+\infty}
        \sum_{i = 1}^n \psi_{\alpha}(O_i; x) \, dx
        - \int_{-\infty}^{+\infty} r_{\alpha,n}(x) \, dx \\
        & =
        - \frac{1}{\sqrt{n}} \int_{-\infty}^{+\infty}
        \sum_{i = 1}^n \psi_{\alpha}(O_i; x) \, dx
        + o_{P}(1).
    \end{align*}
    The second equality follows from the integrability conditions on
    $x \mapsto r_{\alpha,n}(x)$ and $x \mapsto \psi_{\alpha}(O; x)$ where $x \mapsto \psi_{\alpha}(O; x)$ is integrable with probability one because $E\left[ \int_{-\infty}^{+\infty} | \psi_{\alpha}(O; x) | \, dx \right] < \infty$ by Condition~\ref{cond:AL-linear-expansion}. The
    final equality follows because
    $\int_{\mathcal I} |r_{\alpha,n}(x)|\,dx=o_P(1)$.
    Moving the sum outside the integral gives
    \begin{align*}
        \sqrt{n} \left( \int x \, d \hat H_{n,\alpha} - \int x \, dH_{\alpha}  \right)
        &=
        \frac{1}{\sqrt n}\sum_{i=1}^n
        \left\{
        -\int_{-\infty}^{+\infty}\psi_\alpha(O_i;x)\,dx
        \right\}
        +o_P(1) \\
        &=
        \frac{1}{\sqrt n}\sum_{i=1}^n \tilde\psi_\alpha(O_i)+o_P(1).
    \end{align*}

    Condition~\ref{cond:AL-linear-expansion} also requires that $E\left[ \int_{-\infty}^{+\infty} | \psi_{\alpha}(O; x) | \, dx \right] < \infty$, which allows us to apply Fubini's theorem to obtain $E\left[ \tilde\psi_\alpha(O) \right] =  0$ from $E[\psi_\alpha(O;x)] = 0$.
    The asymptotic normality now follows from Condition~\ref{cond:AL-integrated-if} and the central limit theorem.
\end{proof}

\subsection{Proof of Corollary \ref{corollary:uniform-convergence}}

\begin{proof}
    This is a standard application of uniform asymptotic linearity and weak convergence of the empirical process indexed by the influence functions.

    \textbf{Uniform asymptotic linearity.}~~~
    By condition (i) of the corollary, Theorem~\ref{theorem:asymptotic-distribution} applies uniformly to all $\alpha \in [\alpha_{\min},\alpha_{\max}]$. Specifically, the following asymptotic linearity representation holds for all $\alpha \in [\alpha_{\min},\alpha_{\max}]$ with probability tending to one:
    \begin{equation*}
        \sqrt{n}\bigl(\hat q_\alpha^{\,\mathrm{MCB}} - q_\alpha^\star\bigr) = \frac{1}{\sqrt{n}} \sum_{i=1}^n \tilde\psi_\alpha(O_i) + \tilde r_{\alpha,n}
    \end{equation*}
    where $\tilde r_{\alpha,n} := -\int_{\mathcal{I}} r_{\alpha,n}(x) \, dx$ with $\sup_{\alpha \in [\alpha_{\min},\alpha_{\max}]} \left| \tilde r_{\alpha,n} \right| = o_P(1)$ by condition (ii).
    
    \textbf{Weak convergence of empirical process.}~~~
    By condition (iii), the class $\{\tilde\psi_\alpha : \alpha \in [\alpha_{\min},\alpha_{\max}]\}$ is P-Donsker, which implies that
    \begin{equation*}
        \mathbb{G}_n(\tilde\psi_\alpha) := \frac{1}{\sqrt{n}}\sum_{i=1}^n \left[\tilde\psi_\alpha(O_i) - \mathbb{E}[\tilde\psi_\alpha(O)]\right]
    \end{equation*}
    converges weakly in $\ell^\infty([\alpha_{\min},\alpha_{\max}])$ to a centered Gaussian process $\mathbb{G}$ with covariance function
    \begin{equation*}
        \Gamma(\alpha,\alpha') := \mathbb{E}\left[\mathbb{G}(\tilde\psi_\alpha)\mathbb{G}(\tilde\psi_{\alpha'})\right] = \mathbb{E}\left[\tilde\psi_\alpha(O)\,\tilde\psi_{\alpha'}(O)\right] - \mathbb{E}\left[\tilde\psi_\alpha(O)\right]\mathbb{E}\left[\tilde\psi_{\alpha'}(O)\right] = \mathbb{E}\left[\tilde\psi_\alpha(O)\,\tilde\psi_{\alpha'}(O)\right],
    \end{equation*}
    where the last equality follows from the mean-zero property of the influence function $\tilde\psi_\alpha$.
    
    \textbf{Conclusion.}~~~
    From the above two steps, we have that,
    \begin{equation*}
        \sqrt{n}\bigl(\hat q_\alpha^{\,\mathrm{MCB}} - q_\alpha^\star\bigr) = \mathbb{G}_n(\tilde\psi_\alpha) + o_P(1),
    \end{equation*}
    where (i) $\mathbb{G}_n(\tilde\psi_\alpha)_{\alpha \in [\alpha_{\min},\alpha_{\max}]}$
    converges weakly in $\ell^\infty([\alpha_{\min},\alpha_{\max}])$ to the Gaussian process $\mathbb{G}$ with covariance $\Gamma(\alpha,\alpha')$ and (ii) the remainder term is $o_P(1)$ uniformly in $\alpha \in [\alpha_{\min},\alpha_{\max}]$. 
    Hence, by Slutsky's theorem \parencite[Theorem 18.10(iv)]{van2000asymptotic}, we have that $\sqrt{n}\bigl(\hat q_\alpha^{\,\mathrm{MCB}} - q_\alpha^\star\bigr)$ converges weakly in $\ell^\infty([\alpha_{\min},\alpha_{\max}])$ to the same Gaussian process $\mathbb{G}$.
\end{proof}

\section{Additional data applications}\label{appendix:additional-data-applications}

\subsection{Public versus Catholic schools}

\paragraph{Data:} We use the High School and Beyond dataset as presented in the \textcite{raudenbush2002hierarchical} hierarchical data analysis textbook. The data consist of a nationally representative sample of U.S. public and Catholic schools from the 1982 survey. The dataset includes 7,185 students nested within 160 schools, consisting of 90 public schools and 70 Catholic schools. For each school, the data include an indicator of school type (public or Catholic) as well as the school’s mean student socioeconomic status. At the student level, the dataset includes math achievement scores and an indicator of whether the student identifies as a minority. Figure~\ref{fig:school_boxplot} summarizes the school-level sample sizes by school type.

\begin{figure}[!htb]
    \centering
    \includegraphics[width=\textwidth]{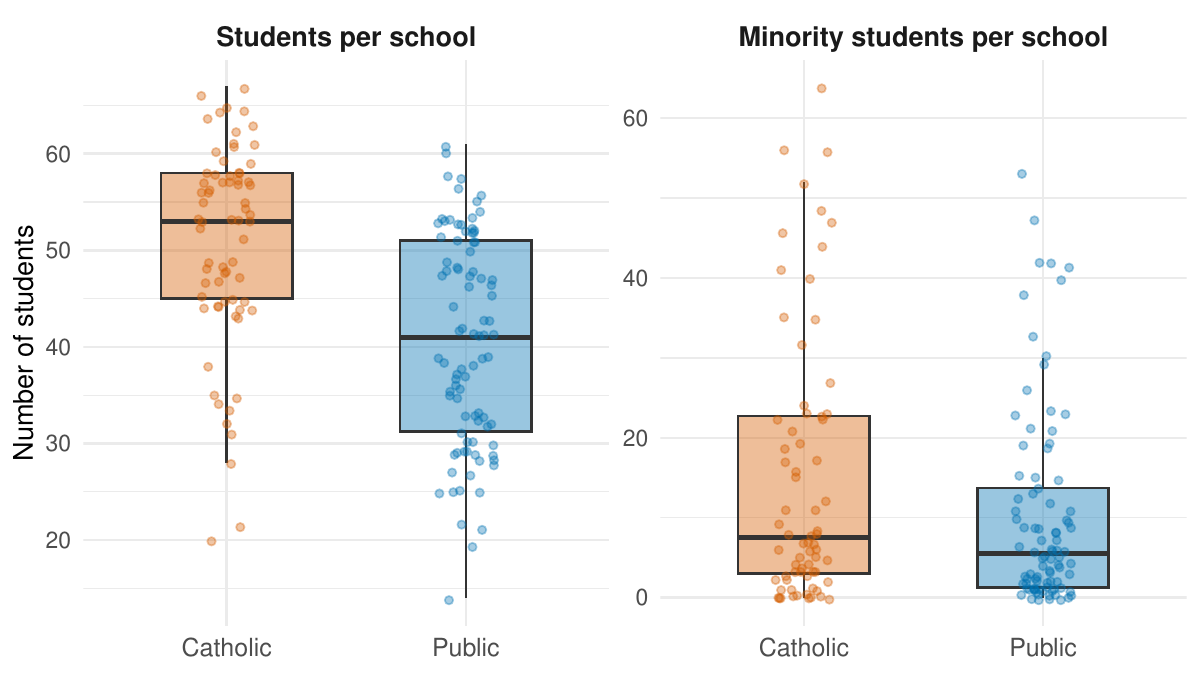}
    \caption{Distribution of the number of students sampled per school, overall and among minority students, stratified by school type. Points represent individual schools.}
    \label{fig:school_boxplot}
\end{figure}

\paragraph{Question of interest:} How does the barycenter distribution of math scores for public schools compare with that for Catholic schools?

\paragraph{Analysis:} Using the MCB estimator, we estimate the Wasserstein barycenter of the distribution of math scores for Catholic schools and compare it with that of public schools. Differences between the two groups are assessed by examining the difference in their Wasserstein barycenter quantile functions. We then repeat this analysis subsetting for minority students. Our primary MCB specification uses splines to estimate the mixing distribution ($df = 5$); for comparison, we also report results from the Beta-configured MCB estimator and the empirical barycenter.

\begin{figure}[!htb]
    \centering
    \includegraphics[width=\textwidth]{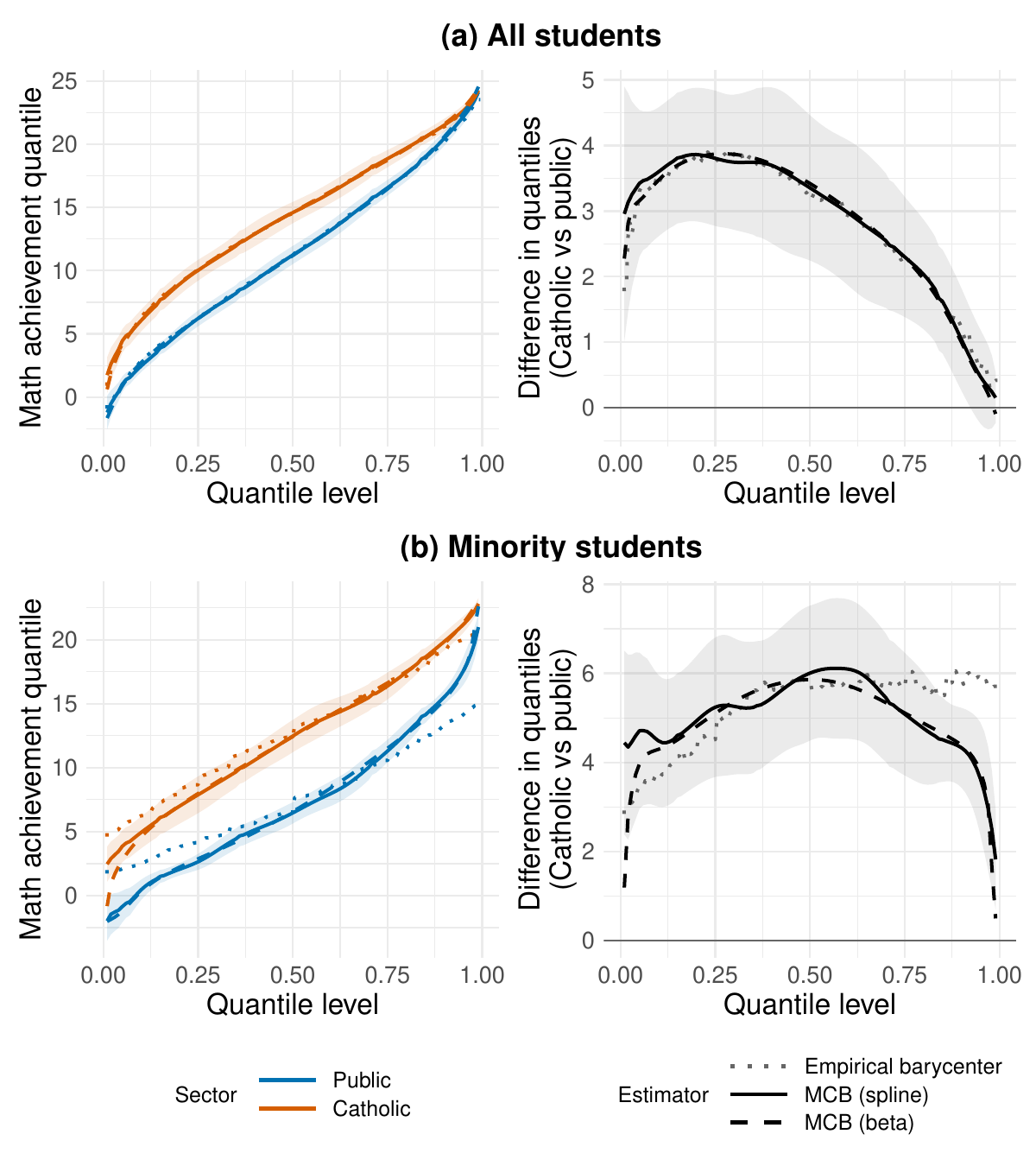}
    \caption{ Estimated Wasserstein barycenters of math achievement score distributions for public and Catholic schools, shown for all students and for minority students. Left panels display the estimated quantile functions by school type, while right panels show the difference in quantile functions comparing Catholic to public schools. Results are shown for the MCB estimators (spline and Beta configurations) and the empirical barycenter. The shaded areas represent pointwise confidence intervals based on the spline-configured MCB estimator.}
    \label{fig:school_results} 
\end{figure}

\paragraph{Results:} Results are shown in Figure \ref{fig:school_results}. For the full student sample, the estimated barycenter quantile function for Catholic schools lies above that for public schools for all quantile levels. 
The quantile difference plot  suggests that the Catholic-public gap is concentrated among lower-to-middle achieving students: for example, the 25th percentile of the Catholic-school distribution is roughly 4 points higher than the corresponding percentile for public schools, whereas the difference is minimal among students near the 99th percentile.
The MCB spline and Beta estimators produce very similar results, and both align closely with the empirical barycenter.

Among minority students, the Catholic-public difference is larger across nearly all quantile levels. Using the spline-based MCB estimator, the differences are most pronounced in the middle of the distribution, reaching their largest values around the 0.50-0.60 quantiles before tapering off toward the upper tail. The spline and Beta MCB estimators generally closely agree with one another, but both differ from the empirical barycenter, particularly in the upper tail where the MCB estimates for the difference-in-quantiles are lower.

\subsection{Ticks on grouse}

\paragraph{Data:} To study parasite aggregation among hosts, \textcite{elston2001analysis} collected data on sheep tick burdens among red grouse chicks sampled in the Glas Choille study area in Scotland. The sample includes 403 chicks from 118 broods observed across three years: 1995, 1996, and 1997, with a median of 3 chicks per brood (range: 1-10). For each chick, the response variable is the number of sheep ticks counted, with counts taken from the chick's head. Brood-level information includes the year of observation and altitude.

\paragraph{Question of interest:} Does the barycenter distribution of tick counts differ between broods captured at lower versus higher altitudes?

\paragraph{Analysis:} We first divide the broods into two groups based on altitude, classifying broods as low-altitude or high-altitude depending on whether their capture altitude is below or above the sample mean. Using the MCB estimator, we estimate the barycenter of the distribution of tick counts for broods in the high-altitude group and compare it with that of broods in the low-altitude group. Differences between the two groups are assessed by examining the difference in their estimated quantile functions.
Our primary MCB specification uses splines to estimate the mixing distribution ($df = 5$); for comparison, we also report results from the Beta-configured MCB estimator and the empirical barycenter.

\begin{figure}[!htb]
    \centering
    \includegraphics[width=\textwidth]{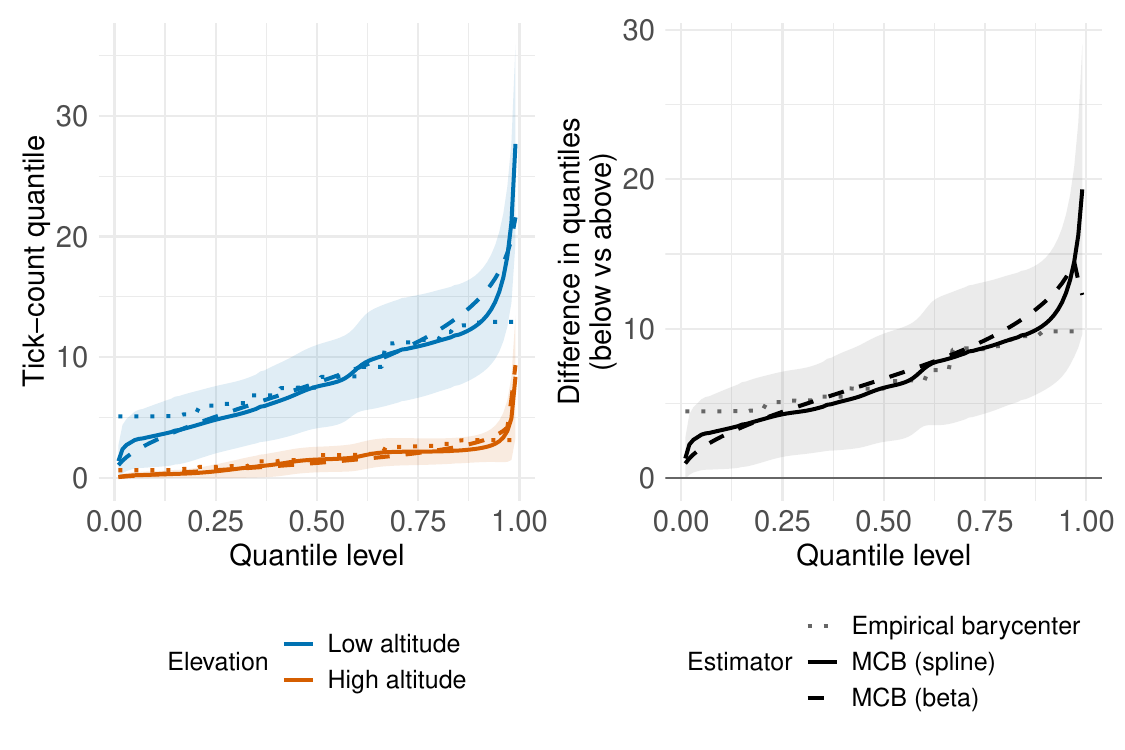}
    \caption{Estimated Wasserstein barycenters of tick-count distributions for broods captured at low and high altitudes. The left panel displays the estimated quantile functions by altitude group, while the right panel shows the difference in quantile functions comparing low-altitude broods to high-altitude broods. Results are shown for the MCB estimators (spline and Beta configurations) and the empirical barycenter. The shaded areas represent pointwise confidence intervals based on the spline-configured MCB estimator.}
    \label{fig:grouse_ticks}
\end{figure}

\paragraph{Results:} Results are shown in Figure \ref{fig:grouse_ticks}. The estimated barycenters suggest substantially higher tick burdens among broods captured at lower altitudes compared with those captured at higher altitudes. The quantile difference plot shows that this difference increases across quantile levels, indicating that the altitude-related difference is most pronounced among chicks with the highest tick counts. The spline and Beta MCB estimators closely agree with one another. Compared with the empirical barycenter, both MCB estimators produce smaller estimated differences in the lower tail and larger estimated differences in the upper tail.

\subsection{Radon exposure in Minnesota}

\paragraph{Background} To study variation in indoor radon concentrations, \textcite{price1996bayesian} analyzed radon measurements from households across Minnesota counties. The radon dataset consists of 919 households from 85 counties, with a median of four household measurements per county (range: 1--100) \parencite{gelman2007data}. For each household, the response variable is the log radon measurement in log picoCuries per liter, and household-level information includes whether the measurement was taken in the basement or on the first floor. County-level information includes average soil uranium content.

\paragraph{Question of interest:} Is the barycenter of the household-level radon distribution within a county  associated with the county's soil uranium concentration?

\paragraph{Analysis:} We estimate the conditional barycenter of radon measurements given a county's mean log-uranium level using the kernel-regression conditional barycenter estimator described in Example~\ref{sec:conditional_bary}.
Our primary MCB specification uses splines to estimate the mixing distribution ($df = 5$); for comparison, we also report results from the empirical barycenter. 

\paragraph{Results:} Figure~\ref{fig:radon} displays the estimated conditional barycenter of household log-radon distributions at five county log-uranium levels. The estimated barycenter shifts upward as county log-uranium increases, indicating higher household radon levels in counties with greater soil uranium content.

\begin{figure}[!htb]
    \centering
    \includegraphics[width=\textwidth]{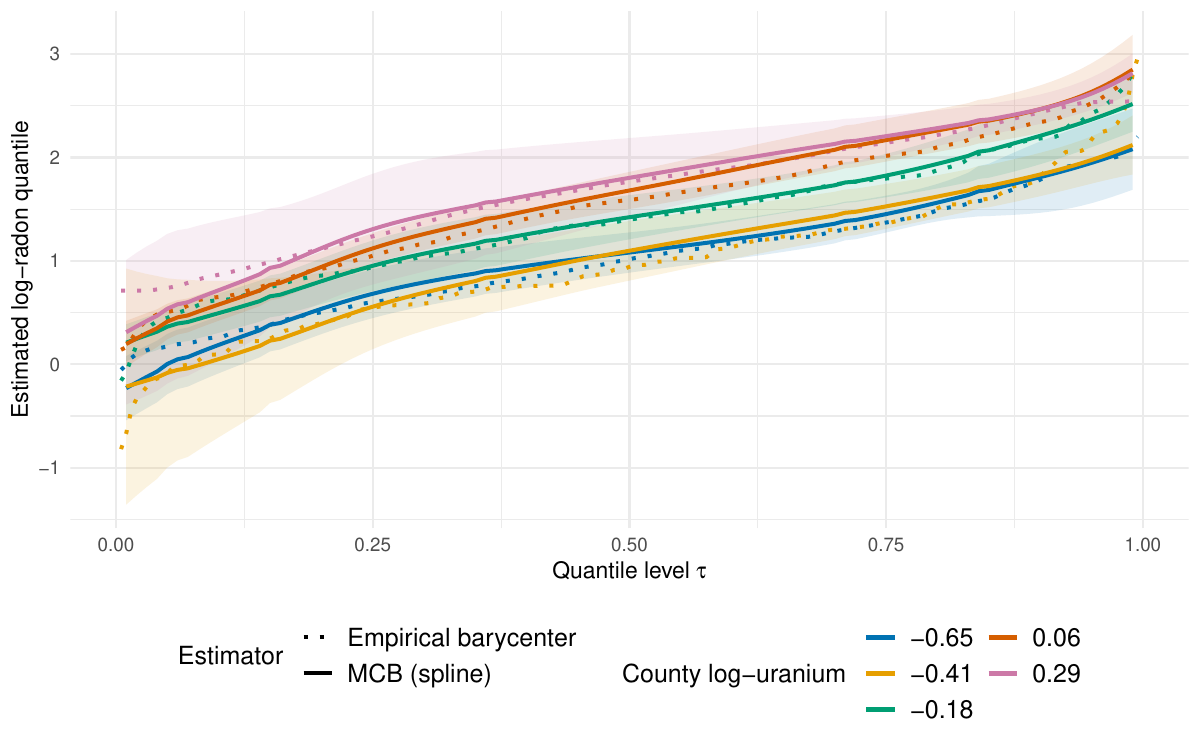}
    \caption{Estimated conditional barycenters of household log-radon distributions across
county log-uranium levels. The curves show estimated log-radon quantile
functions at five fixed values of county log-uranium. Results are shown for
the empirical barycenter and for the MCB estimator with a spline configuration. Shaded bands denote pointwise confidence intervals for the
spline-configured MCB estimator.
    }
    \label{fig:radon}
\end{figure}

\section{Technical tools for weak convergence}

This section develops the technical tools used in the proofs of the consistency results in Theorem~\ref{theorem:consistency-fixed-m-asymptotics} and Corollary~\ref{corollary:consistency-in-L2}.
Since our estimators are naturally defined through estimated distribution functions or quantile functions, we need a way to translate pointwise convergence statements into convergence of the induced random probability measures. In section \ref{appendix:deterministic_sequences}, we review deterministic results showing that, for monotone functions such as CDFs and quantile functions, convergence on a countable dense set of continuity points is enough to imply convergence at all continuity points. Then, in section \ref{appendix:weak-convergence-in-prob}, we extend these ideas to random probability measures and define weak convergence in probability. The main consequence is that pointwise convergence in probability of either the CDFs or the quantile functions at continuity points is equivalent to weak convergence in probability of the associated measures. Finally, we state a uniform integrability condition under which weak convergence in probability also implies convergence of unbounded
moments.

\subsection{Deterministic sequences of distribution functions and quantile functions}\label{appendix:deterministic_sequences}

\begin{lemma}\label{lemma:weak-convergence-countable-pointwise}
    Let $(F_n)_{n \ge 1}$ be a deterministic sequence of cumulative distribution functions on $\mathbb{R}$ and $F$ a distribution function on $\mathbb{R}$. 
    There exists a countable set of continuity points of $x \mapsto F(x)$ that is dense in $\mathbb{R}$, which we denote by $\mathcal{D} \subset \mathbb{R}$.
    If $F_n(x) \to F(x)$ for every $x \in \mathcal{D}$, then $F_n(x) \to F(x)$ for every continuity point of $x \mapsto F(x)$. 
\end{lemma}
\begin{proof}
    \textbf{Part 1: Existence of $\mathcal{D}$.}~~
    Note that the set of discontinuity points of $x \mapsto F(x)$, denoted by $C$, is at most countable because $F$ is non-decreasing càdlàg and bounded. If we remove these points from $\mathbb{R}$, the remaining set of points $\mathbb{R} \setminus C$ is dense in $\mathbb{R}$.\footnote{Assume for contradiction that $\mathbb{R} \setminus C$ is not dense in $\mathbb{R}$. Then there exists a $x \in C$ and $\delta > 0$ such that $B_{\delta}(x) \cap (\mathbb{R} \setminus C) = \emptyset$. This implies that $B_{\delta}(x) \subset C$, which is a contradiction because $C$ is countable and thus cannot contain a non-empty open interval.}
    A countable subset of $\mathbb{R} \setminus C$ can now be constructed as follows:
    Let $\mathcal{B}$ be the collection of open intervals with rational endpoints. This set contains a countable number of intervals, denoted by $(B_n)_{n \ge 1}$. Since $\mathbb{R} \setminus C$ is dense, $B_n \cap (\mathbb{R} \setminus C)$ is non-empty for every $n \in \mathbb{N}$. For every $n \in \mathbb{N}$, let $d_n$ be an element of $B_n \cap (\mathbb{R} \setminus C)$. 
    The set $\{ d_n : n \ge 1\}$ is countable and dense in $\mathbb{R}$.

    \textbf{Part 2: Pointwise convergence at all continuity points.}~~
    Let $x$ be a continuity point of $x \mapsto F(x)$.
    Hence, for every $\epsilon > 0$, there exists a $\delta > 0$ such that $|F(y) - F(x)| < \epsilon$ for every $y \in B_{\delta}(x)$. 
    Now choose $l < x$ and $u > x$ such that $u, l \in \mathcal{D} \cap B_{\delta}(x)$ for $\mathcal{D}$ the set mentioned in the first part of the lemma.

    Because $F_n$ is non-decreasing for all $n \ge 1$, we have that 
    \begin{equation*}
        F_n(l) \le F_n(x) \le F_n(u).
    \end{equation*}
    We consequently have that
    \begin{equation*}
        F(l) = \lim_{n \to \infty} F_n(l) \le \liminf_{n \to \infty} F_n(x) \le \limsup_{n \to \infty} F_n(x) \le \lim_{n \to \infty} F_n(u) = F(u),
    \end{equation*}
    where the equalities follow from $u, l \in \mathcal{D}$. By definition of $u, l$, we have that $|F(l) - F(u)| < \epsilon$. 
    Hence, from the above display follows that
    \begin{equation*}
        \left| \limsup_{n \to \infty} F_n(x) - \liminf_{n \to \infty} F_n(x)  \right| \le \left| F(u) - F(l) \right| < \epsilon.
    \end{equation*}
    Because $\epsilon$ was arbitrary, this implies that $\liminf_{n \to \infty} F_n(x) = \limsup_{n \to \infty} F_n(x) = \lim_{n \to \infty} F_n(x) = F(x)$.
\end{proof}

\begin{lemma}\label{lemma:weak-convergence-countable-pointwise-quantile-function}
    Let $(Q_n)_{n \ge 1}$ be a deterministic sequence of quantile functions and $Q$ a quantile function. 
    There exists a countable set of continuity points of $t \mapsto Q(t)$ that is dense in $(0, 1)$, which we denote by $\mathcal{D} \subset (0, 1)$.
    If $Q_n(t) \to Q(t)$ for every $t \in \mathcal{D}$, then this implies that $Q_n(t) \to Q(t)$ for every continuity point of $t \mapsto Q(t)$. 
\end{lemma}
\begin{proof}
    The proof has the same structure as the proof of Lemma \ref{lemma:weak-convergence-countable-pointwise}.
    
    \textbf{Part 1: Existence of $\mathcal{D}$.}~~
    The set of discontinuity points of $t \mapsto Q(t)$, denoted by $C$, is countable because $Q$ is a non-decreasing function on the unit interval. If we remove these points from $(0, 1)$, the remaining set of points $(0, 1) \setminus C$ is dense in $(0, 1)$.
    A countable subset of $(0, 1) \setminus C$ can now be constructed as follows:
    Let $\mathcal{B}$ be the collection of open intervals with rational endpoints in $(0, 1)$. This set contains a countable number of intervals, denoted by $(B_n)_{n \ge 1}$. Since $(0, 1) \setminus C$ is dense, $B_n \cap ((0, 1) \setminus C)$ is non-empty for every $n \in \mathbb{N}$. For every $n \in \mathbb{N}$, let $d_n$ be an element of $B_n \cap ((0, 1) \setminus C)$. 
    The set $\{ d_n : n \ge 1\}$ is countable and dense in $(0, 1)$.

    \textbf{Part 2: Pointwise convergence at all continuity points}~~
    Let $t$ be a continuity point of $t \mapsto Q(t)$.
    Hence, for every $\epsilon > 0$, there exists a $\delta > 0$ such that $|Q(y) - Q(t)| < \epsilon$ for every $y \in B_{\delta}(t)$. 
    Now choose $l < t$ and $u > t$ such that $u, l \in \mathcal{D} \cap B_{\delta}(t)$ for $\mathcal{D}$ the set mentioned in the first part of the lemma. 

    Because $Q_n$ is non-decreasing for all $n \ge 1$, we have that 
    \begin{equation*}
        Q_n(l) \le Q_n(t) \le Q_n(u).
    \end{equation*}
    We consequently have that
    \begin{equation*}
        Q(l) = \lim_{n \to \infty} Q_n(l) \le \liminf_{n \to \infty} Q_n(t) \le \limsup_{n \to \infty} Q_n(t) \le \lim_{n \to \infty} Q_n(u) = Q(u),
    \end{equation*}
    where the equalities follow from $u, l \in \mathcal{D}$. By the definition of $u, l$, we have that $|Q(l) - Q(u)| < \epsilon$. Hence, from the above display follows that
    \begin{equation*}
        \left| \limsup_{n \to \infty} Q_n(t) - \liminf_{n \to \infty} Q_n(t)  \right| \le \left| Q(u) - Q(l) \right| < \epsilon.
    \end{equation*}
    Because $\epsilon$ was arbitrary, this implies that $\liminf_{n \to \infty} Q_n(t) = \limsup_{n \to \infty} Q_n(t) = \lim_{n \to \infty} Q_n(t) = Q(t)$.
\end{proof}

\subsection{Random probability measures and weak convergence in probability}\label{appendix:weak-convergence-in-prob}

In what follows, we will focus on weak convergence in probability. 
We start by summarizing weak convergence of a deterministic sequence of probability measures and then extend this to weak convergence in probability of a sequence of random probability measures. The latter will largely follow \textcite{schmon_large-sample_2021}; although, we present some novel technical results\footnote{We have not found these statements and proofs in the literature. However, we do not claim that we are the first to present these results, especially since these results can be proven along the same lines as some well-known results.}.

\subsubsection{Weak convergence of deterministic sequences of measures}

Let $\left(\mu_n\right)_{n \ge 1}$ be a deterministic sequence of probability measures on $\mathcal{I} \subset \mathbb{R}$, and let $F_{\mu_n}$ be the corresponding cumulative distribution functions (CDFs). One definition of weak convergence to $\mu$, with CDF $F_{\mu}$, is as follows:
\begin{equation}
    F_{\mu_n}(x) \to F_{\mu}(x) \quad \text{for all continuity points of } x \mapsto F_{\mu}(x).
\end{equation}
By the Portmanteau Lemma, this is equivalent to the following definition:
\begin{equation*}
    \mathbb{E}_{\mu_n}f(X) \to \mathbb{E}_{\mu} f(X) \quad \forall \; f \in C_b(\mathcal{I}),
\end{equation*}
where $C_b(\mathcal{I})$ is the set of bounded continuous functions on $\mathcal{I}$.

\subsubsection{Random probability measures}

Before we define weak convergence in probability of a sequence of random probability measures, we more rigorously define random probability measures. We restate the definition of \textcite{schmon_large-sample_2021}, but specialized to random probability measures on subsets of the real line, denoted by $\mathcal{I} \subset \mathbb{R}$ for $\mathcal{B}(\mathcal{I})$ the corresponding Borel $\sigma$-algebra. Further let $\mathcal{P}(\mathcal{I})$ be the set of probability measures on the measurable space $(\mathcal{I}, \mathcal{B}(\mathcal{I}))$.

\begin{definition}[Random probability measure]
    A random probability  measure is a map $\nu: \Omega \times \mathcal{B}(\mathcal{I}) \to [0, 1]$ such that for every $B \in \mathcal{B}(\mathcal{I})$ the map $\omega \mapsto \nu(\omega, B) = \nu^{\omega}(B)$ is measurable while $\nu(\omega, \cdot) \in \mathcal{P}(\mathcal{I})$ for almost every $\omega \in \Omega$.
\end{definition}

By \textcite[Proposition 3.3]{crauel_random_2002}, we have that $\omega \mapsto \mathbb{E}_{\nu^{\omega}}f := \int_{\mathcal{I}} f \, d\nu^{\omega}$ for $f \in C_b(\mathcal{I})$ is measurable if $\nu$ is a random probability measure.

\subsubsection{Weak convergence in probability}

Weak convergence in probability is defined below for random probability measures on (a subset of) the real line. This can easily be extended to more general probability measures, but this is not required for our purposes.

\begin{definition}[Weak convergence in probability]\label{def:weak-convergence-in-probability}
    Let $\mathcal{P}(\mathcal{I})$ be the set of all probability measures on $(\mathcal{I}, \mathcal{B}(\mathcal{I}))$ and let $(\nu_n)_{n \ge 1}$ be a sequence of random probability measures in $\mathcal{P}(\mathcal{I})$. 
    This sequence converges weakly in probability to $\mu \in \mathcal{P}(\mathcal{I})$ if 
    \begin{equation*}
    \mathbb{E}_{\nu_n}f(X) = \mathbb{E}_{\mu} f(x) + o_P(1) \quad \forall \; f \in C_b(\mathcal{I}),
\end{equation*}
where $C_b(\mathcal{I})$ is the set of bounded continuous functions on $\mathcal{I}$.
\end{definition}
Note that, in the above definition, $\mathbb{E}_{\nu_n} f(X)$ is a random variable because $\omega \mapsto \mathbb{E}_{\nu_n^{\omega}} f(X)$ is measurable (which follows from $\nu_n$ being a random probability measure, see previous subsection).

Definition \ref{def:weak-convergence-in-probability} is different from the definition of weak convergence in probability in \textcite{schmon_large-sample_2021}. However, \textcite[Theorem 1 in Supplementary Material]{schmon_large-sample_2021} show that their definition is equivalent to ours.
The following lemma is a useful extension of the Portmanteau lemma, which complements \textcite[Theorem 1 in Supplementary Material]{schmon_large-sample_2021}.
\begin{lemma}\label{lemma:from-pointwise-to-weak}
    Let $(\nu_n)_{n \ge 1}$ be a sequence of random probability measures and $(F_{\nu_n})_{n \ge 1}$ the corresponding random cumulative distribution functions.
    Further let $\mu$ and $F_{\mu}$ be a probability measure and corresponding cumulative distribution function on $\mathcal{I}$.
    Then, $F_{\nu_n}(x) = F_{\mu}(x) + o_P(1)$ for every $x$ a continuity point of $x \mapsto F_{\mu}(x)$ if and only if $\nu_n$ converges weakly in probability to $\mu$.
\end{lemma}
\begin{proof}
    \textbf{$\Rightarrow$.}~~The difficulty of the proof in this direction is that the $o_P(1)$ term in $F_{\nu_n}(x) = F_{\mu}(x) + o_P(1)$ only holds pointwise. In other words, $\|F_{\nu_n} - F_{\mu}\|_{\infty}$ may not converge in probability to zero. 
    This is addressed by considering a countable dense subset of continuity points and by arguing along subsequences. 
    
    By definition of random probability measures, $\omega \mapsto F_{\nu^{\omega}_n}(x) := \nu^{\omega}_n((-\infty, x])$ is measurable for any $x \in \mathcal{I}$. Hence, $F_{\nu_n}(x)$ is a random variable for any $x \in \mathcal{I}$. 
    For ease of notation, we further let $F_n^{\omega} := F_{\nu_n^{\omega}}$ and $F_n := F_{\nu_n}$.
    
    By Lemma \ref{lemma:weak-convergence-countable-pointwise}, there exists a countable subset of continuity points of $x \mapsto F_{\mu}(x)$ that is dense in $\mathbb{R}$, which we denote by $\mathcal{D}$. Order these points in the sequence $(x_n)_{n \ge 1}$. Now consider an arbitrary subsequence $(n_k)_{k \ge 1}$ of $(n)_{n \ge 1}$.
    
    Let $x_1$ be the first element of $(x_n)_{n \ge 1}$. Because $x_1 \in \mathcal{D}$, it is a continuity point of $x \mapsto F_{\mu}(x)$ and, consequently, we have that $F_n(x_1) = F(x_1) + o_P(1)$. Hence, the subsequence $(n_k)_{k \ge 1}$ has a further subsequence, denoted by $(n^1_l)_{l \ge 1}$, such that $F_{n_l^1}(x_1) \rightarrow F(x_1)$ $P$-a.s. Or equivalently, $P ( A_1 ) = 1$ for $A_1 := \left\{\omega \in \Omega: F_{n_l^1}^\omega(x_1) \to F_{\mu}(x_1) \right\}$. 
    %Note that the size of $A^1$ can only increase if we replace $(n_l^1)_{l \ge 1}$ with a further subsequence.

    We can repeat the above construction for every element of $(x_n)_{n \ge 1}$, where the subsequence $(n^2_l)_{l \ge 1}$ is a subsequence of $(n^1_l)_{l \ge 1}$, the subsequence $(n^3_l)_{l \ge 1}$ is a subsequence of $(n^2_l)_{l \ge 1}$, and so on.
    We now define the subsequence $(m_l)_{l \ge 1}$ element-wise as $m_l = n^l_l$. 
    By construction, for any $k \in \mathbb{N}$ and $\omega \in \Omega$, we have that $F_{n^k_l}^{\omega}(x_k) \to F_{\mu}(x_k) \implies F_{m_l}^{\omega}(x_k) \to F_{\mu}(x_k)$ because $(m_l)_{l \ge 1}$ is eventually a subsequence of $(n^k_l)_{l \ge 1}$. This implies that 
    \begin{equation*}
        \left\{ \omega \in \Omega: F^{\omega}_{n^k_l}(x_k) \to F_{\mu}(x_k) \right\} =: A_k \subset \tilde A_k := \left\{ \omega \in \Omega: F^{\omega}_{m_l}(x_k) \to F_{\mu}(x_k) \right\}.
    \end{equation*}
    Because $P(A_k) = 1$ for every $k \in \mathbb{N}$, we have that $P\left( \bigcap_{k \ge 1} A_k \right) = 1$ and, because $A_k \subset \tilde A_k$ for every $k$, also that $P \left( \bigcap_{k \ge 1} \tilde A_k \right) = 1$.

    For any $\omega \in \bigcap_{k \ge 1} \tilde A_k$, we have that $F^{\omega}_{m_l}(x) \to F_{\mu}(x)$ for all $x \in \mathcal{D}$. By Lemma \ref{lemma:weak-convergence-countable-pointwise}, this implies that $F^{\omega}_{m_l}(x) \to F_{\mu}(x)$ for all continuity points. By the Portmanteau lemma, this implies that $F^{\omega}_{m_l}$ converges weakly to $F_{\mu}$, or equivalently, that $\nu_{m_l}^{\omega}$ converges weakly to $\mu$. 

    Note that the initial subsequence $(n_k)_{k \ge 1}$ was arbitrary. Hence, we have that for every subsequence $(n_k)_{k \ge 1}$, there exists a further subsequence $(m_l)_{l \ge 1}$ such that 
    \begin{equation*}
        P \left( \omega: \nu^{\omega}_{m_l} \rightsquigarrow \mu \right) = 1.
    \end{equation*}
    This is the definition of weak convergence in probability given in \textcite[Definition 3]{schmon_large-sample_2021}\footnote{As explained in the text before Lemma \ref{lemma:from-pointwise-to-weak}, this is equivalent to our Definition \ref{def:weak-convergence-in-probability}.}.

    \textbf{$\Leftarrow$.}~~
    By the definition of weak convergence in probability, for every subsequence $(n_k)_{k \ge 1}$, there exists a further subsequence $(n_l)_{l \ge 1}$ such that 
    \begin{equation*}
        P \left( \omega: \nu^{\omega}_{n_l} \rightsquigarrow \mu \right) = 1.
    \end{equation*}
    For all $\omega$ in the above set, we have that $F_{\nu_{n_l}^{\omega}}(x) \to F_{\mu}(x)$ for every $x$ a continuity point. Because the initial subsequence was arbitrary, this implies that $F_{\nu_n}(x) = F_{\mu}(x) + o_P(1)$ for every $x$ a continuity point.
\end{proof}

The following lemma resembles the previous one, but considers pointwise convergence of the quantile functions (instead of the distribution functions) of a sequence of random probability measures.
The structure of the proof below also resembles the proof of Lemma \ref{lemma:from-pointwise-to-weak}. 

\begin{lemma}[Relation between quantile and cdf estimators]\label{lemma:pointwise-to-weak-quantile}
    Let $(\nu_n)_{n \ge 1}$ be a sequence of random probability measures and $\mu$ a probability measure on $\mathcal{I} \subset \mathbb{R}$. Denote the corresponding quantile functions by $F^-_{\nu_n}$ and $F^-_{\mu}$, respectively. 
    We then have that $F_{\nu_n}^-(t) = F_{\mu}^-(t) + o_P(1)$ for every $t$ a continuity point of $t \mapsto F^-_{\mu}(t)$ if and only if $\nu_n$ converges weakly in probability to $\mu$.
\end{lemma}
\begin{proof}
    \textbf{$\Rightarrow$.}~~
     For ease of notation, we further let $Q_n^{\omega} := F^-_{\nu_n^{\omega}}$, $F_n := F_{\nu_n}$, and $Q = F^-_{\mu}$.
    
    By Lemma \ref{lemma:weak-convergence-countable-pointwise-quantile-function}, there exists a countable set of continuity points of $t \mapsto Q(t)$ that is dense in $(0, 1)$, which we denote by $\mathcal{D}$. Order these points in the sequence $(t_n)_{n \ge 1}$. Now consider an arbitrary subsequence $(n_k)_{k \ge 1}$ of $(n)_{n \ge 1}$.
    
    Let $t_1$ be the first element of $(t_n)_{n \ge 1}$. Because $t_1 \in \mathcal{D}$, it is a continuity point of $t \mapsto Q(t)$ and, consequently, we have that $Q_n(t_1) = Q(t_1) + o_P(1)$. Hence, the subsequence $(n_k)_{k \ge 1}$ has a further subsequence, denoted by $(n^1_l)_{l \ge 1}$, such that $Q_{n_l^1}(t_1) \rightarrow Q(t_1)$ $P$-a.s. Or equivalently, $P ( A_1 ) = 1$ for $A_1 := \left\{\omega \in \Omega: Q_{n_l}^\omega(t_1) \to Q(t_1) \right\}$. 
    %Note that the size of $A^1$ can only increase if we replace $(n_l^1)_{l \ge 1}$ with a further subsequence.

    We can repeat the above construction for every element of $(t_n)_{n \ge 1}$, where the subsequence $(n^2_l)_{l \ge 1}$ is a subsequence of $(n^1_l)_{l \ge 1}$, the subsequence $(n^3_l)_{l \ge 1}$ is a subsequence of $(n^2_l)_{l \ge 1}$, and so on.
    We now define the subsequence $(m_l)_{l \ge 1}$ element-wise as $m_l = n^l_l$. 
    By construction, for any $k \in \mathbb{N}$ and $\omega \in \Omega$, we have that $Q_{n^k_l}^{\omega}(t_k) \to Q(t_k) \implies Q_{m_l}^{\omega}(t_k) \to Q(t_k)$ because $(m_l)_{l \ge 1}$ is eventually a subsequence of $(n^k_l)_{l \ge 1}$. This implies that 
    \begin{equation*}
        \left\{ \omega \in \Omega: Q^{\omega}_{n^k_l}(t_k) \to Q(t_k) \right\} =: A_k \subset \tilde A_k := \left\{ \omega \in \Omega: Q^{\omega}_{m_l}(t_k) \to Q(t_k) \right\}.
    \end{equation*}
    Because $P(A_k) = 1$ for every $k \in \mathbb{N}$, we have that $P\left( \bigcap_{k \ge 1} A_k \right) = 1$ and, because $A_k \subset \tilde A_k$ for every $k$, also that $P \left( \bigcap_{k \ge 1} \tilde A_k \right) = 1$.

    For any $\omega \in \bigcap_{k \ge 1} \tilde A_k$, we have that $Q^{\omega}_{m_l} \to Q(t)$ for all $t \in \mathcal{D}$. By Lemma \ref{lemma:weak-convergence-countable-pointwise-quantile-function}, this implies that $Q^{\omega}_{m_l}(t) \to Q(t)$ for all continuity points. By \textcite[Lemma 21.1]{van2000asymptotic}, this implies that $F^{\omega}_{m_l}$ converges weakly to $F_{\mu}$, or equivalently, that $\nu_{m_l}^{\omega}$ converges weakly to $\mu$. 

    Note that the initial subsequence $(n_k)_{k \ge 1}$ was arbitrary. Hence, we have that for every subsequence $(n_k)_{k \ge 1}$, there exists a further subsequence $(m_l)_{l \ge 1}$ such that 
    \begin{equation*}
        P \left( \omega: \nu^{\omega}_{m_l} \rightsquigarrow \mu \right) = 1.
    \end{equation*}
    This is the definition of weak convergence in probability given in \textcite[Definition 3]{schmon_large-sample_2021}\footnote{As explained in the text before Lemma \ref{lemma:from-pointwise-to-weak}, this is equivalent to our Definition \ref{def:weak-convergence-in-probability}.}.

    \textbf{$\Leftarrow$.}~~
    By the definition of weak convergence in probability, for every subsequence $(n_k)_{k \ge 1}$, there exists a further subsequence $(n_l)_{l \ge 1}$ such that 
    \begin{equation*}
        P \left( \omega: \nu^{\omega}_{n_l} \rightsquigarrow \mu \right) = 1.
    \end{equation*}
    For all $\omega$ in the above set, \textcite[Lemma 21.2]{van2000asymptotic} implies that $F^-_{\nu^{\omega}_{n_l}}(t) \to F^-_{\mu}(t)$ for every $t$ a continuity point of $t \mapsto F^-_{\mu}(t)$. 
    Because the initial subsequence was arbitrary, this implies that $F^-_{\nu_n}(t) = F^-_{\mu}(t) + o_P(1)$ for every $t$ a continuity point of $t \mapsto F^-_{\mu}(t)$.    
\end{proof}

\begin{definition}[Asymptotic uniform integrability]\label{def:asymptotic-uniform-integrability}
    The sequence of random probability measures $\nu_n^{\omega}$ on $\mathbb{R}^k$ is asymptotically uniformly integrable with respect to $f: \mathbb{R}^k \to \mathbb{R}$ if for every subsequence, there exists a further subsequence $\nu^{\omega}_{m_n}$ such that for all $\omega \in A$ with $P(A) = 1$:
    \begin{equation*}
        \lim_{M \to \infty} \lim_{n \to \infty} \mathbb{E}_{\nu_{m_n}^{\omega}} \left[ |f(X)| \cdot \mathbf{1}\left\{ |f(X)| > M \right\} \right] = 0.
    \end{equation*}
\end{definition}

\begin{lemma}[Convergence of moments]\label{lemma:convergence-of-moments}
    Let $\nu_n$ be a sequence of random probability measures on $\mathbb{R}^k$ converging weakly in probability to $\mu$. Let $f: \mathbb{R}^k \to \mathbb{R}$ be a continuous function where $|\mathbb{E}_{\mu}f(X)| < \infty$. Then $\mathbb{E}_{\nu_n} f(X) \overset{P}{\to} \mathbb{E}_{\mu}f(X)$ if the sequence of random probability measures $\nu_n$ is asymptotically uniformly integrable with respect to $f$.
\end{lemma}
\begin{proof}
    Without loss of generality, assume that $f(X)$ is non-negative; otherwise, one can follow the same arguments below for the negative and positive parts of $f(X)$ separately.
    Let $(\nu_{l_n}^{\omega})_{n \ge 1}$ be a subsequence of $(\nu_n^{\omega})_{n \ge 1}$. By the definition of asymptotic uniform integrability, there exists a further subsequence $(\nu^\omega_{m_n})_{n \ge 1}$, and set $A$ with $P(A)$ such that for all $\omega \in A$:
    \begin{equation*}
            \lim_{M \to \infty} \lim_{n \to \infty} \mathbb{E}_{\nu_{m_n}^{\omega}} \left[ |f(X)| \cdot \mathbf{1}\left\{ |f(X)| > M \right\} \right] = 0.
    \end{equation*}
    By the definition of weak convergence in probability\footnote{Our Definition \ref{def:weak-convergence-in-probability} is not formulated in terms of subsequences, but this definition is equivalent to \textcite[Definition 3 in Supplementary Material]{schmon_large-sample_2021} which is formulated in terms of subsequences.}, there is a further subsequence $\nu_{m'_n}^{\omega}$ such that $\mathbb{E}_{\nu_{m'_n}^{\omega}} g(X) \to \mathbb{E}_{\mu} g(X)$ for all $g \in C_b(\mathbb{R}^k)$ and $\omega \in A'$ with $P(A') = 1$. 
    To simplify notation, we denote the latter subsequence  and set simply by $\nu_{m_n}$ and $A$, respectively.

    The proof is complete if we can show that $\mathbb{E}_{\nu_{m_n}^{\omega}}f(X) - \mathbb{E}_{\mu}f(X) \to 0$ for all $\omega \in A$ because the initial subsequence $(\nu_{l_n}^{\omega})_{n \ge 1}$ was arbitrary. Indeed, this would imply that for every subsequence of $(\nu_n)_{n \ge 1}$, there exists a further subsequence $(\nu_{m_n})_{n \ge 1}$ where $\mathbb{E}_{\nu_{m_n}} f(X) \overset{a.e.}{\to} \mathbb{E}_{\mu} f(X)$. This in turn implies that $\mathbb{E}_{\nu_{n}} f(X) \overset{P}{\to} \mathbb{E}_{\mu} f(X)$. 
    
    We now show that $\mathbb{E}_{\nu_{m_n}^{\omega}}f(X) - \mathbb{E}_{\mu}f(X) \to 0$ for all $\omega \in A$. Select $\omega \in A$ and apply the triangle inequality as follows:
    \begin{align*}
        \left| \mathbb{E}_{\nu_{m_n}^{\omega}}f(X) - \mathbb{E}_{\mu}f(X) \right|  \le & \left| \mathbb{E}_{\nu_{m_n}^{\omega}}f(X) - \mathbb{E}_{\nu_{m_n}^{\omega}}f(X) \wedge M \right| \\
        & + \left| \mathbb{E}_{\nu_{m_n}^{\omega}}f(X) \wedge M - \mathbb{E}_{\mu}f(X) \wedge M \right| \\
        & + \left| \mathbb{E}_{\mu}f(X) \wedge M - \mathbb{E}_{\mu}f(X) \right|.
    \end{align*}
    For the expressions above, we first take the limit with respect to $n \to \infty$ and, second, with respect to $M \to \infty$. The left-hand side then simply is $\lim_{n \to \infty} \left| \mathbb{E}_{\nu_{m_n}^{\omega}}f(X) - \mathbb{E}_{\mu}f(X) \right|$ (as there is no dependence on $M$). We now consider the three expressions of the right-hand side separately. 

    We can upper bound the first expression of the right-hand side as follows:
    \begin{align*}
        \left| \mathbb{E}_{\nu_{m_n}^{\omega}}f(X) - \mathbb{E}_{\nu_{m_n}^{\omega}}f(X) \wedge M \right| & = \big| \mathbb{E}_{\nu_{m_n}^{\omega}}f(X) \mathbf{1}(f(X) > M) + \mathbb{E}_{\nu_{m_n}^{\omega}} f(X) \mathbf{1}(f(X) \le M) \\
        & \quad - \mathbb{E}_{\nu_{m_n}^{\omega}} M \mathbf{1}(f(X) > M) - \mathbb{E}_{\nu_{m_n}^{\omega}}f(X) \mathbf{1}(f(X) \le M)   \big| \\
        & = \left| \mathbb{E}_{\nu_{m_n}^{\omega}}f(X) \mathbf{1}(f(X) > M) -  \mathbb{E}_{\nu_{m_n}^{\omega}} M \mathbf{1}(f(X) > M) \right| \\
        & \le \left| \mathbb{E}_{\nu_{m_n}^{\omega}}f(X) \mathbf{1}(f(X) > M)  \right|
    \end{align*}
    The inequality follows from the fact that $f(X)$ is non-negative and that $f(X) \mathbf{1}(f(X) > M) > M \mathbf{1}(f(X) > M)$.
    From asymptotic uniform integrability follows $$\lim_{M \to \infty} \lim_{n \to \infty} \mathbb{E}_{\nu_{m_n}^{\omega}}f(X) \mathbf{1}(f(X) > M) = 0.$$ Hence, we have that $\lim_{M \to \infty} \lim_{n \to \infty} \left| \mathbb{E}_{\nu_{m_n}^{\omega}}f(X) - \mathbb{E}_{\nu_{m_n}^{\omega}}f(X) \wedge M \right| = 0$.

    Because $\nu^{\omega}_{m_n}$ converges weakly to $\mu$ and $x \mapsto f(x) \wedge M$ is bounded and continuous, we have that $\lim_{n \to \infty}  \mathbb{E}_{\nu_{m_n}^{\omega}}f(X) \wedge M = \mathbb{E}_{\mu}f(X) \wedge M$. Hence, it follows that $$\lim_{M \to \infty} \lim_{n \to \infty}  \left| \mathbb{E}_{\nu_{m_n}^{\omega}}f(X) \wedge M - \mathbb{E}_{\mu}f(X) \wedge M \right| = 0.$$ 

    The third expression is upper bounded by $\mathbb{E}_{\mu}f(X) \mathbf{1}(f(X) > M)$ (following the same argument as before). Because $|\mathbb{E}_{\mu} f(X)| < \infty$, we have that $\lim_{M \to \infty} f(X) \mathbf{1}(f(X) > M) = 0$.

    The above results now imply that $\mathbb{E}_{\nu_{m_n}^{\omega}}f(X) - \mathbb{E}_{\mu}f(X) \to 0$ for all $\omega \in A$, which completes the proof.
\end{proof}

\section{Asymptotic theory for parametric binomial mixture estimators}\label{appendix:parametric-binomial-mixtures}

% \textcolor{blue}{Can you review this high-level overview of this section?}

This section develops the asymptotic theory for the parametric binomial mixture estimators used in the main text. For each threshold \(x \in \mathcal I\), the observable count \(Y_i(x) = \sum_{j=1}^{m_i} \mathbf 1\{X_{ij} \le x\}\) follows a binomial mixture distribution, where the mixing distribution is the law of \(F_{\nu_i}(x)\) for $\nu_i \sim \Pi$. In the parametric setting, we assume that this mixing distribution belongs to a family \(G(\cdot;\beta(x))\) with $\beta(x) \in \mathbb R^p$, so that the target is the function-valued parameter
\(x \mapsto \beta_0(x)\).

We first consider a general class of parametric mixture estimators. In this
setting, \(\hat\beta_n(x)\) is obtained by solving estimating equations
pointwise in \(x\), and we analyze the resulting estimator as an element of
\(\ell^\infty(K)^p\) for compact subsets \(K \subset \mathcal I\). We state a
general Z-estimation result giving sufficient conditions for uniform asymptotic
linearity of \(\hat\beta_n\), provide sufficient conditions for these
assumptions to hold, and then apply a delta method argument to obtain
uniform asymptotic linearity of \(G(\alpha;\hat\beta_n(x))\).

We then specialize the general theory to the Beta--Binomial model. We first
define the model and derive the score, Hessian, Fisher information matrix, and
gradient of the Beta CDF. We next verify the conditions needed for uniform
asymptotic linearity of both \(\hat\beta_n\) and
\(G(\alpha;\hat\beta_n(x))\) on \(K\). We also verify the conditions used in
Theorem~\ref{theorem:asymptotic-distribution} for the integrated estimator
\(\int_K G(\alpha;\hat\beta_n(x))\,dx\), which is the object needed for the
asymptotic distribution results in the main text.

The final part of the section collects the supporting technical results and proofs. 
% These include function-analytic results for pointwise operators on \(\ell^\infty(K)^p\), uniform eigenvalue and nonsingularity results needed for the Beta--Binomial Fisher information, bounds on derivatives of the Beta CDF, and entropy calculations used to verify the Glivenko--Cantelli and Donsker conditions.

\subsection{General parametric mixture estimators}

We consider a general semiparametric model in which the mixing distribution at $x \in \mathcal{I}$, denoted by $t \mapsto G(x,t)$, is indexed by a parameter $\beta(x) \in \mathbb{R}^p$. To emphasize this dependence, we also write the mixing distribution as $t \mapsto G(t; \beta(x))$, where we assume the same parametric submodel for the mixing distribution at each $x \in \mathcal{I}$.
The collection of mixing distributions on $\mathcal{I}$ is indexed by the function-valued parameter $\beta \in \ell^{\infty}(\mathcal{I})^p$, with true value $\beta_0 \in \ell^{\infty}(\mathcal{I})^p$.
We study estimation of $\beta_0$ by pointwise parametric maximum likelihood and use Z-estimation theory to analyze $\hat \beta_n$ on compact subsets $K \subset \mathcal{I}$. %In Section \ref{sec:examples} of the main text, we assumed that the regularity conditions hold for $K = \mathcal{I}$.
% We then apply the obtained results to the Beta specification as first considered in Section \ref{sec:structural-restrictions} of the main text. 
% The last part of this appendix contains proofs and technical lemmas used in the first parts of this appendix.

\subsubsection{Estimation of $\beta_0$ and $G(\cdot, \alpha)$}

We consider estimation of $\beta_0(x)$ uniformly on a compact interval $K \subset \mathcal{I}$, where the estimator $\hat \beta_n(x)$ is obtained as the solution of empirical estimating equations. We assume throughout that the estimating equations at $x$ arise from the binomial mixture problem at $x$. Hence, the estimator $\hat \beta_n \in \ell^{\infty}(K)$ is defined through a set of binomial mixture problems.
% We also discuss control of tails when the regularity conditions for uniform asymptotic linearity hold for every compact subset of $\mathcal{I}$ but not for $\mathcal{I}$ itself. 

% \subsubsection{Z-estimation for function-valued parameters on $K \subset \mathcal{I}$}

Let $\left( O_i \right)_{i = 1}^n$ for $O_i := \left\{ m_i, \left( X_{i, j} \right)_{j = 1}^{m_i} \right\}$ be the observed data sampled i.i.d.~from $P$ as described in Section \ref{sec:data-structure} of the main text. 
The observed data used for estimating the mixture distribution at $x$ are $\left\{ Y_i(x), m_i \right\}_{i = 1}^n$ for $Y_i(x) := \sum_{j = 1}^{m_i} I(X_{i, j} \le x)$. 
Let the estimating function for the binomial mixture at $x$ be denoted by $\psi(O_i; \beta(x), x) \in \mathbb{R}^p$, which only depends on $O_i$ through $\left(Y_i(x), m_i\right)$. 
Define the empirical and population estimating functions pointwise as follows:
\begin{align}
    \Psi_n(\beta(x), x) := \mathbb{P}_n \psi(O_i; \beta(x), x), \quad \text{and} \quad \Psi(\beta(x), x) := P \psi( O_i; \beta(x), x).
\end{align}
We further view $x \mapsto \beta_0(x)$ as the function-valued estimand and define the corresponding empirical and true estimating functions as operators from $\ell^{\infty}(K)^p$ into $\ell^{\infty}(K)^p$:
\begin{align*}
    & \Psi_n: \ell^{\infty}(K)^p \to \ell^{\infty}(K)^p, \quad \text{where} \quad \Psi_n(\beta)(x) := \Psi_n(\beta(x), x) \\
    & \Psi: \ell^{\infty}(K)^p \to \ell^{\infty}(K)^p, \quad \text{where} \quad \Psi(\beta)(x) := \Psi(\beta(x), x).
\end{align*}
The space $\ell^{\infty}(K)^p$ is equipped with the norm $\lVert h \rVert_{\infty} := \sup_{x \in K} \lVert h(x) \rVert$ for $\lVert \cdot \rVert$ the Euclidean norm. The estimator $x \mapsto \hat{\beta}_n(x)$, further denoted by $\hat{\beta}_n$, is defined as an estimator satisfying $\lVert \Psi_n(\hat{\beta}_n) \rVert_{\infty} = 0$.
We now reproduce \textcite[Theorem 13.4]{kosorok2008introduction} with some modifications to the current setting, which states the conditions under which $\hat \beta_n$ is uniformly asymptotically linear on $K$.
\begin{theorem}\label{theorem:master-z-estimation}
    Assume that $\Psi(\beta_0) = 0$ and the following hold:
    \begin{enumerate}
        \item[(A)] $\beta \mapsto \Psi(\beta)$ satisfies the following:
        \begin{equation*}
            \lVert \Psi(\beta_n) \rVert_{\infty} \to 0 \quad \text{implies} \quad \lVert \beta_n - \beta_0 \rVert_{\infty} \to 0 \quad \text{for any} \quad (\beta_n)_{n \ge 1}.
        \end{equation*}
        \item[(B)] The class $\left\{ \psi(\cdot; \beta(x), x): \beta \in \ell^{\infty}(K)^p, x \in K \right\}$ is $P$-Glivenko-Cantelli
        \item[(C)] The class $\left\{ \psi(\cdot; \beta(x), x): \lVert \beta - \beta_0 \rVert_{\infty} < \delta, x \in K \right\}$ is $P$-Donsker for some $\delta > 0$.
        \item[(D)] The following holds:
        \begin{equation*}
            \sup_{x \in K} P \lVert \psi(\cdot; \beta_n(x), x) - \psi(\cdot; \beta_0(x), x) \rVert^2 \to 0, \quad \text{as} \quad \beta_n \to \beta_0 \quad \text{in} \quad \ell^{\infty}(K)^p.
        \end{equation*}
        \item[(E)] $\beta \mapsto \Psi(\beta)$ is Fréchet-differentiable at $\beta_0$ with continuously invertible derivative $\dot{\Psi}_{\beta_0}$.
    \end{enumerate}
    Then $n^{1/2} \left( \hat{\beta}_n - \beta_0 \right) = -n^{1/2} \dot{\Psi}_{\beta_0}^{-1} \Psi_n(\beta_0) + o_P(1)$ where $n^{1/2} \left( \Psi_n - \Psi  \right)(\beta_0) \rightsquigarrow Z$ for $Z \in \ell^{\infty}(K)^p$ a tight, mean zero Gaussian limiting distribution.
\end{theorem}

Condition (A) is a standard identifiability condition. 
Conditions (B) and (C) restrict the complexity of the class of estimating functions, and should hold for most parametric models. 
Condition (D) follows from regularity conditions that depend on the parametric model. Condition (E) is a smoothness condition that is closely tied to the uniform asymptotic linearity of $\hat{\beta}_n$.

The proposition below provides sufficient conditions for conditions (A), (D), and (E) to hold. 
\begin{proposition}\label{prop:frechet-derivative-pointwise}
Suppose the following conditions hold:
\begin{enumerate}
    \item[(1)] $\beta(x) \mapsto \psi(O_i; \beta(x), x)$ is twice continuously differentiable for every $x \in K$, with probability one. The first and second-order partial derivatives are denoted by $\dot{\psi}(O_i; \beta(x), x)$ and $\ddot{\psi}(O_i; \beta(x), x)$.
    \item[(2)] For every $x \in K$, there exists an envelope function $M_{x}(O_i)$ with $E[M_{x}(O_i)] < \infty$ such that, for all $\beta(x)$ in a neighborhood of $\beta_0(x)$, we have
    \begin{equation*}
        \lVert \dot{\psi}(O_i; \beta(x), x) \rVert \le M_{x}(O_i), \quad \lVert \ddot{\psi}(O_i; \beta(x), x) \rVert \le M_{x}(O_i).
    \end{equation*}
    \item[(3)] There exists a convex open set $U \subset \mathbb{R}^p$ containing an $\varepsilon$-enlargement of $\left\{ \beta_0(x): x \in K \right\} \subset \mathbb{R}^p$ such that
    \begin{equation*}
        \sup_{x \in K} \sup_{\beta(x) \in U}  P \lVert \dot{\psi}(O_i; \beta(x), x) \rVert^2 < \infty, \qquad  \sup_{x \in K} \sup_{\beta(x) \in U} P \lVert \ddot{\psi}(O_i; \beta(x), x) \rVert < \infty.
    \end{equation*}
    \item[(4)] There exists $c > 0$ such that 
    \begin{equation*}
        \lVert P \dot{\psi}(O_i; \beta_0(x), x) v \rVert \ge c \lVert v \lVert \quad \text{for all $x \in K$ and $v \in \mathbb{R}^p$}.
    \end{equation*}
\end{enumerate}
Then $\beta \mapsto \Psi(\beta)$ satisfies conditions (A), (D), and (E) of Theorem \ref{theorem:master-z-estimation} where the Fréchet derivative and its continuous inverse are defined pointwise as follows:
    \begin{align*}
        (\dot{\Psi}h)(x) & = P \dot{\psi}(O_i; \beta_0(x), x) h(x), \quad \text{for $h \in \ell^{\infty}(K)^p$}, \\
        (\dot{\Psi}^{-1}f)(x) & = (P \dot{\psi}(O_i; \beta_0(x), x))^{-1} f(x), \quad \text{for $f \in \ell^{\infty}(K)^p$}.
    \end{align*}
\end{proposition}
In the following section, we will also verify conditions (B) and (C) for the Beta specification; however, general and directly useful sufficient conditions for (B) and (C) seem to be difficult to obtain.

Finally, the uniform asymptotic linearity of $\hat{\beta}_n$ on $K$, following from Theorem \ref{theorem:master-z-estimation}, implies uniform asymptotic linearity of $G(\alpha; \hat{\beta}_n(\cdot))$ on $K$ under smoothness conditions on the map $\beta(x) \mapsto G(\alpha; \beta(x))$.

\begin{proposition}\label{proposition:smooth-transformation-z-estimator}
Assume the conditions of Theorem \ref{theorem:master-z-estimation} hold.
Further assume that $\beta(x) \mapsto G(\alpha; \beta(x))$ is continuously differentiable on a convex compact set $U \subset \mathbb{R}^p$ containing an $\varepsilon$-enlargement of $\left\{ \beta_0(x): x \in K \right\} \subset \mathbb{R}^p$, with derivative $\dot{G}(\alpha; \beta_0(x))$. Then
\begin{equation*}
    n ^{1/2} \left( G(\alpha; \hat{\beta}_n(x)) - G(\alpha; \beta_0(x))  \right) = -n^{1/2} \dot{G}(\alpha; \beta_0(x)) \dot{\Psi}_{\beta_0}^{-1}\Psi_n(\beta_0)(x)  + r_n(x), \quad \sup_{x \in K} |r_n(x)| = o_P(1).
\end{equation*}
\end{proposition}

\subsection{Beta--Binomial}\label{appendix:beta-binomial}

This subsection applies the general Z-estimation theory from the previous section to the Beta specification.
The structure is as follows:
\begin{enumerate}
    \item \textbf{Preliminaries}: We define the Beta--Binomial model, derive the score vector $\psi$, the Hessian $\dot{\psi}$, and the Fisher information matrix $I(x)$.
    \item \textbf{Asymptotic linearity of $\hat{\beta}_n$ and $G(\alpha; \hat{\beta}_n)$}: We establish uniform asymptotic linearity of $\hat{\beta}_n$ and $G(\alpha; \hat{\beta}_n(x))$ with explicit influence function formulas.
    \item \textbf{Conditions of Theorem~\ref{theorem:asymptotic-distribution}}: We verify the three conditions required for the integrated estimator $\int_K G(\alpha; \hat{\beta}_n(x)) \, dx$ to be asymptotically normal, which are used in the main text for $K = \mathcal I$.
\end{enumerate}

\subsubsection{Preliminaries}

For each $x \in K$, let the mixing distribution be Beta with parameter $\beta(x) := (a_x,b_x)^\top \in (0,\infty)^2$.
The function-valued parameter is therefore $\beta \in \ell^\infty(K)^2$ with true value $\beta_0 \in \ell^\infty(K)^2$ where $\beta_0(x) = (a_{0,x}, b_{0,x})^\top$.
Conditional on $m_i$, the model is $Y_i(x) \sim \mathrm{Beta\text{-}Binomial}(m_i, a_x, b_x)$, with log-likelihood contribution
\begin{equation*}
\ell_i(\beta(x),x)
=
\log B(Y_i(x)+a_x,m_i-Y_i(x)+b_x)
-
\log B(a_x,b_x).
\end{equation*}
The corresponding estimating function is the score function
\begin{equation*}
\psi(O_i;\beta(x),x)
=
\begin{pmatrix}
s_a(O_i;\beta(x),x) \\
s_b(O_i;\beta(x),x)
\end{pmatrix},
\end{equation*}
where
\begin{equation}\label{eq:score-beta-binomial}
\begin{aligned}
s_a(O_i; \beta(x), x) &= \psi_0(Y_i(x)+a_x) - \psi_0(a_x) + \psi_0(a_x+b_x) - \psi_0(m_i+a_x+b_x), \\
s_b(O_i; \beta(x), x) &= \psi_0(m_i-Y_i(x)+b_x) - \psi_0(b_x) + \psi_0(a_x+b_x) - \psi_0(m_i+a_x+b_x),
\end{aligned}
\end{equation}
and $\psi_0$ denotes the digamma function.
The Hessian is
\begin{equation*}
\dot\psi(O_i;\beta(x),x)
=
\begin{pmatrix}
\partial_a s_a & \partial_b s_a \\
\partial_a s_b & \partial_b s_b
\end{pmatrix},
\end{equation*}
where
\begin{equation*}
\begin{aligned}
    \partial_a s_a
    & =
    \psi_1(Y_i(x)+a_x)
    -
    \psi_1(a_x)
    +
    \psi_1(a_x+b_x)
    -
    \psi_1(m_i+a_x+b_x), \\
    \partial_b s_b
    & =
    \psi_1(m_i-Y_i(x)+b_x)
    -
    \psi_1(b_x)
    +
    \psi_1(a_x+b_x)
    -
    \psi_1(m_i+a_x+b_x), \\
    \partial_b s_a
    & =
    \partial_a s_b
    =
    \psi_1(a_x+b_x)
    -
    \psi_1(m_i+a_x+b_x),
\end{aligned}
\end{equation*}
and $\psi_1$ denotes the trigamma function.
The entries of the Fisher information matrix, defined as $I(x) := -P \dot\psi(O_i;\beta_0(x),x)$, are
\begin{equation}\label{eq:fisher-information-beta-mixing}
\begin{split}
    I_{aa}(x)
    & =
    -
    P\left[\psi_1(Y_i(x)+a_{0,x})\right]
    +
    \psi_1(a_{0,x})
    -
    \psi_1(a_{0,x}+b_{0,x})
    +
    P \left[ \psi_1(m_i+a_{0,x}+b_{0,x}) \right], \\
    I_{bb}(x)
    & = 
    -
    P\left[\psi_1(m_i-Y_i(x)+b_{0,x})\right]
    +
    \psi_1(b_{0,x})
    -
    \psi_1(a_{0,x}+b_{0,x})
    +
    P \left[ \psi_1(m_i+a_{0,x}+b_{0,x}) \right], \\
    I_{ab}(x)
    & =
    -
    \psi_1(a_{0,x}+b_{0,x})
    +
    P\left[\psi_1(m_i+a_{0,x}+b_{0,x})\right].
\end{split}
\end{equation}

The Beta CDF with parameters $(a_x, b_x)$ evaluated in $\alpha$ equals $I_{\alpha}(a_x, b_x)$ where $I_{\alpha}(a_x, b_x)$ is the regularized incomplete beta function defined as follows:
\begin{equation*}
    I_{\alpha}(a_x, b_x) := \frac{\operatorname{B}_{\alpha}(a_x, b_x)}{\operatorname{B}(a_x, b_x)}, \quad \text{for} \quad \operatorname{B}_{\alpha}(a_x, b_x) := \int_0^{\alpha} t^{a_x - 1} (1 - t)^{b_x - 1} \, dt.
\end{equation*}
This function is differentiable with respect to $a_x$ and $b_x$ for $a_x, b_x > 0$ and $\alpha \in (0, 1)$. The partial derivatives are given by
\begin{equation}\label{eq:beta-cdf-gradient}
\begin{aligned}
    \partial_{a_x} I_{\alpha}(a_x, b_x) &= \frac{1}{\mathrm{B}(a_x, b_x)} \int_0^\alpha t^{a_x - 1}(1 - t)^{b_x - 1} \log t \, dt - I_\alpha(a_x, b_x) \big( \psi_0(a_x) - \psi_0(a_x + b_x) \big), \\
    \partial_{b_x} I_{\alpha}(a_x, b_x) &= \frac{1}{\mathrm{B}(a_x, b_x)} \int_0^\alpha t^{a_x - 1}(1 - t)^{b_x - 1} \log (1 - t) \, dt - I_{\alpha}(a_x, b_x) \big( \psi_0(b_x) - \psi_0(a_x + b_x) \big).
\end{aligned}
\end{equation}

\subsubsection{Uniform Asymptotic Linearity of $\hat{\beta}_n$ and $G(\alpha; \hat{\beta}_n)$ on $K$}

\paragraph{Uniform Asymptotic Linearity of $\hat{\beta}_n$ on $K$.}

We verify the conditions of Theorem \ref{theorem:master-z-estimation} next. For this, we make the following assumptions on the data generating-process and the true parameter $\beta_0$:
\begin{enumerate}[label=(B\arabic*), ref=B\arabic*, leftmargin=4em]
    \item \label{assumption:compact-pm} $\left\{ (a_{0, x}, b_{0, x}): x \in K  \right\} \subset (0, \infty)^2$ is covered by a compact subset of $(0, \infty)^2$,
    \item \label{assumption:bounded-m} $1 \le m_i < C < \infty$ with probability one and $P(m_i>1) > 0$.
\end{enumerate}
Conditions (A), (D), and (E) of Theorem \ref{theorem:master-z-estimation} follow from Proposition \ref{prop:frechet-derivative-pointwise} where the conditions of this proposition are verified as follows:
\begin{enumerate}
    \item The twice continuous differentiability of $\beta(x) \mapsto \psi(O_i; \beta(x), x)$ follows from the polygamma functions being infinitely differentiable on $(0, \infty)$.
    \item The existence of an envelope function $M_x(O_i)$ with finite expectation follows from the fact that $1 \le m_i < C < \infty$ with probability one and that the digamma and trigamma functions are bounded on all compact subsets of $(0, \infty)$.
    \item The uniform boundedness of $P \lVert \dot{\psi}(O_i; \beta(x), x) \rVert^2$ and $P \lVert \ddot{\psi}(O_i; \beta(x), x) \rVert$ follows from the fact that $1 \le m_i < C < \infty$ with probability one and that the digamma and trigamma functions are bounded on all compact subsets of $(0, \infty)$.
    \item This follows from the uniform nonsingularity of the Fisher information matrix $I(x)$, which itself follows from Lemma \ref{lemma:uniform-nonsingularity-beta}, which is applicable under Assumptions \ref{assumption:compact-pm} and \ref{assumption:bounded-m}.
\end{enumerate}

Conditions (B) and (C) of Theorem \ref{theorem:master-z-estimation} follow from Proposition \ref{proposition:bracketing-entropy-beta} under condition (i). Note that condition (B) only holds by Proposition \ref{proposition:bracketing-entropy-beta} if we restrict the parameter space to a compact subset of $(0, \infty)^2$ that contains $\left\{ (a_{0, x}, b_{0, x}): x \in K  \right\}$.

It now follows from Theorem \ref{theorem:master-z-estimation} that $\hat{\beta}_n$ is uniformly asymptotically linear on $K$ with influence function $I(x)^{-1} \psi(O_i; \beta_0(x), x)$.

\paragraph{Uniform Asymptotic Linearity of $G(\alpha; \hat \beta_n(x))$ on $K$}

The uniform asymptotic linearity of $\hat{\beta}_n$ on $K$ implies the uniform asymptotic linearity of $G(\alpha; \hat{\beta}_n(x))$ on $K$ under the smoothness conditions of Proposition \ref{proposition:smooth-transformation-z-estimator}, which are satisfied in the current setting. Indeed, the map $\beta(x) \mapsto G(\alpha; \beta(x))$ is continuously differentiable on $(0, \infty)^2$, which contains a convex compact set that contains a $\varepsilon$-enlargement of $\left\{ \beta_0(x): x \in K \right\}$.
The influence function of $G(\alpha; \hat{\beta}_n(x))$ is then
\begin{equation}\label{eq:influence-function-beta-G}
    \operatorname{IF}(O_i; x) := \dot{G}(\alpha, \beta_0(x)) I^{-1}(x) \psi(O_i; \beta_0(x), x).
\end{equation} 
% Define the determinant of the Fisher information matrix as $\Delta(x) := I_{aa}(x) I_{bb}(x) - I_{ab}(x)^2$, where $I_{aa}(x)$, $I_{bb}(x)$, and $I_{ab}(x)$ are given in equation \eqref{eq:fisher-information-beta-mixing}. Then the inverse of $I(x)$ is
% \begin{equation*}
%     I(x)^{-1} = \frac{1}{\Delta(x)} \begin{pmatrix} I_{bb}(x) & -I_{ab}(x) \\ -I_{ab}(x) & I_{aa}(x) \end{pmatrix}.
% \end{equation*}
% We then have that 
% \begin{equation*}
%     I(x)^{-1} \begin{pmatrix} s_a(O_i; \beta_0(x), x) \\ s_b(O_i; \beta_0(x), x) \end{pmatrix} = \frac{1}{\Delta(x)} \begin{pmatrix} I_{bb}(x) s_a - I_{ab}(x) s_b \\ -I_{ab}(x) s_a + I_{aa}(x) s_b \end{pmatrix},
% \end{equation*}
% where, for brevity, $s_a := s_a(O_i; \beta_0(x), x)$ and $s_b := s_b(O_i; \beta_0(x), x)$. Then the influence function is
% \begin{equation}\label{eq:influence-function-beta-G}
%     \text{IF}(O_i; x) = \partial_{a_x} I_\alpha(a_{0,x}, b_{0,x}) \cdot \frac{I_{bb}(x) s_a - I_{ab}(x) s_b}{\Delta(x)} + \partial_{b_x} I_\alpha(a_{0,x}, b_{0,x}) \cdot \frac{-I_{ab}(x) s_a + I_{aa}(x) s_b}{\Delta(x)},
% \end{equation}
% where the partial derivatives $\partial_{a_x} I_\alpha(a_x, b_x)$ and $\partial_{b_x} I_\alpha(a_x, b_x)$ are given in equation \eqref{eq:beta-cdf-gradient}, and the score functions are given in equation \eqref{eq:score-beta-binomial}.

\begin{remark}
    Assumption \ref{assumption:compact-pm} requires that the distribution of $F_{\nu}(x)$ for $x$ near the endpoints of $\mathcal I$ is not approaching a point mass at 0 or 1, which implies that the subject-specific distributions $F_{\nu}$ have non-zero mass near the endpoints of $\mathcal I$. This ensures uniform bounds on the scores and guarantees that the Beta-binomial Fisher information is nonsingular uniformly over $\mathcal I$.\footnote{We suspect that this assumption can be relaxed. However, our current strategy of proof requires this assumption. This strategy involves showing uniform asymptotic linearity of $(\hat{a}_{n, x}, \hat{b}_{n, x})$, which in turn implies uniform asymptotic linearity of $I_{\alpha}(\hat a_{n, x}, \hat b_{n, x})$. We suspect that uniform asymptotic linearity may hold for $I_{\alpha}(\hat a_{n, x}, \hat b_{n, x})$, but not for $(\hat a_{n, x}, \hat b_{n, x})$, under relaxed conditions.} 
    If the mixing distribution of the endpoints of $\mathcal{I}$ are point masses at zero and one, then $a_x \to 0$ as $x \to \inf \mathcal{I}$, and $b_x \to 0$ as $x \to \sup \mathcal{I}$. Hence, $\left\{ \beta_0(x): x \in\mathcal{I} \right\}$ is not contained in a compact set. Furthermore, we then have that $\lVert I(x) \rVert^{-1} \to \infty$ as $x$ approaches the endpoints of $\mathcal{I}.$ 
\end{remark}

\begin{remark}
    Assumption \ref{assumption:bounded-m} is reasonable in practically any application with sparse distributional data. We expect that this assumption will change to $P(m_i>p) > 0$ if a parametric binomial mixture submodel with $p + 1$ parameters is used at every $x \in \mathcal I$. 
\end{remark}

\subsubsection{Conditions of Theorem \ref{theorem:asymptotic-distribution}}\label{appendix:conditions-beta-binomial}

We now verify the conditions of Theorem \ref{theorem:asymptotic-distribution} for the Beta specification. In addition to Assumptions \ref{assumption:compact-pm} and \ref{assumption:bounded-m} above, we make the following additional assumption:
\begin{enumerate}[label=(B\arabic*), ref=B\arabic*, leftmargin=4em, start=3]
    \item \label{assumption:lipschitz} There exist $c_1, c_2 > 0$ such that $c_1 |x_1 - x_2| \le |H_{0, \alpha}(x_1) - H_{0, \alpha}(x_2)| \le c_2 |x_1 - x_2|$ for all $x_1, x_2 \in K$, where $H_{0, \alpha}(x) := 1 - G(\alpha; \beta_0(x))$.
\end{enumerate}

To ensure that $\hat H_{n, \alpha}(x) := 1 - G(\alpha; \hat{\beta}_n(x))$ is a valid CDF, we further apply isotonic regression to the initial estimator $x \mapsto 1 - G(\alpha; \hat{\beta}_n(x))$ to obtain a non-decreasing estimator, which we denote by $\hat H_{n, \alpha}^{\text{iso}}(x)$. \textcite{westling2020correcting} show that, if $\hat H_{n, \alpha}(x)$ is uniformly asymptotically linear on $K$ and Assumption \ref{assumption:lipschitz} holds, then $\hat H_{n, \alpha}^{\text{iso}}(x)$ is also uniformly asymptotically linear on $K$ with the same (pointwise) influence function as $\hat H_{n, \alpha}(x)$. In the arguments below, we use $\hat H_{n, \alpha}$ and $H_{n, \alpha}^{\text{iso}}$ interchangeably. 

\paragraph{Condition \ref{cond:AL-linear-expansion}.}

As shown in the previous subsection, $G(\alpha, \hat{\beta}_n(x))$ is uniformly asymptotically linear on $K$ with influence function given in \eqref{eq:influence-function-beta-G}. Because $K$ is assumed to be compact, the integrated remainder is also $o_P(1)$.

Below for condition \ref{cond:AL-integrated-if}, we verify that $\text{IF}(O_i; x)$ is uniformly bounded on $K$ for all $O_i$, which implies that $\mathbb{E}\! \left[  \int_{K} \left| \psi_{\alpha}(O; x)  \right| \right] < \infty$. 
%Hence, $x \mapsto \text{IF}(O_i; x)$ is integrable on $K$ with probability one, and the integrated influence function is well defined. 

\paragraph{Condition \ref{cond:AL-cdf-integrability}.}

By construction, isotonic regression applied to $x \mapsto 1 - G(\alpha; \hat{\beta}_n(x))$ produces a non-decreasing function $\hat H_{n, \alpha}^{\text{iso}}$ on $K$. Without loss of generality,  $\hat H_{n,\alpha}^{\text{iso}}$ can be defined to be right-continuous with limits 0 and 1 at the boundaries of $K$. Since $K$ is compact, the first moment of $\hat{H}_{n,\alpha}^{\text{iso}}$ is finite.
Hence, $\hat{H}_{n,\alpha}^{\text{iso}}$ is a valid CDF (i.e., non-decreasing, right-continuous, with limits 0 and 1 at the boundaries) with finite first moment.

\paragraph{Condition \ref{cond:AL-integrated-if}.}

This condition requires that $\operatorname{Var}(\tilde{\psi}_{\alpha}) \le E[\int_K |\text{IF}(O_i, x)|^2 \, dx]$ is finite, where $\text{IF}(O_i, x)$ is the influence function given in equation \eqref{eq:influence-function-beta-G}.
This influence function is a product of three components: (i) partial derivatives $\partial_{a_x} I_\alpha$ and $\partial_{b_x} I_\alpha$, (ii) entries of $I(x)^{-1}$, and (iii) score functions $s_a$ and $s_b$.
The partial derivatives do not depend on the data and are continuous functions of $\beta_0(x)$. Since $\{ \beta_0(x) : x \in K \}$ is covered by a compact subset of $(0, \infty)^2$, they are bounded on $K$.
The smallest eigenvalue of $I(x)$ is bounded away from zero by Lemma \ref{lemma:uniform-nonsingularity-beta}, so the entries of $I(x)^{-1}$ are bounded on $K$.
The score functions themselves are bounded on $K$ because they only involve digamma functions evaluated at $a_{0,x}$, $Y_i(x) + a_{0,x}$, $b_{0,x}$, $a_{0,x}+b_{0,x}$, $m_i - Y_i(x) + b_{0,x}$, and $m_i + a_{0,x} + b_{0,x}$, which belong to a compact subset of $(0, \infty)$ due to Assumptions \ref{assumption:compact-pm} and \ref{assumption:bounded-m}. Hence, $\text{IF}(O_i, x)$ is uniformly bounded on $K$, and $E[\int_K |\text{IF}(O_i, x)|^2 \, dx]$ is finite.

\begin{remark}
    Assumption \ref{assumption:lipschitz} basically requires that the distribution of subject-specific $\alpha$-quantiles, which has CDF $x \mapsto H_{0,\alpha}(x)$, has a Lebesgue density bounded between $c_1 > 0$ and $c_2 < \infty$. This assumption is required for applying the theoretical results in \textcite{westling2020correcting} which imply that the isotonic regression step does not change uniform asymptotic linearity. This condition can probably be relaxed, especially the lower bound \parencite{westling2020correcting}.
\end{remark}

\subsection{Technical Results and Proofs}

\subsubsection{Function-Analytic Results}

\begin{lemma}[Fréchet differentiability and continuous invertibility of pointwise operators]\label{lemma:frechet-differentiability-superposition}
Let $K$ be an arbitrary index set and define
\begin{equation*}
    \ell^\infty(K)^p
    =
    \Big\{
    \theta : K \to \mathbb{R}^p \; \text{such that} \; \|\theta\|_\infty := \sup_{x\in K} \|\theta(x)\| < \infty
    \Big\}.
\end{equation*}
Let $g : K \times \mathbb{R}^p \to \mathbb{R}^p$ and define the operator
\begin{equation*}
    T : \ell^\infty(K)^p \to \ell^\infty(K)^p, \qquad (T\theta)(x) := g(x,\theta(x)).
\end{equation*}
Fix $\theta_0 \in \ell^\infty(K)^p$ and assume:
\begin{enumerate}
    \item For every $x\in K$, the map $y \mapsto g(x,y)$ is twice continuously differentiable with first and second-order partial derivatives denoted by $\dot{g}(x, y)$ and $\ddot{g}(x, y)$ respectively.
    \item Let $R(\theta_0)$ be the range of $\theta_0$ on $K$. There exists a convex set $U \subset \mathbb{R}^p$ containing an $\varepsilon$-enlargement\footnote{The $\varepsilon$-enlargement of $A$, denoted by $A^{\varepsilon}$, is defined as $A^{\varepsilon} := \left\{ x \in \mathbb{R}^p : \operatorname{dist}(x, A) < \varepsilon \right\}$ where $\operatorname{dist}(x, A) := \inf_{y \in A} \lVert x - y \rVert$. } of $R(\theta_0)$ such that 
    \begin{equation*}
            \sup_{x\in K}\sup_{y\in U} \|\dot{g}(x,y)\| < \infty, \quad
            \sup_{x\in K}\sup_{y\in U} \|\ddot{g}(x,y)\| < \infty.
    \end{equation*}
    \item There exists $c>0$ such that
    \[
    \|\dot{g}(x,\theta_0(x)) v\| \ge c \|v\| \quad \text{for all } x\in K \text{ and } v\in \mathbb{R}^p.
    \]
\end{enumerate}

Then:
\begin{enumerate}
    \item[(i)] $T$ is Fréchet differentiable at $\theta_0$ where the derivative $\dot{T}_{\theta_0}: \ell^{\infty}(K)^p \to \ell^{\infty}(K)^p$ is the following bounded linear operator:
    \begin{equation*}
        (\dot{T}_{\theta_0} h)(x) = \dot{g}(x,\theta_0(x))\, h(x), \qquad h\in \ell^\infty(K)^p.
    \end{equation*}
    \item[(ii)] The derivative $\dot{T}_{\theta_0}$ has a continuous inverse defined pointwise as follows: 
    \begin{equation*}
        (\dot{T}^{-1}_{\theta_0}f)(x) := \dot{g}(x,\theta_0(x))^{-1} f(x), \qquad f \in \ell^{\infty}(K)^p.
    \end{equation*}
\end{enumerate}
\end{lemma}

\begin{proof}
Let $\varepsilon > 0$ be the value corresponding to the $\varepsilon$-enlargement in condition 3. For any $h \in \ell^\infty(K)^p$ with $\lVert h \rVert_{\infty} < \varepsilon$, we then have that $\theta_0(x)+h(x)\in U$ for all $x\in K$.  

\textbf{(Fréchet differentiability)}: Let $j \in \{1, \dots, p\}$ and use subscript $j$ to denote the $j$'th element of $g$ such that $g_j: K \times \mathbb{R}^p \to \mathbb{R}$ and $g = (g_1, \dots, g_p)$. 
For each $x \in K$, by the standard multivariate Taylor expansion and condition 1,
\begin{equation*}
    g_j(x,\theta_0(x)+h(x)) = g_j(x,\theta_0(x)) + \dot{g}_j(x,\theta_0(x)) h(x) + R_j(h)(x),
\end{equation*}
where
\begin{equation*}
    R_j(h)(x) = \frac12 h(x)^\top \ddot{g}_j(x,\tilde\theta(x)) h(x),
\end{equation*}
for some $\tilde\theta(x)$ on the line segment between $\theta_0(x)$ and $\theta_0(x)+h(x)$. Because $U$ is convex and $\theta_0(x), \theta_0(x) + h(x) \in U$, the point $\tilde\theta(x)$ also belongs to $U$ for all $x \in K$. 

By uniform boundedness of the Hessian (condition 2), we have that 
\begin{equation*}
    \|R_j(h)\|_\infty := \sup_{x\in K} \left| R_j(h)(x) \right| \le \lVert h \rVert_{\infty}^2 \cdot \sup_{x\in K}\sup_{y\in U} \|\ddot{g}_j(x,y)\| < \infty.
\end{equation*}
This holds for any $j \in \{1, \dots, p\}$ and thus implies that $\|R(h)\|_\infty < \infty$.
Letting $\dot{T}_{\theta_0}(h)(x) = \dot{g}(x,\theta_0(x))\, h(x)$, this thus shows that 
\begin{equation*}
    \frac{\lVert T(\theta_0 + h) - T(\theta_0) - \dot{T}_{\theta_0}(h)  \rVert_{\infty}}{\lVert h \rVert_{\infty}} = \frac{\lVert R_{\theta_0}(h) \rVert_{\infty}}{\lVert h \rVert_{\infty}} =  O(\lVert h \rVert_{\infty}) = o(1) \quad \text{as} \quad \lVert h \rVert_{\infty} \to 0.
\end{equation*}
Linearity and boundedness of $\dot{T}$ follows directly as follows, for any $h_1, h_2 \in \ell^\infty(K)^p$ and any $x \in K$:
\begin{align*}
    & \dot{T}_{\theta_0}(h_1 + h_2)(x) := \dot{g}(x,\theta_0(x))\, \left( h_1(x) + h_2(x) \right) = \dot{T}_{\theta_0}(h_1)(x) + \dot{T}_{\theta_0}(h_2)(x), \\
    & \lVert \dot{T}_{\theta_0}(h) \rVert_{\infty} = \sup_{x \in K} \lVert \dot{g}(x,\theta_0(x)) h(x) \rVert \le \sup_{x\in K} \|\dot{g}(x,\theta_0(x))\| \cdot \lVert h \rVert_{\infty} < \infty.
\end{align*}
This proves Fréchet differentiability.

\textbf{(Continuous invertibility)}: 
The inverse of $\dot{T}_{\theta_0}$ is defined pointwise as follows $\dot{T}^{-1}(f)(x) := \dot{g}(x,\theta_0(x))^{-1} f(x)$ where the matrix inverse exists by condition 3.
Also, by condition 3, we have for any $h \in \ell^\infty(K)^p$ and some $c > 0$ that,
\begin{equation*}
    \|\dot{T}_{\theta_0}(h)\|_\infty := \sup_{x\in K} \| \dot{g}(x,\theta_0(x)) h(x)\| \ge c \sup_{x\in K} \|h(x)\| = c \|h\|_\infty,
\end{equation*}
which shows that $\dot{T}_{\theta_0}$ is invertible.
\end{proof}

\begin{corollary}\label{corollary:identifiability-z-estimator-frechet}
    Under the assumptions of Lemma \ref{lemma:frechet-differentiability-superposition} and for $\theta_0 \in \ell^{\infty}(K)^p$ such that $\lVert T\theta_0 \rVert_{\infty} = 0$, $T$ satisfies the following:
    \begin{equation*}
        \lVert T \theta_n \rVert_{\infty} \to 0 \quad \text{implies} \quad \lVert \theta_n - \theta_0 \rVert_{\infty} \to 0 \quad \text{for any} \quad (\theta_n)_{n \ge 1}.
    \end{equation*}
\end{corollary}
\begin{proof}
    Let $(\theta_n)_{n \ge 1}$ be any sequence such that $\lVert T \theta_n \rVert_{\infty} \to 0$. By Lemma \ref{lemma:frechet-differentiability-superposition}, $T$ is Fréchet differentiable at $\theta_0$ with derivative $\dot{T}_{\theta_0}$. 
    \begin{align*}
        \lVert \dot{T}_{\theta_0} (\theta_n - \theta_0) \rVert_{\infty} & \le \lVert \dot{T}_{\theta_0} (\theta_n - \theta_0) - T(\theta_n) \rVert_{\infty} + \lVert T(\theta_n) \rVert_{\infty} \\
        & \le \lVert T(\theta_n) - T(\theta_0) - \dot{T}_{\theta_0} (\theta_n - \theta_0) \rVert_{\infty} + \lVert T(\theta_n) \rVert_{\infty} \\ 
        & \le o \left( \lVert \theta_n - \theta_0 \rVert_{\infty} \right) + o(1) 
    \end{align*}
    By continuous invertibility (see last part of the proof of Lemma \ref{lemma:frechet-differentiability-superposition}), we have that $c \lVert \theta_n - \theta_0 \rVert_{\infty} \le  \lVert \dot{T}_{\theta_0} (\theta_n - \theta_0) \rVert_{\infty}$ for some $c > 0$. Together with the previous display, this gives the following:
    \begin{equation*}
        \lVert \theta_n - \theta_0 \rVert_{\infty} \le c^{-1} \lVert \dot{T}_{\theta_0} (\theta_n - \theta_0) \rVert_{\infty} \le o \left( \lVert \theta_n - \theta_0 \rVert_{\infty} \right) + o(1).
    \end{equation*}
    This implies that $\lVert \theta_n - \theta_0 \rVert_{\infty} = o(1)$.
\end{proof}

\begin{lemma}\label{lemma:uniform-min-eigenvalue}
    Let $I(\alpha)$ be a positive definite matrix for each $\alpha \in A$, where $A$ is compact, and assume that the entries of $I(\alpha)$ are continuous in $\alpha$ on $A$. Then there exists a constant $c > 0$ such that $\inf_{\alpha \in A}\lambda_{\text{min}}(I(\alpha)) \ge c$, where $\lambda_{\text{min}}(I(\alpha))$ is the smallest eigenvalue of $I(\alpha)$.
\end{lemma}
\begin{proof}
    Define the function
    \begin{equation*}
        f: A \to \mathbb{R}, \quad f(\alpha) := \lambda_{\text{min}}(I(\alpha)).
    \end{equation*}
    Since the entries of $I(\alpha)$ are continuous in $\alpha$, and the eigenvalues of a matrix depend continuously on the matrix entries, $f$ is continuous on $A$. A continuous real-valued function on a compact set attains its minimum; therefore, there exists some $\alpha_0 \in A$ such that
    \begin{equation*}
        \inf_{\alpha \in A} \lambda_{\text{min}} (I(\alpha)) = \min_{\alpha \in A} \lambda_{\text{min}} (I(\alpha)) = \lambda_{\text{min}} (I(\alpha_0)).
    \end{equation*}
    Since $I(\alpha)$ is positive definite for all $\alpha \in A$, we have that $\lambda_{\text{min}} (I(\alpha_0)) > 0$. Set $c := \lambda_{\text{min}} (I(\alpha_0))$. Hence, $\inf_{\alpha \in A} \lambda_{\text{min}} (I(\alpha)) = c > 0$.
\end{proof}

\subsubsection{Beta-binomial}\label{sec:beta-binomial-technical-results}

\begin{lemma}\label{lemma:uniform-nonsingularity-beta}
Let $I(x)$ be the Fisher information of the beta-binomial mixing distribution for $x\in K$ with the entries defined in \eqref{eq:fisher-information-beta-mixing}.
Assume that $\left\{(a_{0, x}, b_{0, x}): x \in K \right\}$ belongs to a compact subset of $(0, \infty)^2$, that $1 \le m_i < C < \infty$ with probability one, and that $P(m_i>1)>0$.
Then there exists $c > 0$ such that $\lambda_{\min}(I(x))\ge c$ for all $x \in K$.
% In particular, $I(x)$ is uniformly invertible on $\mathcal K$, and the Jacobian is uniformly nonsingular in the sense required for Theorem~\ref{theorem:master-z-estimation}.
\end{lemma}
\begin{proof}
Let $A$ be the compact subset containing $\left\{(a_{0, x}, b_{0, x}): x \in K \right\}$. We now use Lemma \ref{lemma:uniform-min-eigenvalue} where we have to show that $I(\alpha, \beta)$, defined as the Fisher information matrix for the beta-binomial mixture with parameters $(\alpha, \beta)$, is positive definite for all $(\alpha, \beta) \in A$ and its entries are continuous in $(\alpha, \beta)$.
We further denote the entries of the Fisher information matrix at $(\alpha, \beta)$ by $I_{\alpha, \alpha}$, $I_{\beta, \beta}$, and $I_{\alpha, \beta}$.

By standard likelihood theory, we have that $I_{\alpha, \alpha} = \mathbb{E}_{(\alpha, \beta)}\! \left[ s_{\alpha}^2 \right]$, $I_{\beta, \beta} = \mathbb{E}_{(\alpha, \beta)}\! \left[ s_{\beta}^2 \right]$, and $I_{\alpha, \beta} = \mathbb{E}_{(\alpha, \beta)}\! \left[ s_{\alpha} s_{\beta} \right]$ where the scores are defined in \eqref{eq:score-beta-binomial} (replacing $(a, b)$ with $(\alpha, \beta)$). Since $I(\alpha, \beta)$ is a covariance matrix, it is positive semidefinite. It is positive definite if and only if the determinant is positive, which is equivalent to $I_{\alpha, \alpha} I_{\beta, \beta} - I_{\alpha, \beta}^2 > 0$.
Because the digamma function is strictly increasing on $(0, \infty)$, we have that $\mathbb{E}_{(\alpha, \beta)} \left[ s_{\alpha}^2 \right] = 0$ only if the distribution of $Y$ is degenerate. The distribution of $Y$ is not degenerate for $(\alpha, \beta) \in A \subset (0, \infty)^2$ and $m_i \ge 1$, and similarly $\mathbb{E}_{(\alpha, \beta)} \left[ s_{\beta}^2 \right] > 0$.
Moreover, by Cauchy--Schwarz we always have $I_{\alpha, \beta}^2 \le I_{\alpha, \alpha} I_{\beta, \beta}$, with equality if and only if $s_{\beta}=c\,s_{\alpha}$ almost surely for some constant $c$.

If $P(m_i=1)=1$, then $Y\in\{0,1\}$ and the beta-binomial model reduces to a Bernoulli model with success probability $\alpha/(\alpha+\beta)$, so the score components are colinear and $I(\alpha,\beta)$ is singular.

Now assume $P(m_i>1)>0$. Since $m_i$ is integer-valued and $1 \le m_i < C$ almost surely, there exists $m_* \ge 2$ such that $P(m_i=m_*)>0$. Suppose, toward a contradiction, that $s_{\beta}=c\,s_{\alpha}$ almost surely for some constant $c$ and $(\alpha, \beta) \in (0, \infty)^2$. Then this relation also holds conditional on $m_i=m_*$. Under the beta-binomial model with $(\alpha,\beta)\in(0,\infty)^2$, every value $y\in\{0,\dots,m_*\}$ has strictly positive conditional probability, so we must have
\[
s_{\beta}(y;m_*,\alpha,\beta)=c\,s_{\alpha}(y;m_*,\alpha,\beta), \qquad y=0,\dots,m_*.
\]
Taking first differences in $y$ gives, for $y=0,\dots,m_*-1$,
\[
s_{\alpha}(y+1)-s_{\alpha}(y)=\psi_0(y+1+\alpha)-\psi_0(y+\alpha)=\frac{1}{y+\alpha},
\]
\[
s_{\beta}(y+1)-s_{\beta}(y)=\psi_0(m_*-y-1+\beta)-\psi_0(m_*-y+\beta)= -\frac{1}{m_*-y-1+\beta}.
\]
Hence
\begin{align*}
    c = c \frac{s_{\alpha}(y + 1) - s_{\alpha}(y)}{s_{\alpha}(y + 1) - s_{\alpha}(y)} = \frac{s_{\beta}(y + 1) - s_{\beta}(y)}{s_{\alpha}(y + 1) - s_{\alpha}(y)} & = -\frac{y + \alpha}{m_*-y-1+\beta} \quad \text{for} \quad y=0,\dots,m_*-1.
\end{align*}
The right-hand side depends on $y$ when $m_*\ge 2$, which is impossible for a constant $c$. This contradiction shows that equality in Cauchy--Schwarz cannot occur, so
\[
I_{\alpha, \beta}^2 < I_{\alpha, \alpha} I_{\beta, \beta},
\]
and therefore $I(\alpha, \beta)$ is positive definite for all $(\alpha, \beta) \in (0, \infty)^2$.

For continuity, note that $m_i \in \{1,\dots,\lceil C \rceil-1\}$ almost surely and, conditional on $m_i=m$, $Y$ takes values in the finite set $\{0,\dots,m\}$. 
Since expectations $E_{(\alpha, \beta)}[\cdot]$ are finite weighted sums, we have that, for $(\alpha_n, \beta_n) \to (\alpha, \beta) \in (0, \infty)^2$,
\begin{equation*}
    I_{\alpha_n, \alpha_n} = \mathbb{E}_{(\alpha_n, \beta_n)} \! \left[ s_{\alpha_n}^2 \right]  \to I_{\alpha, \alpha} = \mathbb{E}_{(\alpha, \beta)} \! \left[ s_{\alpha}^2 \right] \quad \text{as} \quad (\alpha_n, \beta_n) \to (\alpha, \beta),
\end{equation*}
if for each fixed $m$ and $y$ (i) the weights $P_{(\alpha_n, \beta_n)}(Y=y \mid m_i=m)$ converge to $P_{(\alpha, \beta)}(Y=y \mid m_i=m)$ and (ii) $s_{\alpha}$ is continuous in $(\alpha, \beta)$ for each fixed $m$ and $y$. This is indeed the case because (i) the beta-binomial pmf is continuous in $(\alpha, \beta)$ for each fixed $m$ and $y$ and (ii) $s_{\alpha}$ is continuous in $(\alpha, \beta)$ for each fixed $m$ and $y$ because the digamma function is continuous on $(0, \infty)$. 
The same arguments apply to $I_{\alpha, \beta}$ and $I_{\beta, \beta}$.
\end{proof}

\begin{proposition}\label{proposition:donsker-pointwise-mixture}
    Consider the data-generating mechanism described in Section \ref{sec:data-structure} where the observed data are $O := (m, X_1, \dots, X_{m}) \sim P$. Assume that $1 \le m < C < \infty$ with probability one. Then the function class 
    \begin{align*}
        \mathcal F& := \left\{ \sum_{j = 1}^{m} \mathbf{1} \left( X_j \le x \right) : x \in \mathbb{R} \right\} 
    \end{align*}
    is $P$-Donsker.
\end{proposition}

\begin{proof}
    Since $m \le C$ almost surely, we may embed observations into the fixed-dimensional space
    \[
        O = (m, X_1, \dots, X_{\lfloor C \rfloor}),
    \]
    where we define $X_j = 0$ for $j > m$. Then the class $\mathcal{F}$ can be written as
    \[
        \mathcal{F}
        = \left\{ \sum_{j=1}^{\lfloor C \rfloor} \mathbf{1}(X_j \le x,\, j \le m) : x \in \mathbb{R} \right\}.
    \]

    For each fixed $j \in \{1, \dots, \lfloor C \rfloor\}$, define
    \[
        \mathcal{G}_j
        := \left\{ \mathbf{1}(X_j \le x,\, j \le m) : x \in \mathbb{R} \right\}.
    \]
    We can write
    \[
        \mathbf{1}(X_j \le x,\, j \le m)
        = \mathbf{1}(X_j \le x)\,\mathbf{1}(j \le m).
    \]
    Since $\mathbf{1}(j \le m)$ is a fixed measurable function and the class
    \[
        \left\{ \mathbf{1}(X_j \le x) : x \in \mathbb{R} \right\}
    \]
    is a VC class (with VC index $2$), it follows that $\mathcal{G}_j$ is a VC-subgraph class \parencite[Lemma 9.9]{kosorok2008introduction}.

    The class $\mathcal{F}$ is the finite sum of $\lfloor C \rfloor$ such classes:
    \[
        \mathcal{F} = \left\{ \sum_{j=1}^{\lfloor C \rfloor} g_j : g_j \in \mathcal{G}_j \right\}.
    \]
    Since the finite sum of $P$-Donsker classes is $P$-Donsker \parencite[Corollary 9.32]{kosorok2008introduction}, it suffices to verify that each $\mathcal{G}_j$ is $P$-Donsker. 

    By \textcite[Theorem 9.3]{kosorok2008introduction}, $\mathcal{G}_j$ satisfies the uniform entropy condition of \textcite[Theorem 8.19]{kosorok2008introduction}. 
    To show that $\mathcal{G}_j$ is $P$-Donsker, it remains to verify that $\mathcal{G}_j$ admits a square-integrable envelope and is $P$-measurable.

    For all $g \in \mathcal{G}_j$,
    \[
        0 \le g(O) \le 1,
    \]
    so $\mathcal{G}_j$ admits the square-integrable envelope $G_j(O) := 1$.

    It remains to verify $P$-measurability. We show that $\mathcal{G}_j$ is pointwise measurable. Let $\mathcal{G}_{j, 0} \subset \mathcal{G}_j$ be the subclass obtained by restricting $x$ to $\mathbb{Q}$. Then $\mathcal{G}_{j, 0}$ is countable. For any $x_0 \in \mathbb{R}$ and corresponding $g \in \mathcal{G}_j$, choose a sequence $(q_n)_{n \ge 1} \subset \mathbb{Q}$ such that $q_n \downarrow x_0$. Let $g_n \in \mathcal{G}_{j, 0}$ denote the corresponding functions. Then, for every $O$,
    \[
        g_n(O)
        = \mathbf{1}(X_j \le q_n) \cdot \mathbf{1}(j \le m)
        \to
        \mathbf{1}(X_j \le x_0) \cdot \mathbf{1}(j \le m)
        = g(O).
    \]
    Hence, $\mathcal{G}_j$ is pointwise measurable and therefore $P$-measurable.

    We conclude that $\mathcal{G}_j$ satisfies the uniform entropy condition, admits a square-integrable envelope, and is $P$-measurable. Therefore, $\mathcal{G}_j$ is $P$-Donsker by \textcite[Theorem 8.19]{kosorok2008introduction}.
\end{proof}

\begin{corollary}\label{corollary:donsker-score-components}
    Consider the data-generating mechanism described in Section \ref{sec:data-structure} where the observed data are $O := (m, X_1, \dots, X_{m}) \sim P$. Assume that $1 \le m < C < \infty$ with probability one. Then the classes of functions $\psi_0 \circ \mathcal{F}_1$, $\psi_0 \circ \mathcal{F}_2$, and $\psi_0 \circ \mathcal{F}_3$ for 
    \begin{align*}
        \mathcal{F}_1 & := \left\{ \sum_{j=1}^{m} \mathbf{1}(X_j \le x) + \theta : x \in \mathbb{R}, \theta \in [l, u] \subset (0, \infty) \right\}, \\
        \mathcal{F}_2 & := \left\{ m + \theta : \theta \in [l, 2u] \subset (0, \infty) \right\}, \\
        \mathcal{F}_3 & := \left\{ m - \sum_{j=1}^{m} \mathbf{1}(X_j \le x) + \theta : x \in \mathbb{R}, \theta \in [l, u] \subset (0, \infty) \right\},
    \end{align*}
    are $P$-Donsker.
\end{corollary}
\begin{proof}
    Since $\psi_0$ is continuously differentiable on $(0, \infty)$, it is Lipschitz on $[l, 2u + C]$. Hence, $\psi_0 \circ \mathcal{F}$ is the Lipschitz transformation of $\mathcal{F}$. If $\mathcal{F}$ is a $P$-Donsker class taking values in $[l, 2u + C]$, it follows that $\psi_0 \circ \mathcal{F}$ is also $P$-Donsker by \textcite[Theorem 9.31]{kosorok2008introduction}.

    The class $\mathcal{F}_1$ is the sum of $\left\{ \sum_{j=1}^{m} \mathbf{1}(X_j \le x) : x \in \mathbb{R} \right\}$ and a class of constant functions $\left\{ \theta : \theta \in [l, u] \right\}$. The first class is $P$-Donsker by Proposition \ref{proposition:donsker-pointwise-mixture}, and the second class is trivially $P$-Donsker. Hence, $\mathcal{F}_1$ is $P$-Donsker. The class $\mathcal{F}_2$ is trivially $P$-Donsker, and $\mathcal{F}_3$ is $P$-Donsker because it is the sum of the $P$-Donsker class $\mathcal{F}_1$ and the fixed function $m$.

    The three classes $\mathcal{F}_1$, $\mathcal{F}_2$, and $\mathcal{F}_3$ take values in $[l, 2u + C]$. Hence, $\psi_0 \circ \mathcal{F}_1$, $\psi_0 \circ \mathcal{F}_2$, and $\psi_0 \circ \mathcal{F}_3$ are $P$-Donsker.
\end{proof}

\begin{proposition}\label{proposition:bracketing-entropy-beta} 
    Consider the data-generating mechanism described in Section \ref{sec:data-structure} where the observed data are $O := (m, X_1, \dots, X_{m}) \sim P$. Assume that $1 \le m < C < \infty$ with probability one. 
    Let $A \subset (0, \infty)^2$ be compact and define $s_{\alpha}$ and $s_{\beta}$ as in \eqref{eq:score-beta-binomial}.
    Then, the classes of functions $\mathcal{F}_\alpha := \left\{ s_{\alpha}(\cdot; \alpha, \beta): (\alpha, \beta) \in A  \right\}$ and $\mathcal{F}_\beta := \left\{ s_{\beta}(\cdot; \alpha, \beta): (\alpha, \beta) \in A  \right\}$ are $P$-Donsker.
\end{proposition}
\begin{proof}
    Without loss of generality, we assume that $A = [l, u]^2$ where $l > 0$ and $u < \infty$.
    From \eqref{eq:score-beta-binomial} follows that the classes $\mathcal{F}_\alpha$ and $\mathcal{F}_\beta$ are contained in the sum of the following classes of functions:
    \begin{align*}
        \mathcal{F}_1 & := \left\{ \psi_0\left( \sum_{j=1}^m \mathbf{1}(X_j \le x) + \theta \right): \theta \in [l, u] \right\} \\
        \mathcal{F}_2 & := \left\{ \psi_0( m + \theta): \theta \in [l, 2u] \right\} \\
        \mathcal{F}_3 & := \left\{ \psi_0( m - \sum_{j=1}^m \mathbf{1}(X_j \le x) + \theta): \theta \in [l, u] \right\} \\
        \mathcal{F}_4 & := \left\{ \theta: \theta \in [\psi_0(l), \psi_0(2 u)] \right\},
    \end{align*}
    where $\psi_0$ is the digamma function.
    The classes $\mathcal{F}_1$, $\mathcal{F}_2$, and $\mathcal{F}_3$ are $P$-Donsker by Corollary \ref{corollary:donsker-score-components}.
    The class $\mathcal{F}_4$ is trivially $P$-Donsker. 
\end{proof}

\subsubsection{Proof of Proposition \ref{prop:frechet-derivative-pointwise}}

\begin{proof}
    Let $P$ denote the distribution of $O_i$ and let $Pf$ denote the expectation of $f(O_i)$ when $O_i \sim P$. We verify that conditions (E), (A), and (D) of Theorem \ref{theorem:master-z-estimation} hold.
    
    \textbf{Conditions (E) and (A).}~~
    By dominated convergence (which holds by condition (2)), we have that the first and second-order partial derivatives of $\beta(x) \mapsto \Psi(\beta(x), x)$ at $\beta_0(x)$ are given by $\dot \Psi(\beta(x), x) = P \dot{\psi}(O_i; \beta(x), x)$ and $\ddot{\Psi}(\beta(x), x) = P \ddot{\psi}(O_i; \beta(x), x)$, respectively. We then have by condition (3) that
    \begin{equation*}
        \sup_{x \in K} \sup_{\beta(x) \in U} \lVert \dot{\Psi}(\beta(x), x) \rVert < \infty, \quad \sup_{x \in K} \sup_{\beta(x) \in U} \lVert \ddot{\Psi}(\beta(x), x) \rVert < \infty.
    \end{equation*}
    The above display together with condition (4) is sufficient to apply Lemma \ref{lemma:frechet-differentiability-superposition}, which implies that $\beta \mapsto \Psi(\beta)$ is Fréchet differentiable at $\beta_0$ with Fréchet derivative $\dot{\Psi}_{\beta_0}: \ell^{\infty}(K)^p \to \ell^{\infty}(K)^p$ with a continuous inverse $\dot{\Psi}_{\beta_0}^{-1}: \ell^{\infty}(K)^p \to \ell^{\infty}(K)^p$ defined pointwise below, thus satisfying condition (E) of Theorem \ref{theorem:master-z-estimation}.
    \begin{align*}
        (\dot{\Psi}_{\beta_0}h)(x) & = P \dot{\psi}(O_i; \beta_0(x), x) h(x), \quad \text{for $h \in \ell^{\infty}(K)^p$}, \\
        (\dot{\Psi}_{\beta_0}^{-1}f)(x) & = (P \dot{\psi}(O_i; \beta_0(x), x))^{-1} f(x), \quad \text{for $f \in \ell^{\infty}(K)^p$}.
    \end{align*}
    Condition (A) then follows from Corollary \ref{corollary:identifiability-z-estimator-frechet}.

    \textbf{Condition (D).}~~
    Let $\beta_n \to \beta_0$ in $\ell^{\infty}(K)^p$. By condition (3), for all sufficiently large $n$ and all $x \in K$, both $\beta_n(x)$ and $\beta_0(x)$ belong to $U$. 
    By the mean-value theorem, for each component $j \in \{ 1, \ldots, p \}$ and each $x \in K$ there exists $\tilde\beta_{n,j}(x)$ on the segment between $\beta_n(x)$ and $\beta_0(x)$ such that
    \begin{equation*}
    \psi_j(O_i;\beta_n(x),x)-\psi_j(O_i;\beta_0(x),x)
    =
    \dot{\psi}_j(O_i;\tilde\beta_{n,j}(x),x)\{\beta_n(x)-\beta_0(x)\}.
    \end{equation*}
    Because $U$ is convex, each $\tilde\beta_{n,j}(x)$ also belongs to $U$.
    Hence,
    \begin{equation*}
    \sup_{x\in K} P\left\lVert \psi(O_i;\beta_n(x),x)-\psi(O_i;\beta_0(x),x)\right\rVert^2
    \le
    \left(\sup_{x\in K}\sup_{\beta(x)\in U} P\left\lVert \dot{\psi}(O_i;\beta(x),x)\right\rVert^2\right)\lVert \beta_n-\beta_0\rVert_\infty^2.
    \end{equation*}
    The first term of the right-hand side is finite by condition (3), and the second term converges to zero by assumption. Hence, the left-hand side converges to zero, which implies that condition (D) of Theorem \ref{theorem:master-z-estimation} holds.
\end{proof}

\subsubsection{Proof of Proposition \ref{proposition:smooth-transformation-z-estimator}}

\begin{proof}
Because $\hat{\beta}_n$ is uniformly consistent for $\beta_0$ on $K$, we have that $\hat{\beta}_n(x) \in U$ with probability approaching one uniformly on $K$. We further assume that $\hat{\beta}_n(x) \in U$ for all $x \in K$. 
Hence, we can apply the mean-value expansion to $G(\alpha; \hat{\beta}_n(x)) - G(\alpha; \beta_0(x))$, for $\tilde{\beta}_n(x)$ between $\hat{\beta}_n(x)$ and $\beta_0(x)$.
\begin{align*}
    n ^{1/2} \left( G(\alpha; \hat{\beta}_n(x)) - G(\alpha; \beta_0(x))  \right) & =  n^{1/2} \dot{G}(\alpha; \tilde{\beta}_n(x)) \left( \hat{\beta}_n(x) - \beta_0(x) \right) \\
    & = n^{1/2} \dot{G}(\alpha; \beta_0(x)) \left( \hat \beta_n(x) -  \beta_0(x) \right) \\
    & \quad +  n^{1/2} \left( \dot{G}(\alpha; \tilde{\beta}_n(x)) - \dot{G}(\alpha; \beta_0(x)) \right) \left( \hat \beta_n(x) -  \beta_0(x) \right)
\end{align*}
Note that $n^{1/2} \lVert \hat \beta_n - \beta_0 \rVert_{\infty} = O_P(1)$ by Theorem \ref{theorem:master-z-estimation}. We also have that $\lVert \dot{G}(\alpha; \tilde{\beta}_n(x)) - \dot{G}(\alpha; \beta_0(x)) \rVert_{\infty} = o_P(1)$ because $\dot{G}(\alpha; \beta(x))$ is uniformly continuous in $\beta(x)$ on $U$ by compactness of $U$ and continuity of $\dot{G}(\alpha; \cdot)$ on $U$. Hence, the second term of the above display is uniformly $o_P(1)$.

Moreover, $\sup_{\beta(x) \in U}\dot{G}(\alpha; \beta(x)) < \infty$ since $\dot{G}(\alpha; \beta(x))$ is continuous on the compact set $U$. The first term of the last display is $-n^{1/2} \dot{G}(\alpha; \beta_0(x)) \dot{\Psi}_{\beta_0}^{-1}\Psi_n(\beta_0)(x) + \dot{G}(\alpha; \beta_0(x))r_n(x)$ with $\sup_{x \in K} |r_n(x)| = o_P(1)$ by Theorem \ref{theorem:master-z-estimation}. Hence, $\sup_{x \in K} \lVert \dot{G}(\alpha; \beta_0(x))r_n(x) \rVert = o_P(1)$, which gives the required uniform linearization.
\end{proof}

\end{document}